\documentclass[10pt]{article}
\usepackage{fullpage,graphicx,subfigure,mathdots,mathpazo,color}
\usepackage{amsmath,amscd,tikz,mathrsfs}
\usepackage[normalem]{ulem}
\usepackage{amsmath}
\usepackage{setspace,booktabs}

\usepackage{epsfig,amsmath,graphicx,amssymb,overpic}

\usepackage{graphicx}
\usepackage{dcolumn}
\usepackage{bm}
\usepackage{graphicx}
\usepackage{subfigure}
\usepackage{epsfig,amsmath,graphicx,amssymb,overpic,cite}

\usepackage{graphicx}
\usepackage{dcolumn}
\usepackage{bm}
\usepackage{graphicx}
\usepackage{subfigure}
\usepackage{amsthm}

\newcommand{\ee}{\mathrm{e}}

\usepackage[T1]{fontenc}
\usepackage{amsmath, amsfonts, amssymb, amsthm}
\usepackage{geometry} 
\usepackage{hyperref} \hypersetup{colorlinks=true, linkcolor=blue, citecolor=red}
\usepackage{enumitem}

\usepackage{graphicx}
\usetikzlibrary{positioning,arrows.meta}

\usepackage{graphicx}      
\usepackage{titlesec}
\def\bee{\begin{eqnarray}}
\def\ene{\end{eqnarray}}
\def\bes{\begin{subequations}}
\def\ees{\end{subequations}}

\usepackage{color}

\def\i{\mathrm{i}}
\def\v{\vspace{0.06in}}
\def\be{\begin{equation}}
\def\ee{\end{equation}}
\def\bee{\begin{eqnarray}}
\def\ene{\end{eqnarray}}
\def\bes{\begin{subequations}}
\def\ees{\end{subequations}}

\def\d{\displaystyle}

\theoremstyle{plain}
\newtheorem{theorem}{Theorem}[section]

\newtheorem{proposition}[theorem]{Proposition}
\newtheorem{corollary}[theorem]{Corollary}
\newtheorem{RH}[theorem]{Riemann-Hilbert Problem}

\theoremstyle{definition}

\newtheorem{remark}[theorem]{Remark}

\numberwithin{equation}{section}

\begin{document}

\baselineskip=15pt
\renewcommand {\thefootnote}{\dag}
\renewcommand {\thefootnote}{\ddag}
\renewcommand {\thefootnote}{ }

\pagestyle{plain}

\begin{center}
\baselineskip=16pt \leftline{} \vspace{-0.1in} {\Large \bf On the global well-posedness for the nonlocal Fokas-Lenells equation  with the weighted Sobolev initial data on the line} \\[0.2in]
\end{center}

\begin{center}
{Zhenya Yan$^{a,b,c,*}$,\,  Guoqiang Zhang$^{b,c}$,\, Guancheng Zhu$^{b,c}$,}
\footnote{$^{*}${\it Email address}: zyyan@mmrc.iss.ac.cn (Corresponding author)}  \\[0.05in]
\baselineskip=12pt
{\footnotesize $^a${\it School of Mathematics and Information Science, Zhongyuan University of Technology, Zhengzhou 450007, China}} \\
{\footnotesize
$^b${\it State Key Laboratory of Mathematical Sciences, Academy of Mathematics and Systems Science,\\ Chinese Academy of Sciences, Beijing 100190, China}}\\
{\footnotesize $^c${\it School of Mathematical Sciences, University of Chinese Academy of Sciences, Beijing 100049, China}}
\\
\end{center}

\v

\noindent\rule{\textwidth}{0.6pt}

\noindent {\bf Abstract}

 We establish the global well-posedness of the Cauchy problem for the reverse space-time nonlocal Fokas-Lenells equation with the weighted Sobolev initial data $q_0(x)\in H^{3}(\mathbb{R}) \cap H^{2,1}(\mathbb{R})$ on the line. We develop the inverse scattering transform formulated via the associated Riemann-Hilbert problems to study this issue. A spectral uniformization transform is introduced to resolve the singular behavior inherent in the KN-type negative flow spectral problem. Owing to the reverse space-time reduction, reflection coefficients no longer satisfy the usual Hermitian conjugation symmetry, and the coercivity of the jump matrix is therefore not available a priori. The quantitative smallness condition on the initial data yields uniform bounds on the reflection coefficients and ensures the uniform positive definiteness of the Hermitian part of the associated jump matrix. The resulting coercivity allows us to establish the bounded invertibility of the associated singular integral operator through a Fredholm and vanishing-lemma argument. Under this condition, we prove an $L^{2}$-Sobolev bijective correspondence between the potential and scattering data, exclude spectral singularities on continuous spectra, and obtain the global existence and uniqueness of solutions. Moreover, the associated solution map is Lipschitz continuous on the admissible initial-data class.

\vspace{0.08in} \noindent {\it MSC:} 35P25, 35Q51, 35Q15, 35A01, 35G25, 37K15

\vspace{0.05in} \noindent {\it Keywords:}\,Reverse space-time nonlocal Fokas-Lenells equation; Cauchy problem;  Lipschitz continuous; Riemann-Hilbert problem; Scattering data; Global well-posedness.

\noindent\rule{\textwidth}{0.6pt}


\begin{spacing}{1.2}
\baselineskip=13pt
\tableofcontents
\vspace{-0.2in}
\end{spacing}

\section{Introduction}

The study of completely integrable nonlinear systems has long been a central theme in the fields of mathematical physics and applied mathematics, owing to their rich algebraic structures and deep physical significance \cite{Ablowitz1991, Faddeev1987}. A prototypical example is the integrable nonlinear Schr\"odinger (NLS) equation \cite{ZS}
\begin{equation} \label{nls}
\mathrm{i}q_t+q_{xx}+2\sigma|q|^2q=0, \quad q=q(x,t), \quad \sigma=\pm 1,
\end{equation}
which has served for decades as a fundamental model for the propagation of picosecond optical pulses in single-mode optical fibers \cite{Agrawal2013, Hasegawa1973}. However, as experimental techniques entered the femtosecond regime, the standard NLS equation (\ref{nls}) became insufficient to capture all relevant dynamics, since higher-order nonlinear dispersive effects, including spatio-temporal coupling, can no longer be neglected \cite{Brabec1997, Rothenberg1992}. In this context, the Fokas--Lenells (FL) equation \cite{Lenells2009b}
\begin{equation} \label{FL}
q_{xt}-q+\mathrm{i}|q|^2q_x=0
\end{equation}
arose as an integrable model incorporating such higher-order effects,  where $q_{xt}$ is the spatio-temporal dispersion, $|q|^2q_x$ corresponds to the nonlinear self-steepening effect of the Kerr medium. It was first derived by Fokas \cite{Fokas1995} through two Hamiltonian operators, and was later obtained by Lenells \cite{Lenells2009b} via a canonical transformation from the Kaup-Newell (KN) spectral problem \cite{Kaup1978}. Therefore, the FL equation (\ref{FL}) may be viewed as an integrable generalization of the NLS equation \eqref{nls}, the first negative flow of the KN hierarchy, and is particularly relevant to the modelling of ultrashort optical pulses \cite{Matsuno2012}.
In fact, Eq.~(\ref{FL}) can be equivalent to the general form \cite{Fokas1995,Lenells2009a,lene2010}
\bee\label{gfl}
 \mathrm{i}\hat{q}_{\tau}+\alpha \hat{q}_{\xi\xi}-\beta \hat{q}_{\xi\tau}+\gamma |\hat{q}|^2\hat{q}+\mathrm{i}\beta\gamma |\hat{q}|^2\hat{q}_{\xi}=0,
\ene
via the gauge transformation \cite{Lenells2009b}
\bee
\hat{q}(\xi,\tau)=\frac{1}{\beta}\sqrt{\frac{\alpha}{2\beta|\gamma|}}q(x,t)
\exp\left[\mathrm{i}\left(\frac{1}{\beta}\xi+\frac{2\alpha}{\beta^2}\tau\right)\right],\quad
x=2\xi+\frac{2\alpha}{\beta}\tau,\quad t=-\frac{\alpha}{\beta^3}\tau.
\ene
where $\alpha\beta>0$, otherwise, one can make the transform $\xi\to -\xi$ to make the condition hold.
$\hat{q}_{\xi\xi}$ denotes the group velocity dispersion, $\hat{q}_{\xi\tau}$ is the spatio-temporal dispersion,
$|\hat{q}|^2\hat{q}$ is the Kerr nonlinear term, $|\hat{q}|^2\hat{q}_{\xi}$ corresponds to the nonlinear self-steepening effect of the Kerr medium. Eq.~(\ref{gfl}) is sometime simplified as \cite{Lenells2009ab,lene2010}
\bee
 q_{xt}-q_{xx}+q-2\mathrm{i}q_x\mp \mathrm{i}|q|^2q_x=0
\ene
via $\hat{q}(\xi,\tau)=e^{\mathrm{i}\xi}q(x,t)$ with $x=\xi,\, t=\tau$ and $\alpha=\beta=\pm\gamma=1$.
As $\beta=0$, Eq.~(\ref{gfl}) reduces to the NLS equation (\ref{nls}). However, as $\beta\not=0$, The properties of Eq.~(\ref{gfl}) or (\ref{FL}) ) are rather different from ones of the NLS equation (\ref{nls}).

In recent years, the theory of integrable systems has been substantially enriched by the work of Ablowitz and Musslimani \cite{Ablowitz2013, Ablowitz2016}, who introduced a class of {\it nonlocal integrable equations} (e.g., nonlocal NLS equation)
 \begin{equation} \label{n-nls}
 \mathrm{i}q_t+q_{xx}+2\sigma q^2 \bar{q}(-x,t)=0,\quad q=q(x,t), \quad \sigma=\pm 1
\end{equation}
 involving the parity-time $(\mathcal{PT})$ symmetry \cite{Bender1998}, where the bar denotes the complex conjugate (Notice that (\ref{n-nls}) reduces to (\ref{nls}) as $q(-x,t)=q(x,t)$).
 The notion of $\mathcal{PT}$ symmetry originates in quantum mechanics, where certain non-Hermitian Hamiltonians may possess entirely real spectra under suitable symmetry conditions \cite{Bender1998,benderrmp-24,Konotop2016,Yan2013,yanbook23}. It has since found physical realizations in classical optics and metamaterials, especially in systems whose refractive indices are engineered to exhibit balanced gain and loss \cite{Makris2008, Ruter2010}. Following these developments, nonlocal reductions involving reverse space, reverse time, and reverse space-time symmetries have been systematically explored \cite{Ablowitz2018a,yanaml15,yanaml16,yanaml18}, and their physical relevance has also been investigated in multi-component settings \cite{Yang2018, Yan2013,yanaml15}. Representative examples include the nonlocal NLS equation \cite{Ablowitz2013,Wen2016}, the nonlocal modified Korteweg-de Vries (mKdV) equation \cite{Liu2024}, and the nonlocal derivative NLS (DNLS) equation \cite{Li2023}; related multidimensional and general nonlocal integrable models have also been developed \cite{Fokas2016, Ablowitz2017}. These equations exhibit a variety of wave phenomena, such as unconventional soliton structures and singularity formation, that differ markedly from those of their local counterparts \cite{Konotop2016, Wen2016,Yang2019}.

Over the past several decades, the inverse scattering transform (IST) \cite{Gardner1967, Lax1968}, in its Riemann--Hilbert (RH) formulation \cite{Biondini2014, Zhang2020}, has become one of the principal rigorous tools for analyzing the global dynamics and well-posedness of integrable partial differential equations \cite{Novikov1984, Deift1979}. The $L^2$-Sobolev bijectivity theory developed by Zhou \cite{Zhou1989, Zhou1998}, together with the Deift--Zhou approach \cite{Deift1993a, Deift2003}, provides a powerful analytical framework for constructing direct and inverse scattering maps and for controlling the time evolution of the associated scattering data. In a number of local self-adjoint settings, this framework permits the treatment of large initial data without imposing a small-norm assumption on the potential \cite{Deift1993b}. It also forms a fundamental basis for the nonlinear steepest descent method and for the analysis of long-time asymptotics and asymptotic stability of solitons \cite{Cuccagna2014, Borghese2018, Cuccagna2016}.

Aided by these RH techniques, substantial progress has been made for classical local integrable equations. Global well-posedness results for the DNLS equation were obtained by Pelinovsky {\it et al} \cite{Pelinovsky2017, Pelinovsky2018}, Liu {\it et al} \cite{Liu2016}, and Jenkins {\it et al} \cite{Jenkins2020}. For the local FL equation (\ref{FL}), the global well-posedness in the weighted Sobolev space $H^{3}(\mathbb{R})\cap H^{2,1}(\mathbb{R})$ has been established under various assumptions on the discrete spectrum and spectral singularities \cite{Cheng2025, Cheng2023, Li2025}. Moreover, Zhou's $L^2$-Sobolev RH methodology has been applied to several other integrable systems, including the massive Thirring model \cite{Pelinovsky2019}, the modified Camassa--Holm equation \cite{Yang2025}, and matrix NLS equations \cite{Fan2025}. In many local models, self-adjointness or an appropriate Hermitian reduction yields Schwarz reflection symmetries for the scattering data and coercivity properties for the associated jump matrices. These features are central to the unique solvability of the corresponding RH problems.

The passage from local to nonlocal integrable systems introduces substantial analytical differences. Under nonlocal reductions, the standard involution symmetries of the scattering matrix are generally modified \cite{Gerdjikov2017}. In particular, discrete eigenvalues need not occur in complex-conjugate pairs, and the jump matrices on the continuous spectrum are typically non-Hermitian. Consequently, the coercivity mechanisms available in classical local problems are no longer automatic \cite{Ablowitz2016, Yang2018}. Despite these difficulties, Ablowitz and Musslimani \cite{Ablowitz2016} developed the IST for the reverse-space nonlocal NLS equation, and Rybalko and Shepelsky \cite{Rybalko2019} subsequently derived its long-time asymptotics by the nonlinear steepest descent method. More recently, global well-posedness has been established for several spatially nonlocal equations, including the nonlocal NLS equation \cite{Zhao2024}, the nonlocal mKdV equation \cite{Liu2024}, and the nonlocal DNLS equation \cite{Li2023}, in suitable weighted Sobolev spaces.

While purely spatial nonlocalities have received considerable attention, the incorporation of time-reversal symmetry introduces an additional layer of difficulty. Ablowitz {\it et al} \cite{Ablowitz2018a, Ablowitz2018b} studied the IST for the reverse space-time nonlocal NLS and sine--Gordon equations, showing that the simultaneous coupling of spatial and temporal reflections leads to a nonstandard interaction between the spatial scattering problem and the time evolution. In particular, the potential entering the scattering problem at time $t$ involves the value of the field at the reflected time $-t$. However, to the best of our knowledge,  the global well-posedness of the reverse space-time nonlocal FL equation has not yet been fully addressed.

In this paper, we consider the existence of global solutions of the Cauchy problem for the reverse space-time nonlocal Fokas-Lenells (nFL) equation \cite{Zhang2019h}
\begin{equation}\label{eq:nFL}
\left\{\begin{array}{l}
q_{xt}-q+\mathrm{i}q q^{\mathcal{PT}}q_x=0,
 \vspace{0.1in}\\
q(x,0)=q_0(x)\in H^{3}(\mathbb{R})\cap H^{2,1}(\mathbb{R}),
\end{array}\right.
\end{equation}
where the complex envelope field $q=q(x,t)$, the nonlocal term $q^{\mathcal{PT}}(x,t):=q(-x,-t)$, no complex conjugation is involved in the definition of $q^{\mathcal{PT}}$, which differs from the nonlocal term $\bar{q}(-x,t)$ in the nonlocal NLS equation (\ref{n-nls}). Notice that the term $q q^{\mathcal{PT}}$ in the nonlocal FL equation (\ref{eq:nFL}) must not be real and differs from the similar real-valued term $|q|^2$ in the FL equation (\ref{FL}).  The nFL Eq.~(\ref{eq:nFL}) possesses the Kaup–Newell-type Lax pair
\begin{equation}\label{eq:Lax pair}
		\left\{ \begin{array}{ll}
			  \phi_{x} = X(x,t;\lambda)\phi, &  X = \mathrm{i}\lambda^2 \sigma_3 + \mathrm{i}\lambda Q_x, \\[0.5em]
			  \phi_{t} = T(x,t;\lambda)\phi, & T = -\dfrac{\mathrm{i}}{4\lambda^2}\sigma_3
+\dfrac{1}{2\lambda}Q\sigma_3
+\dfrac{\mathrm{i}}{2}Q^2\sigma_3,
		\end{array} \right.
	\end{equation}
where $\phi=\phi(x,t;\lambda)$ is a matrix-valued eigenfunction, $\lambda\in\mathbb{C}\setminus\{0\}$ is the spectral parameter, and
\begin{equation}\label{eq:PotentialMatrix}
Q=Q(x,t)=
\begin{pmatrix}
0 & q^{\mathcal{PT}}(x,t) \v\\
q(x,t) & 0
\end{pmatrix},
\qquad
\sigma_3=
\begin{pmatrix}
1 & 0 \v\\
0 & -1
\end{pmatrix}.
\end{equation}
The compatibility condition of the Lax pair (\ref{eq:Lax pair}), $X_t-T_x+[X,T]=0$,
just yields the nFL equation \eqref{eq:nFL} together with its reverse space-time reflected counterpart.

The main analytical difficulties arise in the nonlocal FL equation (\ref{eq:nFL}) from two sources: The first one is the singular spectral dependence of the Kaup-Newell-type Lax pair, especially near $\lambda=0$ and $\lambda=\infty$; The second one is the reverse space-time reduction itself, which modifies the usual symmetry relations among the scattering coefficients. As a result, the RH jump matrix is generally non-Hermitian, and the positivity and coercivity properties available in standard local RH formulations cannot be invoked directly. Recovering an appropriate coercive structure is therefore one of the central issues addressed in this paper.

To address these difficulties, we develop an IST framework adapted to the reverse space-time nonlocal FL equation. We establish a direct and inverse scattering correspondence of $L^2$-Sobolev type between the potential and the associated reflection data \cite{Zhou1998} for initial data in the weighted Sobolev space
$H^{3}(\mathbb{R})\cap H^{2,1}(\mathbb{R})$.
The $H^3$ regularity provides the derivative control needed for the effective potential and for the reconstruction formulas, whereas the $H^{2,1}$ regularity supplies the spatial decay and spectral regularity required in the direct and inverse scattering maps. These assumptions yield, in particular, the Sobolev regularity, decay properties, and Lipschitz estimates for the reflection coefficients and for the reconstructed potential.

The analysis proceeds in several steps. In the direct scattering transform, the Kaup--Newell-type spectral problem \cite{Kaup1978} exhibits singular behavior at both the origin and infinity in the natural $\lambda$-plane. Following the strategies developed for local FL equations \cite{Cheng2025, Cheng2023}, we combine a suitable gauge transformation with the spectral change of variables
$z=\lambda^2$,
thereby converting the problem into a Zakharov--Shabat-type spectral problem on the $z$-plane. We then construct the Jost solutions, establish their analytic and asymptotic properties, and derive the symmetry relations imposed on the scattering data by the reverse space-time reduction.

In the inverse scattering transform, we formulate a regularized RH problem on the real axis of the transformed $z$-plane and reduce it, via the Cauchy projection operators, to a Beals--Coifman singular integral equation. Since the reverse space-time reduction does not impose the usual Hermitian conjugation relation between the two reflection coefficients, the coercivity of the associated jump matrix is not automatic. To recover this property, we impose an explicit smallness condition on the initial effective potential. Volterra estimates for the scattering matrix then show that the two reflection coefficients are uniformly bounded in modulus by a constant strictly smaller than one. This bound implies the uniform positive definiteness of the Hermitian part of the RH jump matrix.

\begin{figure}[!t]
		\centering
		\hspace{-0.1in}\begin{tikzpicture}
           [>=Stealth,
			node distance=2cm and 4.8cm ]
			\node (q0) {$\begin{array}{c} q_0(x)\in H^{3}(\mathbb{R})\cap H^{2,1}(\mathbb{R})\vspace{0.05in} \\
                     \footnotesize{Initial \,\, condition}\end{array}$ };
			\node[right=of q0] (r0)
			{$\begin{array}{c}r_{1,2}(0;z)\in\mathcal{W}(\mathbb{R})  \vspace{0.05in}\\
            \footnotesize{Stationary\,\, reflection \,\, coefficients} \end{array} $};
			
			\node[below=of q0] (qt)
			{$\begin{array}{c} q(x,t)\in C(H^{3}(\mathbb{R})\cap H^{2,1}(\mathbb{R}), \mathbb{R}) \vspace{0.05in}\\
               Global \,\, solution
               \end{array} $};
			
			\node[below=of r0] (rt)
			{$\begin{array}{c} r_{1}(0;z) \mathrm{e}^{-\frac{\mathrm{i}t}{2z}}, r_2(0;z) \mathrm{e}^{\frac{\mathrm{i}t}{2z}}\in\mathcal{W}(\mathbb{R})  \vspace{0.05in}\\
            Time$-$dependent \,\, reflection\,\, coefficients \end{array}$};
			
			\draw[->] (q0) -- node[above]{{\it Lipschitz\,\,continuous}} node[below]{{\bf Step\, 1}: direct scattering} (r0);
			
			\draw[->] (r0) -- node[left]{\bf Step \,2} node[right]{\it Time\, propagation} (rt);
			
			\draw[->] (rt) -- node[above]{{\it Lipschitz\,\, continuous}} node[below]{{\bf Step\, 3}: inverse scattering} (qt);
			
			\draw[dashed,->] (q0) -- (qt);
			
		\end{tikzpicture}
		\caption{The general scheme for the global well-posedness of the nFL equation (\ref{eq:nFL}).}
		\label{fig1}
	\end{figure}
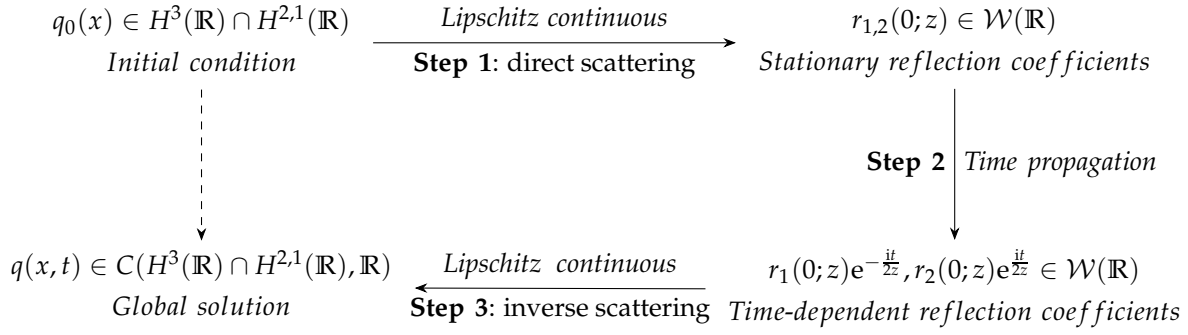

The resulting coercivity enables a Fredholm alternative and vanishing-lemma argument for the corresponding Beals--Coifman singular integral equation \cite{Beals1984}. More precisely, the singular integral operator is shown to be Fredholm of index zero, while the coercivity estimate excludes nontrivial solutions of the associated homogeneous problem. The Fredholm alternative then gives bounded invertibility and hence the unique solvability of the RH problem. The same effective-potential estimates also provide a uniform lower bound for the relevant scattering coefficients on the continuous spectrum, thereby excluding spectral singularities. Thus, the imposed smallness condition is a sufficient technical hypothesis for the inverse-scattering framework developed here, rather than a necessary condition for the well-posedness of the equation itself.

The primary objective of this paper is to rigorously establish the global well-posedness of the reverse space-time nonlocal FL equation. The time evolution of the reflection coefficients preserves their moduli and the coercivity bounds for the jump matrix. Consequently, the inverse problem remains uniquely solvable for every $t\in\mathbb{R}$, allowing the solution to be reconstructed globally in time (see Fig.~\ref{fig1}).

\subsection{ Main results}

Our main result is stated as follows.

\begin{theorem}\label{main-thm}
Let the initial datum $q_0(x)=q(x,0) \in H^3(\mathbb{R}) \cap H^{2,1}(\mathbb{R})$, and denote its reverse-space counterpart at $t=0$ by
$q_0^{\mathcal{PT}}(x):=q_0(-x)$. Let the associated initial effective potential matrix be given by
\begin{equation}
\widetilde{Q}_0(x)
:=
\frac{\mathrm{i}}{2}
\begin{pmatrix}
q_{0,x}(x)(q_0^{\mathcal{PT}})_{x}(x)
&
-(q_0^{\mathcal{PT}})_{x}(x)
\\[0.8em]
q_{0,x}^2(x)(q_0^{\mathcal{PT}})_{x}(x)-2\mathrm{i}q_{0,xx}(x)
&
-q_{0,x}(x)(q_0^{\mathcal{PT}})_{x}(x)
\end{pmatrix}.
\end{equation}
Assume that
\begin{equation}\label{eq:small-norm}
\|\widetilde{Q}_0\|_{L^1(\mathbb{R})}
<
\ln\left(\frac{3}{2}\right).
\end{equation}
Then the Cauchy problem for the nonlocal Fokas--Lenells equation \eqref{eq:nFL} admits a unique global solution
\begin{equation}
q\in C\big(\mathbb{R};H^3(\mathbb{R})\cap H^{2,1}(\mathbb{R})\big).
\end{equation}
Furthermore, the solution map
\begin{equation}
H^3(\mathbb{R})\cap H^{2,1}(\mathbb{R})\ni q_0
\longmapsto
q\in C\big(\mathbb{R};H^3(\mathbb{R})\cap H^{2,1}(\mathbb{R})\big)
\end{equation}
is Lipschitz continuous.
\end{theorem}

\begin{remark} Similarly, the results in Theorem \ref{main-thm} can also be extended to another forms of the nFL equation (\ref{eq:nFL}) \cite{zhang2022}, for example,
\bee
\begin{array}{ll}
{\rm (I)}\quad & q_{xt}+q-q\bar{q}(-x,-t)q_x=0, \v\\
{\rm (II)} \quad & q_{xt}-\i q+\i q\bar{q}(-x,t)q_x=0, \v\\
{\rm (III)} \quad & q_{xt}-\i q-q\bar{q}(x, -t)q_x=0, \v\\
{\rm (IV)} \quad & q_{xt}-q-q\bar{q}(-x,-t)q_x=0,
 \end{array}
\ene
from some nonlocal reductions $p(x,t)=F(q(x,t))$ \cite{yanaml15,yanaml16,yanaml18, yang2017} of the coupled Fokas–Lenells system~\cite{Lenells2009b}
\bee
\left\{\begin{array}{l}
 q_{xt}+q-\i qpq_x=0, \v\\
 p_{xt}+p+\i pqp_x=0.
\end{array}\right.
\ene
\end{remark}

The remainder of this paper is organized as follows. In Section~\ref{section-DST}, we develop the direct scattering transform, including the construction and analyticity of the Jost solutions, as well as the symmetry and regularity properties of the scattering coefficients under the reverse space-time reduction. In Section~\ref{section-RHP}, we formulate the associated RH problems, beginning with a basic RH problem on the $\lambda$-plane and subsequently transforming it into a regularized problem on the $z$-plane. In Section~\ref{section-IST}, we develop the inverse scattering transform, derive coercive estimates for the associated jump matrix, and prove the bounded invertibility of the corresponding Beals-Coifman singular integral operator under the effective-potential smallness condition. In Section~\ref{section-Time}, we determine the time evolution of the scattering data, formulate the corresponding time-dependent RH problems, and verify that the coercivity bounds are preserved for all time. Finally, in Section~\ref{section-proof}, we complete the proof of Theorem~\ref{main-thm} by establishing global existence, uniqueness, and the Lipschitz continuity of the solution map.

\v \textbf{Notations}. We collect here some notations used throughout the paper. Let $I$ be an interval on the real line $\mathbb{R}$ and let $X$ be a Banach space. We denote by $C(I,X)$ the space of continuous functions on $I$ taking values in $X$, equipped with the standard supremum norm
\begin{equation}
\|q\|_{C(I,X)}
=
\sup_{\tau\in I}\|q(\tau)\|_X .
\end{equation}

Throughout this paper, we shall use the weighted Lebesgue and Sobolev spaces defined by
\begin{subequations}
\begin{align}
L^{p,s}(\mathbb{R})
&:=
\left\{
u\in L^p(\mathbb{R})
\;\big|\;
\langle x\rangle^s u\in L^p(\mathbb{R})
\right\},
\\
H^{k,s}(\mathbb{R})
&:=
\left\{
u\in L^2(\mathbb{R})
\;\big|\;
\langle x\rangle^s\partial_x^j u\in L^2(\mathbb{R}),
\quad j=0,1,\ldots,k
\right\},
\end{align}
\end{subequations}
where $\langle x\rangle:=\sqrt{1+|x|^2}$, $1\le p\le\infty$ and $k,s\in\mathbb{N}\cup\{0\}$. Furthermore, to characterize the regularity of the scattering data, we introduce the space $\mathcal{W}(\mathbb{R})$ by
\begin{equation}
\mathcal{W}(\mathbb{R})
:=
\left\{
u\in H^1(\mathbb{R})\cap L^{2,1}(\mathbb{R})
\;\big|\;
x^{-2}u(x)\in L^2(\mathbb{R})
\right\}.
\end{equation}

 The $L^{1}$ matrix norm of the 2-by-2 matrix function $A=(A_{ij})_{2\times 2}$ is defined as $\| A \|_{L^1} := \Sigma_{i,j=1}^2 \| A_{ij} \|_{L^1}.$ And two column vectors $e_j\, (j=1,2)$ are defined as
 \bee
  e_1= (1,\, 0)^T, \qquad  e_2=(0,\, 1)^T.
 \ene


	\section{Direct scattering problem: Jost solutions and scattering data}\label{section-DST}

	In this section, we will prove the existence of the eigenfunctions of the Lax pair (\ref{eq:Lax pair}).
	With the transformation
	\begin{equation}\label{eq:Transform}
		\phi(x,t;\lambda)=\psi(x,t;\lambda) \exp\left[\mathrm{i}\left(\lambda^2 x - \frac{1}{4 \lambda^{2}} t\right) \sigma_3\right],
	\end{equation}
	$\psi(x,t;\lambda)$ satisfies a modified Lax pair
	\begin{equation}\label{eq:New Laxpair}
			\left\{ \begin{array}{ll}
			 \psi_{x} = \mathrm{i} \lambda^2 [\sigma_3, \psi] + \mathrm{i} \lambda Q_{x} \psi, \\[0.5em]
			  \psi_{t} = -\dfrac{\mathrm{i}}{4 \lambda^2} [\sigma_3, \psi] +(\dfrac{1}{2 \lambda}Q+\dfrac{\mathrm{i}}{2} Q^2) \sigma_3 \psi,
		\end{array} \right.
	\end{equation}
	which can be written in the full derivative form
	\begin{equation}\label{eq: fullderivative}
		d(\mathrm{e}^{\mathrm{i} (\lambda^2 x - (4 \lambda)^{-2} t ) \hat{\sigma}_3} \psi) = \mathrm{e}^{\mathrm{i} (\lambda^2 x - (4 \lambda)^{-2} t ) \hat{\sigma}_3} \left[\mathrm{i} \lambda Q_{x}\, dx + (\dfrac{Q}{2 \lambda} + \dfrac{\mathrm{i}}{2} Q^2) \sigma_3\, dt\right] \psi.
	\end{equation}
	and $\hat{\sigma}_3 (A) = \sigma_3 A \sigma_3^{-1}$. At the same time, this modified Lax pair (\ref{eq:New Laxpair}) admits the Jost solutions with asymptotics as the potential matrix $Q_x$ vanishes sufficiently quickly at infinity, $\psi^{\pm}(x,t,\lambda)\to I, \quad x\to\pm\infty.$
	and satisfy Volterra integral equations
	\begin{equation}\label{eq: psi-Volterra}
		\psi^{\pm}(x,t;\lambda) = I + \mathrm{i} \lambda \int_{\pm \infty}^{x} \mathrm{e}^{\mathrm{i} \lambda^2 (x-s) \hat{\sigma}_3} Q_s(s,t) \psi^{\pm}(s,t;\lambda) ds.
	\end{equation}
	\par
	The Lax pair \eqref{eq:New Laxpair} has singularities at $\lambda=0$ and $\lambda=\infty$, which is the same as the FL equation. Here, we will make a transformation to construct the Jost solutions and so on.

\begin{figure}[!t]
	\centering
	\includegraphics[scale=0.75]{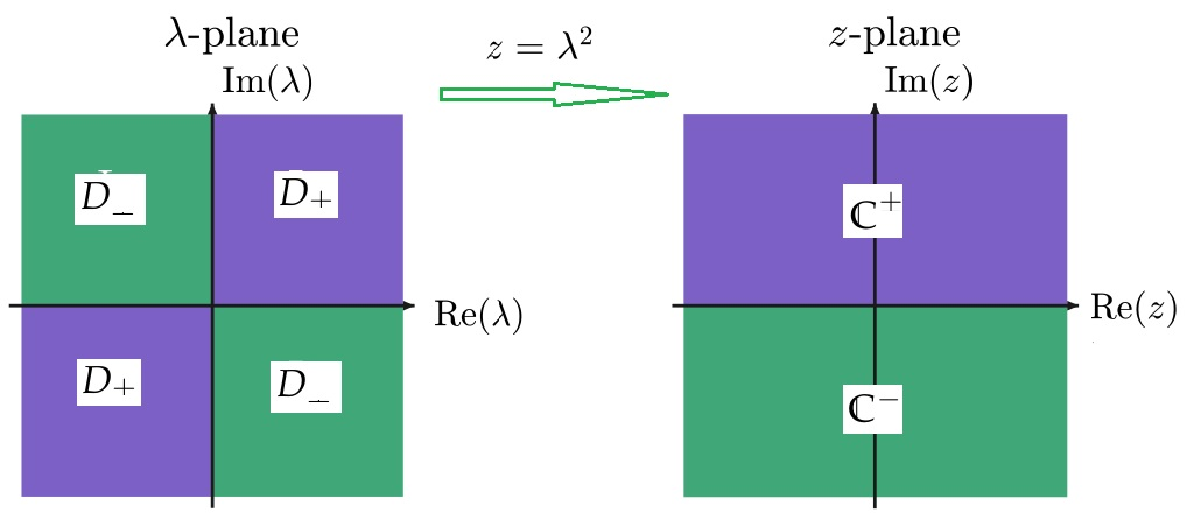}
	\caption{Analytical regions on the $\lambda$-plane (left) and $z=\lambda^2$-plane (right).}
	\label{fig2}
\end{figure}

	\subsection{Jost solutions and analyticity on the $z$-plane}

	As usual, we set $t$ to be fixed and omit it in the following analysis. To more conveniently analyze the Lax pair \eqref{eq:New Laxpair}, we introduce a transformation defined by
	\begin{equation}\label{eq:lambda to z}
		\Psi(x;z) = T_1(x;\lambda) \psi(x;\lambda) T_2(x;\lambda),
	\end{equation}
	with $z = \lambda^2 $ and
	\bee T_1(x;\lambda) = \begin{pmatrix}
		1 & 0 \vspace{0.05in}\\
		q_x & -2\lambda
	\end{pmatrix}, \quad
	T_2(x;\lambda)=\begin{pmatrix}
		1 & 0 \vspace{0.05in} \\
		0 & (-2\lambda)^{-1}
	\end{pmatrix},
\ene
	then the Lax pair \eqref{eq:New Laxpair} is changed into a Zakharov-Shabat-type spectral problem
	\begin{equation}\label{eq:ZSprobelm}
		\Psi_x = \mathrm{i} z \left[\sigma_3, \Psi\right]+ \widehat{Q} \Psi,
	\end{equation}
	where
	\bee \widehat{Q} = \frac{\mathrm{i}}{2} \begin{pmatrix}
		q_x (q^{PT})_{x} & -(q^{PT})_{x} \vspace{0.1in}\\
		q_x^2 (q^{PT})_{x} - 2\mathrm{i} q_{xx} & -q_x (q^{PT})_{x}
	\end{pmatrix}.\ene
 Similarly, the time-part of Lax pair \eqref{eq:New Laxpair} is also changed into the corresponding form, which is omitted here.
 	
It can be shown that $\Psi^{\pm}(x;z)$ satisfies the Volterra integral equation
	\begin{equation}\label{eq:Volterra equation}
		\Psi^{\pm}(x,t;z) = I + \int_{\pm \infty}^{x} \mathrm{e}^{\mathrm{i} z (x-y) \hat{\sigma}_3} \widehat{Q} \Psi^{\pm}(s,t;z) ds.
	\end{equation}
Let $\Psi^{\pm}(x,t;z) =(\Psi_1^{\pm}(x,t;z), \Psi_2^{\pm}(x;z))$, where $\Psi_j^{\pm}(x,t;z)\, (j=1,2)$ denotes the column vectors. Thus, we can show the following proposition (see Fig.~\ref{fig2}(right)).

	\begin{proposition}\label{prop:Analytic}
	Let $q \in H^{3}(\mathbb{R}) \cap H^{2,1}(\mathbb{R})$. Then for every $z \in \mathbb{R},$ there exist unique solutions $\Psi^{\pm}(\cdot,t;z) \in L^{\infty} (\mathbb{R})$ satisfying \eqref{eq:Volterra equation}. Moreover, for every $x\in \mathbb{R}, \Psi_1^{\mp}(x,t; \cdot)$ and $\Psi_2^{\pm}(x,t;z)$ are continued analytically in $\mathbb{C}^{\mp}$. Finally, there exists a $z$-independent constant $C>0$ such that
		\begin{equation}\label{eq:bound1}
			\| \Psi_1^{\pm}(\cdot,t;z) \|_{L^{\infty}} + \| \Psi_2^{\mp}(\cdot,t;z) \|_{L^{\infty}} \le C \quad z\in \mathbb{C}^{\pm}.
		\end{equation}
	\end{proposition}
	\begin{proof}
		It suffices to prove the statement for one of the situations, for example, for $\Psi_1^{-}$. The proof for others is analogous. From \eqref{eq:Volterra equation}, we can obtain that
		\begin{equation}\label{eq:Psi1}
			\Psi_1^{-}(x;z) = e_1 + (K \Psi_1^{-})(x;z),\qquad e_1=(1, 0)^T,
		\end{equation}
		with the definition of the integral operator $K$ being
		\begin{equation}\label{op:K operator}
			(Kg)(x;z) := \dfrac{i}{2} \int_{-\infty}^{x}
			\begin{pmatrix}
				1 & 0 \v\\
				0 & \mathrm{e}^{-2  \mathrm{i}z (x-y)}
			\end{pmatrix}
			\begin{pmatrix}
				q_y (q^{PT})_{y} & -(q^{PT})_{y} \vspace{0.1in}\\
				q_y^2 (q^{PT})_{y} - 2\mathrm{i} q_{yy} & -q_y (q^{PT})_{y}
			\end{pmatrix} g(y) \, dy.
		\end{equation}
	where
   $g(x) = (g_1, g_2)^{T}$. We define that $ \| g \|_{L^{\infty}} = \| g_1\|_{L^{\infty}} + \| g_2 \| _{L^{\infty}}.$   Thus for every $z \in \mathbb{C}^{-}$ and every $x_0 \in \mathbb{R}$, we have
		\begin{equation}
		\begin{split}
	    \| (Kg)(\cdot;z) \|_{L^{\infty} (-\infty, x_0)} \le \bigg( \|q_x\|_{L^2(-\infty, x_0)} \|(q^{PT})_x\|_{{L^2(-\infty,x_0)}} + \dfrac{1}{2} \| (q^{PT})_{x} \|_{L^1(-\infty, x_0)} + \| q_{xx} \|_{L^1(-\infty, x_0)} \\
	    +\dfrac{1}{2} \|q_x\|_{L^\infty(-\infty, x_0)} \|q_x\|_{L^2(-\infty, x_0)} \|(q^{PT})_x\|_{{L^2(-\infty,x_0)}} \bigg) \| g \|_{L^{\infty}(-\infty, x_0)},
	    \end{split}
	    \end{equation}
		which implies that the operator $K$  is contraction if $x_0 \in \mathbb{R} $ is chosen so that
		\begin{equation}\label{eq:estimate1}
			 \|q_x\|_{L^2} \|(q^{PT})_{x}\|_{L^2} + \dfrac{1}{2} \| (q^{PT})_x \|_{L^1} + \| q_{xx} \|_{L^1} \\
			+ \dfrac{1}{2}\|q_x\|_{L^\infty} \|q_x\|_{L^2} \|(q^{PT})_{x}\|_{L^2} < 1,
		\end{equation}
		here we omit the range to make writing easier. By the Banach fixed point theorem, for $x_0$ and every $z\in \mathbb{C}^{-}$, there exist a unique solution $\Psi_1^{-}(x;z) \in L^{\infty}(-\infty, x_0)$ to the equation \eqref{eq:Psi1}. To extend this result to $L^{\infty}(\mathbb{R})$, we can split $\mathbb{R}$ into a finite number of intervals such that the estimate \eqref{eq:estimate1} is satisfied in each interval. Unique solutions in each interval can be glued together to obtain the unique solution $\Psi_1^{-}(\cdot;z) \in L^{\infty}(\mathbb{R})$ for every $z \in \mathbb{C}^{-}$.\par
		The analyticity of $\Psi_1^{-}(x;\cdot)$ in $\mathbb{C}^{-}$ for every $x\in \mathbb{R}$ follows from the absolute and uniform convergence of the Neumann series of analytic functions in $z$.
		If $q\in H^{3}(\mathbb{R}) \cap H^{2,1}(\mathbb{R})$, we have $\widehat{Q} \in L^{1}(\mathbb{R})$. For every $g(x;z) \in L^{\infty}(\mathbb{R} \times \mathbb{C}^{-})$, we have
		\begin{equation}\label{eq:converge}
			\| K^{n} g \|_{L^{\infty}} \leq \frac{1}{n!} \| \widehat{Q} \|_{L^1}^n \| g \|_{L^{\infty}}.
		\end{equation}
		As a result, the Neumann series for the Volterra integral equation \eqref{eq:Psi1} for $\Psi_1^{-}$ converges
		absolutely and uniformly for every $x\in \mathbb{R}$ and $z\in \mathbb{C}^{-}$ and contains analytic functions of $z$ for $z\in\mathbb{C}^{-}$. Therefore, $\Psi_1(x;\cdot)$ is analytic in $\mathbb{C}^{-}$ for every $x\in \mathbb{R}$ and it satisfies the bound \eqref{eq:bound1}. \hfill
	\end{proof}\par
	\begin{proposition}\label{prop:Limits}
		Under the conditions of Proposition \ref{prop:Analytic}, for every $x\in \mathbb{R}$, the Jost solutions $\Psi^{\pm}(x;z)$ admit the following limits along a contour in the domains of their analyticity extended to $|\mathrm{Im} z| \to \infty$
		\begin{equation}\label{eq:inftylimits}
			\lim_{|z|\to \infty} \Psi^{\pm}(x;z) = \mathrm{e}^{\frac{i}{2} m_{\pm} \sigma_3},\quad m_{\pm}(x) = \int_{\pm \infty}^{x} q_y (q^{PT})_{y} \, dy.
		\end{equation}
		Moreover, if $q\in C^{2}(\mathbb{R})$, then for every $x\in \mathbb{R}$, the Jost solutions satisfy the following limits along a contour in the domains of their analyticity extended to $|\mathrm{Im} z| \to \infty$
		\begin{subequations}\label{eq:zlimits}
			\begin{align}
				\lim_{|z| \to \infty} z(\Psi_1^{\pm}(x;z) - \mathrm{e}^{\frac{\mathrm{i}}{2} m_{\pm}(x)} e_1) &= s_1^{\pm}(x) e_1 + s_2^{\pm}(x) e_2, \\
				\lim_{|z| \to \infty} z(\Psi_2^{\pm}(x;z) - \mathrm{e}^{-\frac{\mathrm{i}}{2} m_{\pm}(x)} e_2) &= \nu_{1}^{\pm}(x) e_1 + \nu_{2}^{\pm}(x) e_2,
			\end{align}
		\end{subequations}
		where
		\begin{align*}
		s_1^{\pm}(x) &= -\frac{1}{4} \mathrm{e}^{\frac{\mathrm{i}}{2} m_{\pm}(x)} \int_{\pm \infty}^{x} \left(  (q^{PT})_{y}(y) q_{yy}(y) + \frac{\mathrm{i}}{2} ((q^{PT})_{y})^2(y) q_y^2(y) \right) dy, \\
		 s_2^{\pm}(x) &= -\dfrac{\mathrm{i}}{2} \partial_x (q_{x}(x) \mathrm{e}^{\frac{i}{2} m_{\pm}(x)} ), \\
		 \nu_1^{\pm}(x) &= \frac{1}{4} (q^{PT})_{x}(x) \mathrm{e}^{-\frac{\mathrm{i}}{2} m_{\pm}(x)}, \\
		 \nu_2^{\pm}(x) &= \dfrac{1}{4} \mathrm{e}^{-\frac{\mathrm{i}}{2} m_{\pm}(x)} \int_{\pm\infty}^{x} \left(  (q^{PT})_{y}(y) q_{yy}(y) + \frac{\mathrm{i}}{2} ((q^{PT})_{y})^2(y) q_y^2(y) \right) dy.
		\end{align*}
	\end{proposition}
	\begin{proof}
		Again, we prove the statement for the Jost solution $\Psi_1^{-}$ only. The proofs of others are analogous. We rewrite the integral equation \eqref{eq:Volterra equation} in the component form:
		\begin{equation}\label{eq:component1}
			\Psi_{11}^{-}(x;z) = 1 + \dfrac{\i}{2} \int_{-\infty}^{x} (q^{PT})_{y} \left[q_y \Psi_{11}^{-}(y;z) - \Psi_{21}^{-}(y;z) \right] dy,
		\end{equation}
		\begin{equation}\label{eq:component2}
			\Psi_{21}^{-}(x;z) = \dfrac{\i}{2} \int_{-\infty}^{x} \mathrm{e}^{-2 \mathrm{i} z (x-y)} \left[\left(q_y^2 (q^{PT})_{y} -2 \mathrm{i} q_{yy} \right)\Psi_{11}^{-}(y;z) - q_y (q^{PT})_{y} \Psi_{21}^{-}(y;z)
			\right] dy.
		\end{equation}
		By bounds \eqref{eq:bound1} in Proposition \ref{prop:Analytic}, for every $q \in H^{3}(\mathbb{R}) \cap H^{2,1}(\mathbb{R})$, the integrand of  \eqref{eq:component2} is bounded for every $z\in \mathbb{C}^{-}$ by am absolutely integrable $z$-independent function. While for $z\in \mathbb{C}^{-}$, we have $| \mathrm{e}^{-2 \mathrm{i} z (x-y)} |= \mathrm{e}^{2 \mathrm{Im}z (x-y)} \to 0, \quad |z|\to \infty$, by Lebesgue's dominated convergence theorem, we obtain that $\Psi_{21}^{-}(x;z)\to 0,\,\, |z|\to \infty$ and
		\begin{equation}
			\lim_{|z| \to \infty} \Psi_{11}^{-}(x;z) = 1 + \dfrac{\i}{2} \int_{-\infty}^{x} q_y (q^{PT})_{y} \lim_{|z| \to \infty} \Psi_{11}^{-}(y;z) \, dy.
		\end{equation}
		By solving the integral equation, we have a unique solution
		\begin{equation}
			\lim_{|z| \to \infty} \Psi_{11}^{-}(x;z) = \mathrm{e}^{\frac{\mathrm{i}}{2} m_{-}(x)},
		\end{equation}
		which means
		\begin{equation}
			\lim_{|z| \to \infty} \Psi_{1}^{-}(x;z) = \mathrm{e}^{\frac{\mathrm{i}}{2} m_{-}(x)} e_1.
		\end{equation}
		In a similar way, we can show that
		\begin{equation}
		    \lim_{|z| \to \infty} \Psi_{2}^{-}(x;z) = \mathrm{e}^{-\frac{\mathrm{i}}{2} m_{-}(x)} e_2.
		\end{equation}
		Combining the above equations, we can get \eqref{eq:inftylimits}.\par
		For every $x\in \mathbb{R}$ and every $\delta > 0$, we use the Laplace method of asymptotic analysis, which allows us to split integration for $(-\infty, x-\delta)$ and $(x-\delta, x)$. And we rewrite the integral equation of $\Psi_{21}^{-}$ in the equivalent form:
		\begin{equation}
			\begin{array}{rl}
				\Psi_{21}^{-}(x;z) & \d = \int_{-\infty}^{x-\delta} \mathrm{e}^{-2 \mathrm{i} z (x-y)} \mu(y;z)\, dy + \mu(x;z) \int_{x - \delta}^{x} \mathrm{e}^{-2 \mathrm{i} z (x-y)}\, dy  \v\\
				&\qquad  \d + \int_{x - \delta}^{x} \mathrm{e}^{-2 \mathrm{i} z (x-y)} (\mu(y;z)- \mu(x;z))\, dy  \v\\
            &\d \equiv I_1 + I_2 + I_3,
			\end{array}
		\end{equation}
		where
		$$ \mu(x;z) = \dfrac{\i}{2} \left[\left(q_x^2 (q^{PT})_{x} -2 \mathrm{i} q_{xx} \right)\Psi_{11}^{-}(x;z) - q_x (q^{PT})_{x} \Psi_{21}^{-}(x;z)
		\right]. $$
		Since $\mu(\cdot;z) \in L^{1}(\mathbb{R})$, we have
		$$ | I_1 | \le \mathrm{e}^{2 \delta \mathrm{Im}z} \| \mu(\cdot;z) \|_{L^{1}}. $$
		Since $\mu(\cdot;z) \in C(\mathbb{R})$, we have
		$$ | I_3 | \le \dfrac{1}{2 | \mathrm{Im}z |} \| \mu(\cdot;z) - \mu(x;z) \|_{L^{\infty}(x - \delta, x)}.$$
		On the other hand, we have the exact value
		$$ I_2 = \dfrac{1}{2 \mathrm{i} z} (1 - \mathrm{e}^{-2 \mathrm{i} z \delta}) \mu(x;z).$$
		Let us choose $\delta := [-\mathrm{Im}z]^{-1/2}$ such that $\delta\to 0$ as $| \mathrm{Im}z | \to \infty$, then we obtain
		\begin{equation}\label{eq:zcomponet2}
			\begin{split}
			\lim_{|z| \to \infty} z \Psi_{21}^{-}(x;z)& = -\dfrac{\mathrm{i}}{2} \lim_{|z| \to \infty} \mu(x;z)
			= \dfrac{1}{4} \left(-2 \mathrm{i} q_{xx}(x) + q_x^2(x) (q^{PT})_{x}(x) \right) \mathrm{e}^{\frac{i}{2} m_{-}(x)} = s_{2}^{-}(x).
		    \end{split}
		\end{equation}
		with $s_{2}^{-}(x) = -\dfrac{\mathrm{i}}{2} \partial_x (q_{x}(x) \mathrm{e}^{\frac{i}{2} m_{-}(x)} ).$\par
		On the other hand, the first equation \eqref{eq:component1} can be rewritten as the differential equation
		$$ \partial_x \Psi_{11}^{-}(x;z) = \frac{\mathrm{i}}{2} (q_x (q^{PT})_{x})(x) \Psi_{11}^{-}(x;z) - \frac{\mathrm{i}}{2} (q^{PT})_{x}(x) \Psi_{21}^{-}(x;z),$$
		and we can use the integrating factor $\mathrm{e}^{-\frac{i}{2} m_{-}(x)}$, then
		$$\partial_x \left( \Psi_{11}^{-}(x;z) \mathrm{e}^{-\frac{\mathrm{i}}{2} m_{-}(x)} \right) = - \frac{\mathrm{i}}{2} (q^{PT})_{x}(x) \Psi_{21}^{-}(x;z) \mathrm{e}^{-\frac{\mathrm{i}}{2} m_{-}(x)}.$$
		and we can get that
		\begin{equation}
			\Psi_{11}^{-}(x;z) - \mathrm{e}^{\frac{\mathrm{i}}{2} m_{-}(x)} = \mathrm{e}^{\frac{\mathrm{i}}{2} m_{-}(x)} \int_{-\infty}^{x} - \frac{\mathrm{i}}{2} (q^{PT})_{y}(y) \Psi_{21}^{-}(y;z) \mathrm{e}^{-\frac{\mathrm{i}}{2} m_{-}(y)}\, dy.
		\end{equation}
		Also, taking the limit $|z| \to \infty$ of $z (\Psi_{11}^{-}(x;z) - \mathrm{e}^{\frac{\mathrm{i}}{2} m_{-}(x)})$, we have
		\begin{equation}\label{eq:zcomponet1}
	\begin{array}{rl}
	\displaystyle	\lim_{|z| \to \infty} z (\Psi_{11}^{-}(x;z) - \mathrm{e}^{\frac{\mathrm{i}}{2} m_{-}(x)})
= & \displaystyle \!\!\!\! -\frac{1}{4} \mathrm{e}^{\frac{\mathrm{i}}{2} m_{-}(x)} \int_{-\infty}^{x} \left(  (q^{PT})_{y}(y) q_{yy}(y) + \frac{\mathrm{i}}{2} ((q^{PT})_{y})^2(y) q_y^2(y) \right) dy \vspace{0.1in}\\
 = &\!\!\!\! s_1^{-}(x).
 \end{array}
		\end{equation}
		In the end, we get the limits of $\Psi_1^{-}$ in \eqref{eq:zlimits}. \hfill
	\end{proof}\par
	The next proposition plays a very important role in our subsequent estimates. Its proof is similar to the argument in \cite{Pelinovsky2018}, which we just modify the notation, so we will not elaborate on it here.
	
\begin{proposition}\label{prop: Estimate}
		If $g \in L^{2}(\mathbb{R})$, then
		\begin{equation}\label{eq:L2control}
			\sup_{x \in \mathbb{R}} \left\| \int_{-\infty}^{x} \mathrm{e}^{-2\mathrm{i}z(x-y)} g(y) \, dy \right\|_{L^2_z(\mathbb{R})} \le \sqrt{\pi} \|g\|_{L^2}.
		\end{equation}
	If $g \in H^{1}(\mathbb{R})$, then
		\begin{equation}\label{eq:partialx L2control}
			\sup_{x \in \mathbb{R}} \left\| g(x)-2\mathrm{i}z \int_{-\infty}^{x} \mathrm{e}^{-2\mathrm{i}z(x-y)} g(y) \, dy  \right\|_{L^2_z(\mathbb{R})} \le \sqrt{\pi} \|g_x\|_{L^2}.
		\end{equation}
		Moreover, if $g \in L^{2,1}(\mathbb{R})$, then
		\begin{equation}\label{eq:L2,1control}
			\sup_{x \in (-\infty, x_0)} \left\| \langle x \rangle \int_{-\infty}^{x} \mathrm{e}^{-2\mathrm{i}z(x-y)} g(y) \, dy \right\|_{L^2_z(\mathbb{R})} \le \sqrt{\pi} \|g\|_{L^{2,1}(-\infty, x_0)}
		\end{equation}
for every $x_0 \in \mathbb{R}^-$.
	\end{proposition}\par
In fact, Proposition \ref{prop: Estimate} can also allow us to analyze the smoothness of various remainder terms. It is worthy noting that $q^{PT}(x) = q(-x)$, which means if we integrate over the entire real axis, the integral values of the two functions will actually be the same in the sense of the norm. The following proposition gives the smoothness of the Jost solutions
	\begin{proposition}\label{prop: Smooth}
		If $q\in H^{2,1}(\mathbb{R})$, then for every $x\in \mathbb{R}^{\pm}$, we have
		\begin{equation}\label{eq:smooth1}
			\Psi^{\pm}(x;\cdot) - \mathrm{e}^{\frac{\mathrm{i}}{2} m_{\pm}(x) \sigma_3} \in H^{1}(\mathbb{R}).
		\end{equation}
		Moreover, if $q\in H^{3}(\mathbb{R}) \cap H^{2,1}(\mathbb{R})$, then for every $x \in \mathbb{R} $, we have
		\begin{subequations}\label{eq:smooth2}
			\begin{align}
				z(\Psi_1^{\pm}(x;z) - \mathrm{e}^{\frac{\mathrm{i}}{2} m_{\pm}(x)} e_1) &- (s_1^{\pm}(x) e_1 + s_2^{\pm}(x) e_2) \in L_{z}^2(\mathbb{R}), \\
				z(\Psi_2^{\pm}(x;z) - \mathrm{e}^{-\frac{\mathrm{i}}{2} m_{\pm}(x)} e_2) &- (\nu_{1}^{\pm}(x) e_1 + \nu_{2}^{\pm}(x) e_2) \in L_{z}^2 (\mathbb{R}).
			\end{align}
		\end{subequations}
	\end{proposition}
	\begin{proof}
		Again, we prove the statement for the Jost solution $\Psi_1^{-}$. The proof for other Jost solutions is analogous. We rewrite equation \eqref{eq:Psi1} in the equivalent form
		\begin{equation}\label{eq:difference}
			(I-K)(\Psi_1^{-} - \mathrm{e}^{\frac{\mathrm{i}}{2} m_{-}(x)} e_1) = h e_2,
		\end{equation}
		where
		\begin{equation}
			h(x;z) = \int_{-\infty}^{x} \mathrm{e}^{-2 \mathrm{i} z (x-y)} w(y)\, dy, \quad
			w(x) = \partial_x \left( q_x(x) \mathrm{e}^{\frac{\mathrm{i}}{2} m_{-}(x)} \right).
		\end{equation}
		If $q\in H^{2,1}(\mathbb{R})$, then $w\in L^{2,1}(\mathbb{R})$. By the bounds \eqref{eq:L2control} and \eqref{eq:L2,1control}, we have $h(x;z) \in L_{x}^{\infty}(\mathbb{R} ; L_{z}^{2}(\mathbb{R}))$ and for every $x_0 \in \mathbb{R}^{-}$, the following bound is satisfied:
		\begin{equation}\label{eq:bound2}
			\begin{split}
				\sup_{x \in (-\infty, x_0)} \| \langle x \rangle h(x;z) \|_{L_{z}^{2}(\mathbb{R})} &\le \mathrm{e}^{\frac{\| q_x \|_{L^{2}}^2}{2}} \sqrt{\pi}\left( \|q_{xx}\|_{L^{2,1}} + \frac{1}{2} \| q_x^2 (q^{PT})_{x} \|_{L^{2,1}}
				\right) \\
				&\le C \left( \| q \|_{H^{2,1}} + \| q \|_{H^{2,1}}^3
				\right),
			\end{split}
		\end{equation}
		where $C$ is a positive constant and we use the Sobolev embedding to get the final inequality.\par
		By using estimates similar to those in the derivation of the bound \eqref{eq:converge} we find that for every $g(x;z)\in L_x^{\infty}(\mathbb{R}; L_z^2(\mathbb{R}))$, we have
		\begin{equation}
			\| K^n g \|_{L_{x}^{\infty} L_{z}^2} \le \dfrac{1}{n!} \|\widehat{Q}\|_{L^1}^{n} \| g(x;z) \|_{L_{x}^{\infty} L_z^{2}}.
		\end{equation}
		Therefore, the operator $I-K$ is invertible on the space $L_x^{\infty}(\mathbb{R}; L_z^2(\mathbb{R}))$, and a bound on the inverse operator is given by
		\begin{equation}\label{eq:inverse operator}
			\| (I - K)^{-1} \|_{L_x^{\infty} L_z^2(\mathbb{R}) \to L_x^{\infty} L_z^2(\mathbb{R})} \le \mathrm{e}^{\| \widehat{Q} \|_{L^1}}.
		\end{equation}
		Moreover, the same estimate \eqref{eq:inverse operator} can be obtained in the norm $L_x^{\infty}((-\infty, x_0); L_z^2(\mathbb{R}))$ for every $x_0 \in \mathbb{R}^{-}$. By using \eqref{eq:difference}, \eqref{eq:bound2} and \eqref{eq:inverse operator}, we have the estimate
		\begin{equation}
			\sup_{x \in (-\infty, x_0)} \| \langle x \rangle \left( \Psi_1^{-}(x;z) -\mathrm{e}^{\frac{\mathrm{i}}{2} m_{-}(x)} e_1
			\right) \|_{L_{z}^{2}} \le C \mathrm{e}^{\| \widehat{Q} \|_{L^1}} \left( \| q \|_{H^{2,1}} + \| q \|_{H^{2,1}}^3
			\right).
		\end{equation}\par
	To end the proof of \eqref{eq:smooth2}, we need to show $\partial_z \Psi_1^{-}(x;z) \in L_x^{\infty}((-\infty, x_0); L_z^2(\mathbb{R}))$ for every $x_0 \in \mathbb{R}$. We
		differentiate the integral equation \eqref{eq:Psi1} in $z$ and introduce the vector
		$$ v(x;z) := \left[ \partial_z \Psi_{11}^{-}(x;z), \,  \partial_z \Psi_{21}^{-}(x;z) + 2ix \Psi_{21}^{-}(x;z)
		\right],$$
		Thus, we obtain from \eqref{eq:Psi1} that
		\begin{equation}\label{eq:partial difference}
			(I - K) v = \mu_1 e_1 + \mu_2 e_2 + \mu_3 e_2,
		\end{equation}
		where
		\begin{align*}
			\mu_1(x;z) &= -\int_{-\infty}^{x} y (q^{PT})_{y}(y) \Psi_{21}^{-}(y;z)\, dy, \\
			\mu_2(x;z) &= \int_{-\infty}^{x} y \mathrm{e}^{-2\mathrm{i}z(x-y)} \left(2\mathrm{i}q_{yy}(y) - q_y^2(y) (q^{PT})_{y}(y) \right) \left( \Psi_{11}^{-}(y;z) - \mathrm{e}^{\frac{\mathrm{i}}{2}m_{-}(y)} \right) dy, \\
			\mu_3(x;z) &= \int_{-\infty}^{x} y \mathrm{e}^{-2\mathrm{i}z(x-y)} \left(2\mathrm{i}q_{yy}(y) - q_y^2(y) (q^{PT})_{y}(y) \right) \mathrm{e}^{\frac{\mathrm{i}}{2}m_{-}(y)}\, dy.
		\end{align*}
		For every $x_0 \in \mathbb{R}^{-}$, applying the H\"{o}lder's inequality to $\mu_1$ and using the estimate \eqref{eq:L2control} for $\mu_2$ and $\mu_3$, we obtain the following bounds:
	\begin{align*}
	\begin{array}{rl}
	\d\sup_{x \in (-\infty, x_0)} \|\mu_1(x;\cdot) \|_{L_z^2(\mathbb{R})}\!\!\! & \le \d \| (q^{PT})_x \|_{L^1} \sup_{x \in (-\infty, x_0)} \| \langle x \rangle \Psi_{21}^{-}(x;z) \|_{L_z^2(\mathbb{R})}, \vspace{0.1in} \\
		\d	\sup_{x \in (-\infty, x_0)} \|\mu_2(x;\cdot) \|_{L_z^2(\mathbb{R})} \!\!\! & \d \le \left(2\|q_{xx}\|_{L^1} + \| q_x \|_{L^{\infty}} \| q_x \|_{L^2}\|(q^{PT})_x \|_{L^2} \right) \v\\
\!\!\! & \qquad \d \times \sup_{x \in (-\infty, x_0)} \left\| \langle x \rangle \left(\Psi_{11}^{-}(x;z) - \mathrm{e}^{\frac{\mathrm{i}}{2}m_{-}(x)}\right) \right\|_{L_z^2(\mathbb{R})}, \v\\
			\d \sup_{x \in (-\infty, x_0)} \| \mu_3(x;\cdot) \|_{L_z^2(\mathbb{R})} \!\!\! &\le  \d
\sqrt{\pi}\mathrm{e}^{\frac{1}{2}\| q_x \|_{L^{2}}\| (q^{PT})_x \|_{L^{2}}} \left( 2\|q_{xx}\|_{L^{2,1}} +\|q_x^2 (q^{PT})_{x}\|_{L^{2,1}} \right).
	\end{array}
	\end{align*}
		Because of the estimate \eqref{eq:bound2}, the first two bounds are finite. By the invertibility of $(I-K)$ in \eqref{eq:inverse operator} via the uniformly convergent Neumann series, we summarize that $v(x;z) \in L_x^{\infty}((-\infty, x_0); L_z^2(\mathbb{R}))$ for every $x_0 \in \mathbb{R}^{-}$. This implies $\partial_z \Psi_1^{-}(x;\cdot) \in L_z^2(\mathbb{R})$, yielding the proof of \eqref{eq:smooth1}.
		
		To prove \eqref{eq:smooth2}, we note that $q \in H^3(\mathbb{R}) \cap H^{2,1}(\mathbb{R})$ guarantees that $w(x) \in H^1(\mathbb{R})$. We can obtain that
		\begin{equation}\label{eq: zlimits diff}
			(I - K)\left[z \left(\Psi_1^{-} - \mathrm{e}^{\frac{\mathrm{i}}{2} m_{-}(x)} e_1
			\right) - \left(s_{1}^{-} e_1 +s_{2}^{-} e_2
			\right) \right] = \tilde{\mu} e_2,
		\end{equation}
		with
		$$
\begin{array}{rl}
\tilde{\mu} (x;z) \!\!\! &= \d z \int_{-\infty}^{x} \mathrm{e}^{-2 \mathrm{i} z (x-y)} w(y)\, dy - s_2^{-}(x) \vspace{0.1in}\\
 \!\!\! & \d \qquad + \frac{\mathrm{i}}{2} \int_{-\infty}^x \mathrm{e}^{-2\mathrm{i}z(x-y)} \left[ \left(q_y^2 (q^{PT})_{y} - 2\mathrm{i}q_{yy}\right) s_1^-(y) - q_y (q^{PT})_{y} s_2^-(y) \right] dy.
 \end{array}$$
		By using bounds \eqref{eq:L2control} and \eqref{eq:partialx L2control}, we have $\tilde{\mu}(x;z) \in L_{x}^{\infty}(\mathbb{R}; L_{z}^{2}(\mathbb{R}))$ if $w \in H^{1}(\mathbb{R})$. Inverting $(I - K)$ on $L_{x}^{\infty}(\mathbb{R}; L_{z}^{2}(\mathbb{R}))$, we finally obtain \eqref{eq:smooth2} for $\Psi_1^{-}$. \hfill
	\end{proof} \par
	
To deduce the Lipschitz continuity of the scattering data, we need to analyze the Lipschitz continuity of the Jost solutions. The following result shows that the map
	\begin{equation}\label{map:L-continue1}
		H^{2,1}(\mathbb{R}) \ni q \to \left(\Psi^{\pm}(x;z) - \mathrm{e}^{\frac{\mathrm{i}}{2} m_{\pm}(x) \sigma_3}
		\right) \in L_{x}^{\infty}(\mathbb{R}^{\pm}; H_{z}^{1}(\mathbb{R}))
	\end{equation}
	is Lipschitz continuous. Moreover, the $\mathcal{O} (z^{-1})$ term also has the continuity.
	
\begin{corollary}\label{corollay:Jost L-C}
		Let $q, \tilde{q} \in H^{2,1}(\mathbb{R})$ satisfy $\| u \|_{H^{2,1}}, \| \tilde{u} \|_{H^{2,1}} \le \delta$ for some $\delta > 0$. Then, there is a positive $\delta$-independent constant $C(\delta)$ such that for every $x \in \mathbb{R}^{\pm}$, we have
		\begin{equation}\label{eq:L c control1}
			\| \Psi_1^{\pm}(x;\cdot) - \mathrm{e}^{\frac{\mathrm{i}}{2} m_{\pm}(x)} e_1 - \tilde{\Psi}_1^{\pm} (x;\cdot) + \mathrm{e}^{\frac{i}{2} \tilde{m}_{\pm}(x)} e_1 \|_{H^{1}} \le C(\delta) \| q - \tilde{q} \|_{H^{2,1}}.
		\end{equation}
		and
		\begin{equation}\label{eq:L c control2}
			\| \Psi_2^{\pm}(x;\cdot) - \mathrm{e}^{-\frac{\mathrm{i}}{2} m_{\pm}(x)} e_2 - \tilde{\Psi}_2^{\pm} (x;\cdot) + \mathrm{e}^{-\frac{i}{2} \tilde{m}_{\pm}(x)} e_2 \|_{H^{1}} \le C(\delta) \| q - \tilde{q} \|_{H^{2,1}}.
		\end{equation}
		Furthermore, if $q, \tilde{q} \in H^{3}(\mathbb{R}) \cap H^{2,1}(\mathbb{R})$ satisfy $\| q \|_{H^{3} \cap H^{2,1}}, \| \tilde{q} \|_{H^{3} \cap H^{2,1}} \le \delta$, then for every $x\in \mathbb{R}$, there is a positive constant such that
		\begin{equation}\label{eq:z-L-c control}
			\| \hat{\Psi}_1^{\pm}(x;\cdot) - \hat{\tilde{\Psi}}_1^{\pm}(x;\cdot) \|_{L^2} + \| \hat{\Psi}_2^{\pm}(x;\cdot) - \hat{\tilde{\Psi}}_2^{\pm}(x;\cdot) \|_{L^2} \le C(\delta) \| q - \tilde{q} \|_{H^{3} \cap H^{2,1}}.
		\end{equation}
		where
		\begin{align}
		 \hat{\Psi}_1^{\pm} (x;z) &:= z(\Psi_1^{\pm}(x;z) - \mathrm{e}^{\frac{\mathrm{i}}{2} m_{\pm}(x)} e_1) - (s_1^{\pm}(x) e_1 + s_2^{\pm}(x) e_2), \v \\
		 \hat{\Psi}_2^{\pm} (x;z) &:= z(\Psi_2^{\pm}(x;z) - \mathrm{e}^{-\frac{\mathrm{i}}{2} m_{\pm}(x)} e_2) - (\nu_{1}^{\pm}(x) e_1 + \nu_{2}^{\pm}(x) e_2).
	    \end{align}
	\end{corollary}
	\begin{proof}
		Again, we prove the statement for the Jost solution $\Psi_1^{-}$. The proof for other Jost solutions is analogous.\par
		For every $x \in \mathbb{R}$, utilizing the nonlocal symmetry $q^{PT}(x)=q(-x)$ and $\tilde{q}^{PT}(x)=\tilde{q}(-x)$, we can estimate the difference of the exponential phase factors. Since $\|(q^{PT})_{x}\|_{L^2} = \|q_x\|_{L^2}$, we have:
		\begin{equation}\label{eq:term1}
		    \begin{split}
		    	|m_{-}(x) - \tilde{m}_{-}(x)| &\le \int_{-\infty}^{x} |q_y (q^{PT})_{y} - \tilde{q}_y (\tilde{q}^{PT})_y|\, dy \\
		    	&\le \int_{-\infty}^{x} |q_y((q^{PT})_{y} - \tilde{q}^{PT}_y)| dy + \int_{-\infty}^{x} |(q_y - \tilde{q}_y)\tilde{q}^{PT}_y|\, dy \\
		    	&\le \|q_x\|_{L^2} \|(q^{PT})_{x} - (\tilde{q}^{PT})_x\|_{L^2} + \|q_x - \tilde{q}_x\|_{L^2} \|(\tilde{q}^{PT})_x\|_{L^2} \le 2\delta \|q - \tilde{q}\|_{H^{2,1}}.
			\end{split}
		\end{equation}
		Using the basic inequality $|\mathrm{e}^{A} - \mathrm{e}^{B}| \le |A - B| \max(|\mathrm{e}^A|, |\mathrm{e}^B|)$ for complex exponents, we can bound the difference of the exponential terms by $C_1(\delta)\|q - \tilde{q}\|_{H^{2,1}}$.\par Using the integral equation \eqref{eq:difference}, we can express the difference of the Jost solutions via the resolvent identity:
		\begin{equation}\label{eq:Lip diff}
			\begin{split}
				(\Psi_1^{-} - \mathrm{e}^{\frac{\mathrm{i}}{2} m_{-}(x)} e_1) - (\tilde{\Psi}_1^{-} - \mathrm{e}^{\frac{\mathrm{i}}{2} \tilde{m}_{-}(x)} e_1) &= (I-K)^{-1}(\mu e_2 - \tilde{\mu} e_2)
				+ (I-K)^{-1}(K - \tilde{K})(I-\tilde{K})^{-1} \tilde{\mu} e_2,
			\end{split}
		\end{equation}
		where $\tilde{K}$ and $\tilde{\mu}$ denote the same quantities as $K$ and $\mu$ but with $q$ replaced by $\tilde{q}$.\par
		To estimate the first term on the right-hand side of \eqref{eq:Lip diff}, we analyze the difference $h(x;z) - \tilde{h}(x;z) = \int_{-\infty}^{x} \mathrm{e}^{-2\mathrm{i}z(x-y)} (w(y) - \tilde{w}(y)) dy$. The term $w(y) - \tilde{w}(y)$ is given by:
		\begin{equation}
			w(y) - \tilde{w}(y) = \left( q_{yy} + \frac{\mathrm{i}}{2} q_y^2 (q^{PT})_{y} \right) \mathrm{e}^{\frac{\mathrm{i}}{2} m_{-}(y)} - \left( \tilde{q}_{yy} + \frac{\mathrm{i}}{2} \tilde{q}_y^2 (\tilde{q}^{PT})_y \right) \mathrm{e}^{\frac{\mathrm{i}}{2} \tilde{m}_{-}(y)}.
		\end{equation}
		Utilizing \eqref{eq:term1} and the Sobolev embedding, we can deduce that $\|w - \tilde{w}\|_{L^{2,1}} \le C_2(\delta) \|q - \tilde{q}\|_{H^{2,1}}$. Then, applying the estimate \eqref{eq:L2,1control}, we obtain for every $x_0 \in \mathbb{R}^{-}$:
		\begin{equation}
			\sup_{x \in (-\infty, x_0)} \|\langle x \rangle [\mu(x;\cdot) - \tilde{\mu}(x;\cdot)]\|_{L_z^2(\mathbb{R})} \le \sqrt{\pi}\, C_2(\delta) \|q - \tilde{q}\|_{H^{2,1}}.
		\end{equation}\par
		For the second term in \eqref{eq:Lip diff}, we note that $K$ is a Lipschitz continuous operator from $L_x^\infty(\mathbb{R}; L_z^2(\mathbb{R}))$ to itself. The difference of the integral kernels is bounded by $\|\tilde{Q} - \tilde{\tilde{Q}}\|_{L^1} \le C_3(\delta)\|q - \tilde{q}\|_{H^{2,1}}$. Thus, for any $f \in L_x^\infty(\mathbb{R}; L_z^2(\mathbb{R}))$, we have:
		\begin{equation}
			\|(K - \tilde{K})f\|_{L_x^\infty L_z^2} \le C_3(\delta) \|q - \tilde{q}\|_{H^{2,1}} \|f\|_{L_x^\infty L_z^2}.
		\end{equation}\par
		Combining these estimates with the uniform boundedness of the inverse operators $(I-K)^{-1}$ and $(I-\tilde{K})^{-1}$ from \eqref{eq:inverse operator}, we derive that:
		\begin{equation}
			\sup_{x \in (-\infty, x_0)} \|\langle x \rangle (\Psi_1^{-}(x;\cdot) - \mathrm{e}^{\frac{\mathrm{i}}{2} m_{-}(x)} e_1 - \tilde{\Psi}_1^{-}(x;\cdot) + \mathrm{e}^{\frac{\mathrm{i}}{2} \tilde{m}_{-}(x)} e_1)\|_{L_z^2(\mathbb{R})} \le C_4(\delta) \|q - \tilde{q}\|_{H^{2,1}}.
		\end{equation}\par
		The bound for the $z$-derivative in $H^1_z(\mathbb{R})$ follows from applying the identical difference scheme to the differentiated integral equation \eqref{eq:partial difference}. Combining these results yields the Lipschitz continuity in \eqref{eq:L c control1}.\par Finally, the proof for the bound \eqref{eq:z-L-c control} involving the asymptotic terms $\hat{\Psi}_1^{\pm}$ and $\hat{\Psi}_2^{\pm}$ is completed by recalling the same difference analysis on the integral equations to \eqref{eq:partial difference} and \eqref{eq: zlimits diff} previously. \hfill
	\end{proof}

	\subsection{Scattering coefficients on the $\lambda$-plane}
	
 In this subsection, we introduce the scattering data through the original Jost solutions
	$\psi^{\pm}(x;\lambda)$ associated with \eqref{eq: psi-Volterra}. This is the same point of
	view as in the inverse scattering transform on the $\lambda$-plane. The transformation
	\eqref{eq:lambda to z} will only be used later as an auxiliary transformation to obtain a
	Zakharov--Shabat type problem, but it is not used to define the scattering coefficients.\par
	Throughout this subsection, the time variable $t$ is fixed and suppressed. We denote the continuous spectra by
	\begin{equation}\label{eq:lambda-continuous-spectrum}
		\Sigma_{\lambda}:=\{\lambda\in\mathbb{C}:\operatorname{Im}\lambda^2=0\}
		=\mathbb{R}\cup \mathrm{i}\mathbb{R},
	\end{equation}
	and set (see Fig.~\ref{fig2}(left))
	\begin{equation}\label{eq:lambda-domains}
		D_{\pm}:=\{\lambda\in\mathbb{C}: \pm\operatorname{Im}\lambda^2>0\}.
	\end{equation}\par
	From the estimate of  \cite{Cheng2022}, we can get the following proposition, which we omit the proof:
	\begin{proposition}\label{prop:original analytic}
		Let $q \in H^{3}(\mathbb{R}) \cap H^{2,1}(\mathbb{R})$, denote the Jost function
		$\psi^{\pm}(x;\lambda)=		\bigl(\psi_1^{\pm}(x;\lambda),\psi_2^{\pm}(x;\lambda)\bigr),$
 where $\psi_j^{\pm}(x;\lambda)\, (j=1,2)$ denotes the column vectors. Then the integral equation \eqref{eq: psi-Volterra} admits the unique solutions $\psi^{\pm}(x,t;\lambda)$. Moreover, $\psi_1^{\pm}(x;\lambda),\, \psi_2^{\mp}(x;\lambda)$ are analytic in $D_{\pm}$.
	\end{proposition}\par
	Analytic properties of the Jost solutions $\psi^{\pm}$ for every $x \in \mathbb{R}$ are summarized in the following result. The result is a corollary of Proposition \ref{prop: Smooth}.
	
\begin{corollary}\label{corollary: original smooth}
		If $q_0(x) \in H^{3} (\mathbb{R}) \cap H^{2,1}(\mathbb{R})$, then for every $x \in \mathbb{R}^{\pm}$, the Jost solutions $\psi^{\pm}$ have the following properties
		\begin{equation}\label{eq:orginal-Jost estimate1}
			\psi_{11}^{\pm} - \mathrm{e}^{\frac{\mathrm{i}}{2} m_{\pm}(x)}, \ 2 \lambda \psi_{21}^{\pm} - q_x \mathrm{e}^{\frac{\mathrm{i}}{2} m_{\pm}(x)}, \ \lambda^{-1} \psi_{21}^{\pm} \in H_{\lambda}^{1}(\mathbb{R}),
		\end{equation}
		and
		\begin{equation}\label{eq:orginal-Jost estimate2}
			\psi_{22}^{\pm} - \mathrm{e}^{-\frac{\mathrm{i}}{2} m_{\pm}(x)}, \ \lambda^{-1} \psi_{12}^{\pm} \in H_{\lambda}^{1}(\mathbb{R}).
		\end{equation}
	\end{corollary}
	\begin{proof}
		We still only prove the statement for $\psi_1^{\pm}$.  The proof for other Jost solutions is analogous. From the first column of the transformation \eqref{eq:lambda to z}, we have
		\begin{equation}
			\begin{aligned}
				\Psi_{11}^{\pm}(x;z) &= \psi_{11}^{\pm}(x;\lambda), \\
				\Psi_{21}^{\pm}(x;z) &= q_x \psi_{11}^{\pm}(x;\lambda) - 2 \lambda \psi_{21}^{\pm}(x;\lambda).
			\end{aligned}
		\end{equation}
		Recalling the result \eqref{eq:smooth1}, we can obtain
		$$ \psi_{11}^{\pm}(x;\lambda) - \mathrm{e}^{\frac{\mathrm{i}}{2} m_{\pm}(x)} = \Psi_{11}^{\pm}(x;z) - \mathrm{e}^{\frac{\mathrm{i}}{2} m_{\pm}(x)} \in H_{z}^{1}(\mathbb{R}).$$
		Noting that $q_x(x) \in L^{\infty}(\mathbb{R})$ and $\Psi_{21}^{\pm}(x;z) \in H_{z}^{1}(\mathbb{R})$, we find
		$$ 2 \lambda \psi_{21}^{\pm}(x;\lambda) - q_x \mathrm{e}^{\frac{\mathrm{i}}{2} m_{\pm}(x)} = -\Psi_{21}^{\pm}(x;z) + q_x (\psi_{11}^{\pm}(x;\lambda) - \mathrm{e}^{\frac{\mathrm{i}}{2} m_{\pm}(x)}) \in H_{z}^{1}(\mathbb{R}).$$\par
		To prove $ \lambda^{-1}\psi_{21}^{\pm} \in H_{z}^{1}(\mathbb{R})$, we write explicitly from the integral equation \eqref{eq: psi-Volterra}
		\begin{equation}
			\begin{split}
				\lambda^{-1} \psi_{21}^{\pm}(x;\lambda) &=  \mathrm{i} \int_{\pm \infty}^{x} \mathrm{e}^{-2 \mathrm{i} z (x-y)} q_y \mathrm{e}^{\frac{\mathrm{i}}{2} m_{\pm}(y)} \, dy
				+ \int_{\pm \infty}^{x} \mathrm{e}^{-2 \mathrm{i} z (x-y)} q_y (\psi_{11}^{\pm} - \mathrm{e}^{\frac{\mathrm{i}}{2} m_{\pm}(y)})\, dy.
			\end{split}
		\end{equation}
		Due to $q \in H^{3}(\mathbb{R}) \cap H^{2,1}(\mathbb{R})$ and $\psi_{11}^{\pm}(x;\lambda) - \mathrm{e}^{\frac{\mathrm{i}}{2} m_{\pm}(x)} \in H_{z}^{1}(\mathbb{R})$, we can use Proposition \ref{prop: Estimate} in the same way as it was used in the proof of Proposition \ref{prop: Smooth}. Then, we can obtain $\lambda^{-1} \psi_{21}^{\pm}(x;\lambda) \in H_{z}^{1}(\mathbb{R})$ for every $x\in \mathbb{R}^{\pm}$. The proof of $\psi_2(x;\lambda)$ is similar. \hfill
	\end{proof}\par
	Since $\psi^{\pm}(x;\lambda) \mathrm{e}^{\mathrm{i}\lambda^2x\sigma_3}$ are two fundamental matrix solutions of the original spectral problem, they are linearly dependent for $\lambda\in\Sigma_{\lambda}$. Hence there exists a scattering matrix
	$S(\lambda)$ such that
	\begin{equation}\label{eq:psi-scattering-relation}
		\psi^{-}(x;\lambda)
		=
		\psi^{+}(x;\lambda)
		\mathrm{e}^{\mathrm{i}\lambda^2x\,\hat{\sigma}_3}
		S(\lambda),  \quad S(\lambda)=
		\begin{pmatrix}
			a(\lambda) & \breve b(\lambda) \v\\
			b(\lambda) & \breve a(\lambda)
		\end{pmatrix},
		\qquad \lambda\in\Sigma_{\lambda}.
	\end{equation}
	Equivalently,
	\begin{equation}\label{eq:S-lambda-def}
		S(\lambda)
		=
		\mathrm{e}^{-\mathrm{i}\lambda^2x\,\hat{\sigma}_3}
		\bigl(\psi^{+}(x;\lambda)\bigr)^{-1}\psi^{-}(x;\lambda),
	\end{equation}
	and the right-hand side is independent of $x$. In particular, by taking $x=0$, we have
	$$
	S(\lambda)=\bigl(\psi^{+}(0;\lambda)\bigr)^{-1}\psi^{-}(0;\lambda).
	$$
	
	Since $\operatorname{tr}(Q_x)=0$, it follows from \eqref{eq:New Laxpair} and the normalization
	$\psi^{\pm}(x;\lambda)\to I$ as $x\to\pm\infty$ that
	\bee
	\det \psi^{\pm}(x;\lambda)=1.
	\ene
	Therefore
	\begin{equation}\label{eq:det-S-lambda}
		\det S(\lambda)=a(\lambda)\breve a(\lambda)-b(\lambda)\breve b(\lambda)=1.
	\end{equation}
	
	Using \eqref{eq:psi-scattering-relation}, the scattering coefficients can be written in terms
	of Wronskians as
	\begin{align}
		a(\lambda)
		&=
		\mathrm{Wr} \bigl(\psi_1^{-}(x;\lambda), \psi_2^{+}(x;\lambda)\bigr) = \mathrm{Wr} \bigl(\psi_1^{-}(0;\lambda), \psi_2^{+}(0;\lambda)\bigr),
		\label{eq:a-lambda-def}\\
		\breve a(\lambda)
		&=
		\mathrm{Wr} \bigl(\psi_1^{+}(x;\lambda), \psi_2^{-}(x;\lambda)\bigr) = \mathrm{Wr} \bigl(\psi_1^{+}(0;\lambda), \psi_2^{-}(0;\lambda)\bigr),
		\label{eq:abreve-lambda-def}\\
		b(\lambda)
		&=
		\mathrm{e}^{2 \mathrm{i} \lambda^2 x}
		\mathrm{Wr} \bigl(\psi_1^{+}(x;\lambda), \psi_1^{-}(x;\lambda)\bigr) = \mathrm{Wr} \bigl(\psi_1^{+}(0;\lambda), \psi_1^{-}(0;\lambda)\bigr),
		\label{eq:b-lambda-def}\\
		\breve b(\lambda)
		&=
		\mathrm{e}^{-2 \mathrm{i} \lambda^2 x}
		\mathrm{Wr} \bigl(\psi_2^{-}(x;\lambda), \psi_2^{+}(x;\lambda)\bigr) = \mathrm{Wr} \bigl(\psi_2^{-}(0;\lambda), \psi_2^{+}(0;\lambda)\bigr).
		\label{eq:bb-lambda-def}
	\end{align}\par

We define the reflection coefficients by
	\begin{equation}\label{refl}
		r(\lambda) := \dfrac{b(\lambda)}{a(\lambda)}, \quad
		\breve{r}(\lambda) := \dfrac{\breve{b}(\lambda)}{\breve{a}(\lambda)}, \quad \lambda \in \Sigma_{\lambda}.
	\end{equation}

	Now we note the symmetry on solutions to the first linear equation of \eqref{eq:New Laxpair}. By using the boundary conditions for the normalized Jost solutions, we obtain the following results
	\begin{equation}\label{eq:Jost-solution-Symmetry}
		\psi^{+}(x,t;\lambda) = \sigma_1 \psi^{-}(-x,-t;\lambda) \sigma_1, \quad \psi^{+}(x,t;\lambda) = \sigma_3 \psi^{-}(x,t;-\lambda) \sigma_3.
	\end{equation}
	where $\sigma_1$ is Pauli matrix. As usual, the direct scattering transform is based on $t = 0$, the scattering coefficients have the results
	\begin{equation}\label{eq:Scattering-matrix symmetry}
		S(\lambda;t) = \sigma_3 S(-\lambda;t) \sigma_3, \quad S(\lambda;t) = \sigma_1 S^{-1}(\lambda;-t) \sigma_1, \quad \lambda \in \Sigma_{\lambda}.
	\end{equation}
	which means that when $t = 0$, we have
	\begin{align}
		a(\lambda) &= a(-\lambda), \qquad \breve{a}(\lambda) = \breve{a}(-\lambda), \quad \lambda \in \Sigma_{\lambda},\label{eq:s-c symmetry1} \\
		b(\lambda) &= -b(-\lambda), \quad \breve{b}(\lambda) = -\breve{b}(-\lambda), \quad \breve{b}(\lambda) = -b(\lambda), \quad \lambda \in \Sigma_{\lambda},\\
 r(\lambda) &= -r(-\lambda), \quad  \breve{r}(\lambda)=-\breve{r}(-\lambda),\quad \lambda \in \Sigma_{\lambda}.
\label{eq:s-c symmetry2}
	\end{align}

We first consider the properties of them about $|\lambda| \to \infty$ and we have the following proposition:

	\begin{proposition}\label{prop: Scattering-coefficients Smooth}
		If $q \in H^{3}(\mathbb{R}) \cap H^{2,1}(\mathbb{R})$, then the function $a(\lambda)$ and $\breve{a}(\lambda)$ are analytically in $\mathbb{C}^{-}$ and $\mathbb{C}^{+}$ with respect to $z = \lambda^2$, respectively. In addition,
		\begin{equation}
			a(\lambda) - a_{\infty}, \ \breve{a}(\lambda) - \breve{a}_{\infty},\ \lambda b(\lambda), \ \lambda^{-1} b(\lambda), \ \lambda\breve{b}(\lambda), \ \lambda^{-1} \breve{b}(\lambda) \in H_{z}^{1}(\mathbb{R}),
		\end{equation}
		with
	\bee a_{\infty} = \exp\left(\frac{\mathrm{i}}{2}
		\int_{\mathbb{R}} q_y(y)(q^{PT})_{y}(y) dy\right),\qquad
	 \breve{a}_{\infty} = \exp\left(-
			\frac{\mathrm{i}}{2}
			\int_{\mathbb{R}} q_y(y)(q^{PT})_{y}(y) dy\right),
\ene
		and
		\begin{equation}
			\ \lambda b(\lambda), \ \lambda^{-1} b(\lambda), \ \lambda\breve{b}(\lambda), \ \lambda^{-1} \breve{b}(\lambda) \in L_{z}^{2,1}(\mathbb{R}).
		\end{equation}
	\end{proposition}
	\begin{proof}
		Analytical proofs can be completed using the Wronskians determinant representation. And we only prove the results for $a(\lambda)$ and $b(\lambda)$. The proof of others is similar. In fact, the integration of $a(\lambda)$ is as follows:
		\begin{equation}
			\begin{split}
				a(\lambda) &= 1 + \mathrm{i} \lambda \int_{\mathbb{R}} (q^{PT})_{y}(y) \psi_{21}^{-}(y;\lambda)\, dy \\
				&= 1 + \dfrac{\mathrm{i}}{2}
				\int_{\mathbb R}
				q_y(y) (q^{PT})_{y}(y) \psi_{11}^{-}(y;\lambda) \, dy
				-
				\frac{\mathrm{i}}{2}
				\int_{\mathbb R}
				(q^{PT})_{y}(y)\Psi_{21}^{-}(y;z)\, dy,
			\end{split}
		\end{equation}
		which means
		\begin{equation}
			\lim_{|\lambda| \to \infty} a(\lambda) = a_{\infty}
         = \exp\left(\frac{\mathrm{i}}{2}
				\int_{\mathbb{R}} q_y(y)(q^{PT})_{y}(y) dy\right).
		\end{equation}
		To prove that $a(\lambda) - a_{\infty} \in H_{z}^{1}(\mathbb{R})$, we add and subtract the asymptotic limits of the Jost solutions obtained in Corollary \ref{corollary: original smooth}. Let $L_{11}^{\pm} = \mathrm{e}^{\frac{\mathrm{i}}{2} m_{\pm}(0)}$ and $L_{22}^{\pm} = \mathrm{e}^{-\frac{\mathrm{i}}{2} m_{\pm}(0)}$. We can factor $a(\lambda)$ as:
		\begin{equation}
			\begin{split}
				a(\lambda) - a_{\infty} &= (\psi_{11}^{-}(0;\lambda) - L_{11}^{-}) (\psi_{22}^{+}(0;\lambda) - L_{22}^{+}) + L_{11}^{-} (\psi_{22}^{+}(0;\lambda) - L_{22}^{+}) \v\\
				&\qquad + L_{22}^{+} (\psi_{11}^{-}(0;\lambda) - L_{11}^{-}) + L_{11}^{-} L_{22}^{+}
				- \psi_{21}^{-}(0;\lambda) \psi_{12}^{+}(0;\lambda).
			\end{split}
		\end{equation}
		Note that the constant term evaluates exactly to the asymptotic limit: $L_{11}^{-}L_{22}^{+} = \mathrm{e}^{\frac{\mathrm{i}}{2}(m_{-}(0) - m_{+}(0))} = a_{\infty}$. By Corollary \ref{corollary: original smooth}, functions like $\psi_{11}^{-} - L_{11}^{-}$, $\psi_{22}^{+} - L_{22}^{+}$, and $\lambda^{-1}\psi_{12}^{+}$ all belong to $H_{z}^{1}(\mathbb{R})$. Since $2\lambda\psi_{21}^{-} = (2\lambda\psi_{21}^{-} - q_x(0)L_{11}^{-}) + q_x(0)L_{11}^{-} \in H_{z}^{1}(\mathbb{R}) \oplus \mathbb{C}$, and $H_{z}^{1}(\mathbb{R})$ is a Banach algebra, their sums and products remain in $H_{z}^{1}(\mathbb{R})$. Thus, $a(\lambda) - a_{\infty} \in H_{z}^{1}(\mathbb{R})$.
		
		Similarly, we have
		\begin{equation}
			\lambda^{-1}b(\lambda) = \psi_{11}^{+}(0;\lambda)(\lambda^{-1}\psi_{21}^{-}(0;\lambda)) - (\lambda^{-1}\psi_{21}^{+}(0;\lambda))\psi_{11}^{-}(0;\lambda).
		\end{equation}
		Since $\lambda^{-1}\psi_{21}^{\pm} \in H_{z}^{1}(\mathbb{R})$ and $\psi_{11}^{\pm}$ differ from constants by an $H_{z}^{1}(\mathbb{R})$ function, we immediately obtain $\lambda^{-1}b(\lambda) \in H_{z}^{1}(\mathbb{R})$.
		
		Similarly, to establish the relation for $2\lambda b(\lambda)$ in terms of the components of $\Psi^{\pm}(0;z)$, we evaluate the Wronskian \eqref{eq:b-lambda-def} at $x=0$:
		\begin{equation}
			2\lambda b(\lambda) = \psi_{11}^{+}(0;\lambda)(2\lambda\psi_{21}^{-}(0;\lambda)) - (2\lambda\psi_{21}^{+}(0;\lambda))\psi_{11}^{-}(0;\lambda).
		\end{equation}
		Substituting the component relations from the transformation \eqref{eq:lambda to z}, namely $2\lambda\psi_{21}^{\pm}(0;\lambda) = q_x(0)\Psi_{11}^{\pm}(0;z) - \Psi_{21}^{\pm}(0;z)$, we obtain:
		\begin{equation}\label{eq: lambda b}
			\begin{split}
				2\lambda b(\lambda) &= \Psi_{11}^{+}(0;z) \left[ q_x(0)\Psi_{11}^{-}(0;z) - \Psi_{21}^{-}(0;z) \right] - \left[ q_x(0)\Psi_{11}^{+}(0;z) - \Psi_{21}^{+}(0;z) \right] \Psi_{11}^{-}(0;z) \v\\
				&= \Psi_{21}^{+}(0;z)\Psi_{11}^{-}(0;z) - \Psi_{11}^{+}(0;z)\Psi_{21}^{-}(0;z).
			\end{split}
		\end{equation}
		The remaining terms are combinations of $H_{z}^{1}(\mathbb{R})$ functions, proving that $\lambda b(\lambda) \in H_{z}^{1}(\mathbb{R})$.
		
		Since $z \lambda^{-1} b(\lambda) = \lambda b(\lambda) \in H_{z}^{1}(\mathbb{R})$, we have $ \lambda^{-1} b(\lambda) \in L_{z}^{2,1}(\mathbb{R})$. In addition, to prove $\lambda b(\lambda) \in  L_{z}^{2,1}(\mathbb{R})$, we multiply equation \eqref{eq: lambda b} by $z$ and write the resulting equation in the form
		\begin{equation}\label{eq:z-lambda-b}
			\begin{split}
				2\lambda z b(\lambda) &= -\Psi_{11}^{+}(0;z)\left(z\Psi_{21}^{-}(0;z) - s_{2}^{-}(0)\right) + \Psi_{11}^{-}(0;z)\left(z\Psi_{21}^{+}(0;z) - s_{2}^{+}(0)\right) \v\\
				&\qquad - s_{2}^{-}(0)\left(\Psi_{11}^{+}(0;z) - \mathrm{e}^{\frac{\mathrm{i}}{2} m_{+}(0)}\right) + s_{2}^{+}(0)\left(\Psi_{11}^{-}(0;z) - \mathrm{e}^{\frac{\mathrm{i}}{2} m_{-}(0)}\right),
			\end{split}
		\end{equation}
		where we have used the identity $s_{2}^{-}(0)\mathrm{e}^{\frac{\mathrm{i}}{2} m_{+}(0)} - s_{2}^{+}(0)\mathrm{e}^{\frac{\mathrm{i}}{2} m_{-}(0)} = 0$, which is obtained from the limits in \eqref{eq:zlimits}.
		By \eqref{eq:smooth1} and \eqref{eq:smooth2}, each term in the representation \eqref{eq:z-lambda-b} is in $L_{z}^{2}(\mathbb{R})$. Therefore, we derive $\lambda b(\lambda) \in L_{z}^{2,1}(\mathbb{R})$. \hfill
	\end{proof}
	
	Next we record the behavior of the scattering coefficients at $\lambda=0$. We seek the
	expansion of the original Jost functions in the form
	\begin{equation}\label{eq:psi-small-lambda-expansion}
		\psi^{\pm}(x;\lambda)
		=
		I+\lambda U_1(x)+\lambda^2U_2^{\pm}(x)+\mathcal{O}(\lambda^3),
		\qquad \lambda\to0.
	\end{equation}
	Substituting \eqref{eq:psi-small-lambda-expansion} into the $x$-part of
	\eqref{eq:New Laxpair} and comparing coefficients of the same powers of $\lambda$, we obtain
	$$
	U_{1,x}=\mathrm{i}Q_x,
	\qquad
	U_{2,x}^{\pm}=-Q_xQ.
	$$
	Using the normalization of $\psi^{\pm}$ at $\pm\infty$, we get
	\begin{equation}\label{eq:U-small-lambda}
		U_1(x)=\mathrm{i}Q(x),
		\qquad
		U_2^{\pm}(x)
		=
		-\int_{\pm\infty}^{x}Q_y(y)Q(y)\, dy.
	\end{equation}
	Therefore
	\begin{equation}\label{eq:psi-small-lambda}
		\psi^{\pm}(x;\lambda)
		=
		I+\mathrm{i}\lambda Q(x)
		-\lambda^2\int_{\pm\infty}^{x}Q_y(y)Q(y)\, dy
		+\mathcal{O}(\lambda^3),
		\qquad \lambda\to0.
	\end{equation}
	Taking $x=0$ in \eqref{eq:S-lambda-def}, we obtain
	\begin{equation}\label{eq:S-small-lambda}
		S(\lambda)
		=
		I
		-\lambda^2\int_{-\infty}^{+\infty}Q_y(y)Q(y)\, dy
		+\mathcal{O}(\lambda^3),
		\qquad \lambda\to0.
	\end{equation}
	Since
	$$
	Q_yQ=
	\begin{pmatrix}
		(q^{PT})_{y}q & 0 \v\\
		0 & q_y q^{PT}
	\end{pmatrix},
	$$
	we have
	\begin{align}
		a(\lambda)
		&=
		1-\lambda^2\int_{-\infty}^{+\infty}(q^{PT})_{y}(y)q(y)\, dy
		+\mathcal{O}(\lambda^4),
		\label{eq:a-small-lambda}\\
		\breve a(\lambda)
		&=
		1-\lambda^2\int_{-\infty}^{+\infty} q_y(y) q^{PT}(y)\, dy
		+\mathcal{O}(\lambda^4),
		\label{eq:abreve-small-lambda}\\
		b(\lambda)&=\mathcal{O}(\lambda^3),
		\qquad
		\breve b(\lambda)=\mathcal{O}(\lambda^3),
		\qquad \lambda\to0.
		\label{eq:b-small-lambda}
	\end{align}
	where we use \eqref{eq:s-c symmetry1}, \eqref{eq:s-c symmetry2}. In particular, the scattering coefficients are regular at $\lambda=0$.
	
	We also need the behavior of the scattering coefficients as $|\lambda|\to\infty$. Let
	\begin{equation}\label{eq:M-lambda-def}
		M:=\int_{-\infty}^{+\infty}q_y(y)(q^{PT})_{y}(y)\, dy.
	\end{equation}
	From the large-$\lambda$ asymptotics of the original Jost functions,
	$$
	\psi^{\pm}(x;\lambda)
	=
	\mathrm{e}^{\frac{\mathrm{i}}{2} m_{\pm}(x) \sigma_3}
	+\mathcal{O}(\lambda^{-1}),
	\qquad
	m_{\pm}(x):=\int_{\pm\infty}^{x}q_y(y)(q^{PT})_{y}(y)\, dy,
	$$
	we obtain
	\begin{align}
		a(\lambda)
		&=
		\mathrm{e}^{\frac{\mathrm{i}}{2} M}
		+\mathcal{O}(\lambda^{-1}),
		\qquad \lambda\in\overline{D_{-}},\quad |\lambda|\to\infty,
		\label{eq:a-large-lambda}\\
		\breve a(\lambda)
		&=
		\mathrm{e}^{-\frac{\mathrm{i}}{2} M}
		+\mathcal{O}(\lambda^{-1}),
		\qquad \lambda\in\overline{D_{+}},\quad |\lambda|\to\infty,
		\label{eq:abreve-large-lambda}\\
		b(\lambda)&=\mathcal{O}(\lambda^{-1}),
		\qquad
		\breve b(\lambda)=\mathcal{O}(\lambda^{-1}),
		\qquad \lambda\in\Sigma_{\lambda},\quad |\lambda|\to\infty.
		\label{eq:b-large-lambda}
	\end{align}
	
	For the convenience of the subsequent analysis, we provide here a sufficient condition that has no solitons and no spectral singularities. And in the subsequent discussions, we always assume no isolated vertices and no spectral singularities.

	\begin{proposition}\label{prop:z-minus-one-b} Let $q\in H^{3}(\mathbb{R})\cap H^{2,1}(\mathbb{R})$. Then
		\begin{equation}\label{eq:zminus1-lambdaminus1-b}
			z^{-1}\lambda^{-1}b(\lambda)\in H_z^{1}(\mathbb{R}),
			\qquad
			z^{-1}\lambda^{-1}\breve b(\lambda)\in H_z^{1}(\mathbb{R}).
		\end{equation}
	\end{proposition}
	
	\begin{proof}
		We only prove the statement for $b(\lambda)$, since the proof for
		$\breve b(\lambda)$ is analogous. From the Volterra integral equation
		\eqref{eq: psi-Volterra} and the scattering relation
		\eqref{eq:psi-scattering-relation}, the scattering coefficient
		$b(\lambda)$ admits the representation
		\begin{equation}\label{eq:b-integral-representation}
			\lambda^{-1}b(\lambda)
			=
			\mathrm{i}\int_{\mathbb{R}}
			q_y(y)\psi_{11}^{-}(y;\lambda)
			\mathrm{e}^{2\mathrm{i}zy}\,dy .
		\end{equation}
		Moreover, the first column of $\psi^{-}$ satisfies
		\begin{align}
			\psi_{11}^{-}(x;\lambda)
			&=
			1+\mathrm{i}\lambda\int_{-\infty}^{x}
			(q^{PT})_{y}(y)\psi_{21}^{-}(y;\lambda)\,dy,
			\label{eq:psi11-minus-eq}\\
			\psi_{21}^{-}(x;\lambda)
			&=
			\mathrm{i}\lambda\int_{-\infty}^{x}
			\mathrm{e}^{-2\mathrm{i}z(x-y)}
			q_y(y)\psi_{11}^{-}(y;\lambda)\,dy .
			\label{eq:psi21-minus-eq}
		\end{align}
		Substituting \eqref{eq:psi21-minus-eq} into
		\eqref{eq:psi11-minus-eq}, we obtain
		\begin{equation}\label{eq:psi11-minus-K}
			\psi_{11}^{-}(x;\lambda)-1
			=
			-z(\mathcal{T} \psi_{11}^{-})(x;z),
		\end{equation}
		where the integral operator $\mathcal{T}$ is defined by
		\begin{equation}\label{eq:T-operator-b}
			(\mathcal{T} f)(x;z)
			:=
			\int_{-\infty}^{x}(q^{PT})_{y}(y)
			\int_{-\infty}^{y}
			\mathrm{e}^{-2\mathrm{i}z(y-s)}
			q_s(s)f(s;z)\,ds\,dy .
		\end{equation}
		Since $q\in H^{3}(\mathbb{R})\cap H^{2,1}(\mathbb{R})$ , we have $(q^{PT})_{x},q_x\in L^1(\mathbb{R})\cap L^2(\mathbb{R})$.
		Hence
		\begin{equation}\label{eq:K-bound}
			\|\mathcal{T}\|_{L^\infty\to L^\infty}
			\le
			\|(q^{PT})_{x}\|_{L^1}\|q_x\|_{L^1} .
		\end{equation}
		Using \eqref{eq:b-integral-representation} and
		\eqref{eq:psi11-minus-K}, we get
		\begin{align}
			\lambda^{-1}b(\lambda)
			&=
			\mathrm{i}\int_{\mathbb{R}}
			q_y(y)\bigl(\psi_{11}^{-}(y;\lambda)-1\bigr)
			\mathrm{e}^{2\mathrm{i}zy}\,dy
			+
			\mathrm{i}\int_{\mathbb{R}}
			q_y(y)\mathrm{e}^{2\mathrm{i}zy}\,dy
			\nonumber\\
			&=
			-\mathrm{i}z\int_{\mathbb{R}}
			q_y(y)(\mathcal{T} \psi_{11}^{-})(y;z)
			\mathrm{e}^{2\mathrm{i}zy}\,dy
			+
			2z\int_{\mathbb{R}}q(y)
			\mathrm{e}^{2\mathrm{i}zy}\,dy .
			\label{eq:lambdaminus1-b-decomposition}
		\end{align}
		Therefore,
		\begin{equation}\label{eq:zminus1-lambdaminus1-b-formula}
			z^{-1}\lambda^{-1}b(\lambda)
			=
			-\mathrm{i}\int_{\mathbb{R}}
			q_y(y)(\mathcal{T} \psi_{11}^{-})(y;z)
			\mathrm{e}^{2\mathrm{i}zy}\,dy
			+
			2\int_{\mathbb{R}}q(y)
			\mathrm{e}^{2\mathrm{i}zy}\,dy .
		\end{equation}
		By the Plancherel's theorem and the weighted Fourier estimate, together
		with the uniform boundedness of $\psi_{11}^{-}$, we have
		\begin{align}
			\|z^{-1}\lambda^{-1}b(\lambda)\|_{H_z^1}
			&\le
			C\|q_x \mathcal{T} \psi_{11}^{-}\|_{L^{2,1}}
			+
			C\|q\|_{L^{2,1}}
			\nonumber\\
			&\le
			C\|\mathcal{T}\|_{L^\infty\to L^\infty}
			\sup_{x\in\mathbb{R}}\|\psi_{11}^{-}(\cdot;\lambda)\|_{L^\infty_z}
			\|q_x\|_{L^{2,1}}
			+
			C\|q\|_{L^{2,1}}
			\nonumber\\
			&\le
			C\bigl(\|q\|_{H^3\cap H^{2,1}}\bigr),
			\label{eq:zminus1-b-H1-bound}
		\end{align}
		which proves
		$z^{-1}\lambda^{-1}b(\lambda)\in H_z^1(\mathbb{R})$.
		The proof for $\breve b(\lambda)$ follows in the same way from the
		second column Volterra equation, or equivalently from the symmetry
		$\breve b(\lambda)=-b(\lambda)$ at $t=0$.  \hfill
	\end{proof}

	In conclusion, we can get the following corollary, which is the same as Corollary \ref{corollay:Jost L-C}.
	\begin{corollary}\label{corollary: Scattering-coefficients-L-C}
		If $q \in H^{3}(\mathbb{R}) \cap H^{2,1}(\mathbb{R})$, then the map
		\begin{equation}\label{eq: S-C Lipschitz continuous1}
			H^{2,1}(\mathbb{R}) \ni q \to [a(\lambda) - a_{\infty},\ \breve{a}(\lambda) - \breve{a}_{\infty},\ \lambda b(\lambda), \ \lambda \breve{b}(\lambda), \ \lambda^{-1} b(\lambda), \ \lambda^{-1} \breve{b}(\lambda)] \in H_{z}^{1}(\mathbb{R})
		\end{equation}
		is Lipschitz continuous. Moreover, the map
		\begin{equation}\label{eq: S-C Lipschitz continuous2}
			\begin{aligned}
				H^{3}(\mathbb{R}) \cap H^{2,1}(\mathbb{R}) \ni q &\to [\lambda b(\lambda), \ \lambda \breve{b}(\lambda), \ \lambda^{-1} b(\lambda), \ \lambda^{-1} \breve{b}(\lambda)] \in L_{z}^{2,1}(\mathbb{R}),
			\end{aligned}
		\end{equation}
		is Lipschitz continuous.
	\end{corollary}
	The proof of this corollary is directly from Proposition \ref{prop: Scattering-coefficients Smooth} and Corollary \ref{corollay:Jost L-C}.
	
	\section{Construction of the Riemann-Hilbert problems}\label{section-RHP}

	\subsection{A basic RH problem on the $\lambda$-plane}

		From \eqref{refl}, \eqref{eq:a-small-lambda}--\eqref{eq:b-small-lambda} and \eqref{eq:a-large-lambda}--\eqref{eq:b-large-lambda}, we can get that $r(0) = \breve{r}(0) = 0$. Now we define the matrix function
	\begin{equation}\label{eq:Phi-lambda-def }
		\Phi(x;\lambda) := \mathrm{e}^{ -\frac{\mathrm{i}}{2} m_+(x) \sigma_3}
		\begin{cases}
			\left( \dfrac{\psi_1^{-}(x;\lambda)}{a(\lambda)},\, \psi_2^{+}(x;\lambda)
			\right), \quad &\lambda \in D_{-}, \\[1.2em]
			\left( \psi_1^{+}(x;\lambda),\, \dfrac{\psi_2^{-}(x;\lambda)}{\breve{a}(\lambda)}
			\right), \quad &\lambda \in D_{+},
		\end{cases}
	\end{equation}
	which satisfy the following RH problem:
	\begin{RH}\label{RH:lambda-basic}
		Find a matrix function $\Phi(x;\lambda)$ with the following properties:
		\begin{enumerate}
			\item \textbf{Analyticity:}
			$\Phi(x;\lambda)$ is analytical in $\mathbb{C} \setminus \Sigma_{\lambda}$.
			
			\item \textbf{Jump condition:}
			Let $\Phi_+(x;\lambda)$ and $ \Phi_{-}(x;\lambda)$ denote the non-tangential boundary values of $\Phi(x;\lambda)$ from $D_+$ and $D_-$, respectively. Then \begin{equation}\label{eq:Phi-lambda-jump}
			\Phi_+(x;\lambda) = \Phi_-(x;\lambda) (I + V(x;\lambda)), \quad \lambda \in \Sigma_{\lambda}, \end{equation}
			where the jump matrix is \begin{equation}\label{eq:V-lambda}
			V(x;\lambda) =
			    \begin{pmatrix}
				 0 & \breve{r}(\lambda) \mathrm{e}^{2\mathrm{i} \lambda^2 x} \\[0.4em]
				 -r(\lambda) \mathrm{e}^{-2\mathrm{i} \lambda^2 x} & -r(\lambda) \breve{r}(\lambda)
				\end{pmatrix}.
			\end{equation}
			
			\item \textbf{Normalization:}
			$	\Phi(x;\lambda) = I + \mathcal{O}(\lambda^{-1})$ as $\lambda \to \infty.$
		\end{enumerate}
	\end{RH}
	In the nonlocal case $q^{PT}(x)=q(-x)$, the standard involution symmetry is
	broken, and hence the reflection coefficients $r(\lambda)$ and
	$\breve r(\lambda)$ are not necessarily complex conjugates of each other.
	Consequently, the jump matrix $V(x;\lambda)$ in \eqref{eq:V-lambda} is in
	general non-Hermitian. In order to apply the Zhou's vanishing lemma, we impose
	a small-norm condition on the effective Volterra potential.
	
	\begin{proposition}\label{prop:reflection-small-bound}
		Suppose that $q\in H^3(\mathbb{R})\cap H^{2,1}(\mathbb{R})$, $q^{PT}(x)=q(-x)$, and satisfies the small-norm restriction
		\begin{equation}\label{eq:small-norm}
			\rho := \| \widehat{Q} \|_{L^1(\mathbb{R})} < \ln \left(\frac{3}{2}\right).
		\end{equation}
		Then the reflection coefficients satisfy the uniform bounds
		\begin{equation}\label{eq:r-rbreve-strict-bound}
			|r(\lambda)|\le c_0<1, \qquad |\breve r(\lambda)|\le c_0<1, \qquad \lambda\in\Sigma_{\lambda},
		\end{equation}
		for some positive constant $c_0$.
	\end{proposition}
	
	\begin{proof}
		We first prove the estimate for $r(\lambda)$. Let us consider the Zakharov-Shabat type spectral problem \eqref{eq:ZSprobelm} and its associated Volterra integral equation \eqref{eq:Volterra equation}.
		
		For $z \in \mathbb{R}$, the Jost solution $\Psi^-(x; z)$ satisfies:
		\begin{equation}\label{eq:Volterra-Psi}
			\Psi^-(x; z) = I + \int_{-\infty}^{x} \mathrm{e}^{\mathrm{i}z(x-y)\hat{\sigma}_3} \widehat{Q}(y) \Psi^-(y; z) \, dy.
		\end{equation}
		According to the classical definition in the inverse scattering transform, the scattering matrix $S_{\Psi}(z)$ is associated with the asymptotics of the Jost solution as $x \to +\infty$:
		\begin{equation}
			S_{\Psi}(z) = \lim_{x \to +\infty} \mathrm{e}^{-\mathrm{i}zx\sigma_3} \Psi^-(x; z) \mathrm{e}^{\mathrm{i}zx\sigma_3}.
		\end{equation}
		Multiplying \eqref{eq:Volterra-Psi} by $\mathrm{e}^{-\mathrm{i}zx\sigma_3}$ from the left and $\mathrm{e}^{\mathrm{i}zx\sigma_3}$ from the right, and taking the limit $x \to +\infty$, we get the exact integral representation for $S_{\Psi}(z) - I$:
		\begin{equation}\label{eq:S-minus-I}
			S_{\Psi}(z) - I = \int_{-\infty}^{+\infty} \mathrm{e}^{-\mathrm{i}zy\hat{\sigma}_3} \left( \widehat{Q}(y) \Psi^-(y; z) \right) \, dy.
		\end{equation}
		
		To estimate $\Psi^-(x; z)$ in \eqref{eq:S-minus-I}, we employ the Volterra-Neumann expansion.
  Let $\Psi^-(x; z) = \sum_{n=0}^{\infty} \Psi_n(x; z)$, where $\Psi_0(x; z) = I$ and
		\begin{equation}
			\Psi_{n+1}(x; z) = \int_{-\infty}^{x} \mathrm{e}^{\mathrm{i}z(x-y)\hat{\sigma}_3} \widehat{Q}(y) \Psi_n(y; z) \, dy.
		\end{equation}
		For $z \in \mathbb{R}$, the matrix norm (defined as the sum of $L^1$ norms of entries) is non-increasing under conjugation by $\mathrm{e}^{\mathrm{i}z(x-y)\sigma_3}$. Thus, $\|\mathrm{e}^{\mathrm{i}z(x-y)\hat{\sigma}_3} \widehat{Q}(y) \Psi_n(y; z)\| \le \|\widehat{Q}(y)\| \|\Psi_n(y; z)\|$. Introducing $\rho(x) = \int_{-\infty}^{x} \|\widehat{Q}(y)\| \, dy$, we can prove by induction that $\|\Psi_n(x; z)\| \le \frac{1}{n!} \rho(x)^n$. Summing these norms yields the uniform bound:
		\begin{equation}\label{eq:Psi_bound}
			\|\Psi^-(x; z)\| \le \sum_{n=0}^{\infty} \frac{1}{n!} \rho(x)^n = \mathrm{e}^{\rho(x)}.
		\end{equation}
		
		Taking the matrix norm of \eqref{eq:S-minus-I} and using \eqref{eq:Psi_bound}, we obtain
		\begin{equation}
			\|S_{\Psi}(z) - I\| \le \int_{-\infty}^{+\infty} \|\widehat{Q}(y)\| \|\Psi^-(y; z)\| \, dy \le \int_{-\infty}^{+\infty} \|\widehat{Q}(y)\| \mathrm{e}^{\rho(y)} \, dy.
		\end{equation}
		Using the substitution $d\rho(y) = \|\widehat{Q}(y)\| \, dy$, we obtain the uniform global bound:
		\begin{equation}
			\|S_\Psi(z) - I\| \le \int_{0}^{\rho} \mathrm{e}^u \, du = \mathrm{e}^\rho - 1, \qquad z \in \mathbb{R}.
		\end{equation}
		
		This strictly bounds the components of $S_\Psi(z)$:
		\begin{equation}\label{eq:Psi-bounds}
			|a_\Psi(z) - 1| \le \mathrm{e}^\rho - 1, \quad
			|\breve{a}_{\Psi}(z) - 1| \le \mathrm{e}^{\rho} - 1, \quad
			|b_\Psi(z)| \le \mathrm{e}^\rho - 1, \quad |\breve{b}_\Psi(z)| \le \mathrm{e}^\rho - 1.
		\end{equation}
		
		From the gauge transformation \eqref{eq:lambda to z}, the original scattering coefficients are related to the $\Psi$-system components by $a(\lambda) = a_\Psi(z)$, $\breve a(\lambda) = \breve a_{\Psi}(z)$,  $b(\lambda) = -(2\lambda)^{-1} b_\Psi(z)$, and $\breve{b}(\lambda) = -2\lambda \breve{b}_\Psi(z)$.
		For $a(\lambda)$, we immediately obtain the uniform lower bound:
		\begin{equation}\label{eq:a-lower-bound-proof}
			|a(\lambda)| \ge 1 - |a_\Psi(z) - 1| \ge 2 - \mathrm{e}^\rho.
		\end{equation}
		
		To bound $b(\lambda)$ globally and circumvent the spectral singularities at $\lambda=0$ and $\lambda=\infty$, we critically exploit the involution symmetry of the nonlocal Fokas-Lenells equation.
		Under $p(x) = q(-x)$, the scattering data satisfy $\breve{b}(\lambda) = -b(\lambda)$ for $\lambda \in \Sigma_\lambda$.
		This symmetry provides two simultaneous constraints for $|b(\lambda)|$:
		\begin{align}
			|b(\lambda)| &= \left|-\frac{1}{2\lambda} b_\Psi(z)\right| \le \frac{1}{2|\lambda|}(\mathrm{e}^\rho - 1), \\
			|b(\lambda)| &= |\breve{b}(\lambda)| = \left|-2\lambda \breve{b}_\Psi(z)\right| \le 2|\lambda|(\mathrm{e}^\rho - 1).
		\end{align}
		Therefore, for all $\lambda \in \Sigma_\lambda$, $|b(\lambda)|$ must be bounded by their minimum:
		\begin{equation}\label{eq:b-global-bound}
			|b(\lambda)| \le (\mathrm{e}^\rho - 1) \min\left(2|\lambda|, \frac{1}{2|\lambda|}\right) \le \mathrm{e}^\rho - 1,
		\end{equation}
		where we used the elementary fact that $\min(x, x^{-1}) \le 1$ for all $x > 0$.
		To guarantee $|r(\lambda)| = |b(\lambda)/a(\lambda)| \le c_0 < 1$, it is sufficient to require $|b(\lambda)| < |a(\lambda)|$.
		Substituting \eqref{eq:a-lower-bound-proof} and \eqref{eq:b-global-bound}, this is satisfied if:
		\begin{equation}
			\mathrm{e}^\rho - 1 < 2 - \mathrm{e}^\rho \quad \iff \quad 2\mathrm{e}^\rho < 3 \quad \iff \quad \rho < \ln \left(\frac{3}{2}\right).
		\end{equation}
		Since our small norm restriction explicitly demands $\rho < \ln(3/2)$, there exists a constant $c_0 = \frac{\mathrm{e}^\rho - 1}{2 - \mathrm{e}^\rho} < 1$ such that $|r(\lambda)| \le c_0$.
		
		The proof for $|\breve{r}(\lambda)| \le c_0 < 1$ follows analogously.
		\hfill
	\end{proof}
	
	In fact, in the above estimation, although we set $z \in \mathbb{R}$, if we discuss $a(\lambda)$ and $\breve{a}(\lambda)$, they can be extended to the corresponding analytic region, which ensures the absence of solitons and spectral singularities. This also means that under \eqref{eq:small-norm}, the propositions discussed earlier are valid.
	
	\begin{remark}\label{re:small-norm-restriction}
		The small norm inequality \eqref{eq:small-norm} leads to the absence of both discrete spectrum and spectral singularities. Therefore, this situation cannot be avoided when characterizing nonlocal regions.
	\end{remark}
	
	It is worthy noting that under \eqref{eq:small-norm}, the Hermitian part of the jump matrix
	$V(x;\lambda)$ is strictly positive definite. More precisely, we have the
	following proposition.
	
	\begin{proposition}\label{prop:V-positive-definite}
		Suppose that \eqref{eq:small-norm} holds. Then there exist
		two positive constants $\eta_-$ and $\eta_+$ such that, for any
		$x\in\mathbb{R}$, $\lambda\in\Sigma_{\lambda}$, and any column vector
		$f\in\mathbb{C}^2$, one has
		\begin{equation}\label{eq:V-coercive}
			\mathrm{Re}\left[f^{\dagger} (I + V(x;\lambda)) f\right]
			\ge
			\eta_- f^{\dagger}f,
		\end{equation}
		and
		\begin{equation}\label{eq:V-bounded}
			\|(I + V(x;\lambda)) f\|
			\le
			\eta_+\|f\|.
		\end{equation}
		In particular, the Hermitian part
		\begin{equation}
			I+V_H(x;\lambda):= I +
			\frac{1}{2} \left(V(x;\lambda)+V^{\dagger}(x;\lambda)\right)
		\end{equation}
		is uniformly positive definite on $\mathbb{R}\times\Sigma_{\lambda}$.
	\end{proposition}
	
	\begin{proof}
		By Proposition \ref{prop:reflection-small-bound}, there exists
		$c_0<1$ such that
		\bee
		|r(\lambda)|\le c_0,
		\qquad
		|\breve r(\lambda)|\le c_0,
		\qquad \lambda\in\Sigma_{\lambda}.
		\ene
		From the explicit expression of $V(x;\lambda)$ in \eqref{eq:V-lambda},
		we compute
		\begin{equation}\label{eq:VH-explicit}
			I+V_H(x;\lambda)
			=
			\begin{pmatrix}
				1 &
				\dfrac12
				\left(\breve r(\lambda)-\overline{r(\lambda)}\right)
				\mathrm{e}^{2\mathrm{i}\lambda^2x}
				\\[0.8em]
				\dfrac12
				\left(\overline{\breve r(\lambda)}-r(\lambda)\right)
				\mathrm{e}^{-2\mathrm{i}\lambda^2x}
				&
				1-\mathrm{Re}\left[r(\lambda)\breve r(\lambda)\right]
			\end{pmatrix}.
		\end{equation}
		For any $f\in\mathbb{C}^2$,
		\begin{equation}\label{eq:real-part-V}
			\mathrm{Re}\left[f^{\dagger} (I+V(x;\lambda)) f\right]
			=
			f^{\dagger}(I+V_H(x;\lambda)) f.
		\end{equation}
		Thus the quadratic form in \eqref{eq:real-part-V} is real-valued.
		
		We next show that $V_H(x;\lambda)$ is positive definite. Since the
		upper-left entry of $V_H$ is $1$, it suffices to prove that
		$\det(I+ V_H(x;\lambda))>0$. A direct calculation gives
		\begin{align}
			\det(I + V_H(x;\lambda))
			&=
			1-\mathrm{Re}\left[r(\lambda)\breve r(\lambda)\right]
			-\frac14
			\left|
			\breve r(\lambda)-\overline{r(\lambda)}
			\right|^2
			\nonumber\\
			&=
			1-\frac14
			\left|
			r(\lambda)+\overline{\breve r(\lambda)}
			\right|^2.
			\label{eq:det-VH}
		\end{align}
		Using \eqref{eq:r-rbreve-strict-bound}, we obtain
		\begin{equation}\label{eq:det-VH-lower}
			\det (I+V_H(x;\lambda))
			\ge
			1-\frac14
			\left(|r(\lambda)|+|\breve r(\lambda)|\right)^2
			\ge
			1-c_0^2
			>0.
		\end{equation}
		Therefore $V_H(x;\lambda)$ is strictly positive definite for all
		$x\in\mathbb{R}$ and $\lambda\in\Sigma_{\lambda}$.
		
		Moreover, the entries of $V_H$ are uniformly bounded. Hence its smallest
		eigenvalue has a strictly positive uniform lower bound. Equivalently,
		there exists a constant $\eta_->0$, depending only on $c_0$, such that
		\begin{equation}
			f^{\dagger} (I+V_H(x;\lambda)) f
			\ge
			\eta_- f^{\dagger}f.
		\end{equation}
		Combining this estimate with \eqref{eq:real-part-V}, we obtain
		\eqref{eq:V-coercive}.
		
		It remains to prove \eqref{eq:V-bounded}. Let $f=(f_1,f_2)^T$. By
		\eqref{eq:V-lambda},
		\begin{equation}
			(I+V(x;\lambda))f
			=
			\begin{pmatrix}
				f_1+\breve r(\lambda)
				\mathrm{e}^{2\mathrm{i}\lambda^2x}f_2
				\\[0.4em]
				-r(\lambda)\mathrm{e}^{-2\mathrm{i}\lambda^2x}f_1
				+
				\left(1-r(\lambda)\breve r(\lambda)\right)f_2
			\end{pmatrix}.
		\end{equation}
		Using again $|r(\lambda)|\le c_0$ and
		$|\breve r(\lambda)|\le c_0$, we get
		\begin{align}
			\|(I+V(x;\lambda)) f\|^2
			&\le
			\left(|f_1|+c_0|f_2|\right)^2
			+
			\left(c_0|f_1|+(1+c_0^2)|f_2|\right)^2
			\nonumber\\
			&\le
			\eta_+^2
			\left(|f_1|^2+|f_2|^2\right)
			=
			\eta_+^2\|f\|^2,
		\end{align}
		for some positive constant $\eta_+$ depending only on $c_0$, which proves
		\eqref{eq:V-bounded}. \hfill
	\end{proof}
	
	\subsection{A new RH problem on the $z$-plane}

	In order to use the theorems on classical Cauchy integral and its projection on the real axis,
	we change the original RH Problem \ref{RH:lambda-basic} with jump contour on $\Sigma_{\lambda}$ in the $\lambda$-plane into a RH problem with jump contour on $\mathbb{R}$ in the $z$-plane.
	
	We define a matrix function as
	\begin{equation}\label{eq:N-z-def}
		N(x;z) := \mathrm{e}^{ -\frac{\mathrm{i}}{2} m_+(x) \sigma_3}
		\begin{cases}
			\left( \dfrac{\Psi_1^{-}(x;z)}{a(\lambda)},\, \Psi_2^{+}(x;z)
			\right), \quad &z \in \mathbb{C}^{-} \\[1.2em]
			\left( \Psi_1^{+}(x;z),\, \dfrac{\Psi_2^{-}(x;z)}{\breve{a}(\lambda)}
			\right), \quad &z \in \mathbb{C}^{+}
		\end{cases}
	\end{equation}
	and it satisfies the following RH problem on the $z$-plane.
	\begin{RH}\label{RH:z-plane}
		Find a matrix function $N(x;z)$ with the following properties:
		\begin{enumerate}
		    \item \textbf{Analyticity:}
		    $N(x;z)$ is analytic in $\mathbb{C} \setminus \mathbb{R}$.
		
		    \item \textbf{Jump condition:}
		    $N(x;z)$ satisfies the jump condition
		    \begin{equation}\label{eq:z-jump-condition}
		    	N_{+}(x;z) = N_{-}(x;z) (I+R(x;z)), \quad z \in \mathbb{R},
		    \end{equation}
		    where the jump matrix is defined by
		    \begin{equation}\label{eq:R-jump-def}
		    	R(x;z) :=
		    	\begin{pmatrix}
		    		0 & r_1(z) \mathrm{e}^{2 \mathrm{i} z x} \v\\
		    		r_2(z) \mathrm{e}^{-2 \mathrm{i} z x} &
		    		 r_{1}(z) r_{2}(z)
		    	\end{pmatrix},
		    \end{equation}
		    with
		    \begin{equation}\label{eq:r12-definition}
		    	r_1(z):=-\frac{\breve b(\lambda)}{2\lambda\breve a(\lambda)},
		    	\qquad
		    	r_2(z):=\frac{2\lambda b(\lambda)}{a(\lambda)},
		    	\qquad z=\lambda^2 .
		    \end{equation}
		
		    \item \textbf{Normalization:}
		    	$N(x;z) = I + \mathcal{O}(z^{-1})$, as $ z \to \infty.$
		\end{enumerate}
	\end{RH}
    Now we try to analyze some properties of the reflection coefficient.

    \begin{proposition}\label{prop:r1,2-strict-control}
    By the small norm assumption \eqref{eq:small-norm}, we can get that
    	\begin{equation}\label{eq:r1,2-strict-bound}
    		|r_j (z)| \le c_0 < 1, \quad z \in \mathbb{R},\ j=1,2.
    	\end{equation}
    \end{proposition}

    \begin{proof}
    	One can find that the definition of $r_{1,2}(z)$ in \eqref{eq:r12-definition} can be represented that
    	\begin{equation}
    		r_1(z) = \dfrac{\breve{b}_{\Psi}(z)}{\breve{a}_{\Psi}(z)}, \quad r_2(z) = \dfrac{-b_{\Psi}(z)}{a_{\Psi}(z)},
    	\end{equation}
    	which the notations are used in Proposition \ref{prop:reflection-small-bound}. With the analysis in  Proposition \ref{prop:reflection-small-bound}, we can find that
    	\begin{equation}
    		\left|\dfrac{\breve{b}_{\Psi}(z)}{\breve{a}_{\Psi}(z)}\right|,
    		\left|\dfrac{b_{\Psi}(z)}{a_{\Psi}(z)}\right| \le \dfrac{\mathrm{e}^\rho-1}{2-\mathrm{e}^\rho}= c_0 <1,
    	\end{equation}
    	which complete the proof.
    	\hfill
    \end{proof}

\begin{proposition}\label{prop:Lp-interpolation}
		Let $\Omega\subset\mathbb{R}$. If
		$f\in L^{p_1}(\Omega) \cap L^\infty(\Omega),\, 1\leq p_1<\infty$, then $f\in L^{p_2}(\Omega) \cap L^\infty(\Omega)$ for arbitrary $ 1\le p_1\le p_2<\infty$, and
\begin{equation}\label{eq:Lp-interpolation}
			\|f\|_{L^{p_2}(\Omega)}
			\le
			C\|f\|_{L^{p_1}(\Omega)}^{p_1/p_2},
		\end{equation}
   where the constant $C$ depends on $p_{1,2}$ and
		$\|f\|_{L^\infty(\Omega)}$. Conversely, if $f\in L^{p_2}(\Omega) \cap L^\infty(\Omega),\, 1<p_2<\infty$, then
 $\|f\|_{L^{p_1}(\Omega)}$ must not exist for arbitry $1\le p_1< p_2<\infty$.
			\end{proposition}
	\begin{proof}
	 Since $f\in L^\infty(\Omega)$, for $1\le p_1\le p_2<\infty$ we have
		$|f|^{p_2-p_1}\in L^\infty(\Omega)$. By the H\"older's inequality,
		\begin{equation} \label{fp}
			\|f\|_{L^{p_2}(\Omega)}^{p_2}
			=
			\int_\Omega |f|^{p_2}\,dx
			=
			\int_\Omega |f|^{p_1}|f|^{p_2-p_1}\,dx
			\le
			\|f\|_{L^\infty(\Omega)}^{p_2-p_1}
			\|f\|_{L^{p_1}(\Omega)}^{p_1} .
		\end{equation}
		Taking the $p_1$-th root on both sides of (\ref{fp}) gives rise to \eqref{eq:Lp-interpolation}.

      Conversely, for example, we take $f_{\alpha}(x)=x^{-\alpha},\, x\in [1, \infty]$ with $1/p_2<\alpha\le 1/p_1$ and $1\le p_1< p_2<\infty$ (i.e., $\alpha p_2>1,\, 0<\alpha p_1\leq 1)$. Thus we have
      \bee
       \int_1^{\infty}|f_{\alpha}(x)|^{p_2}dx=\int_1^{\infty}\frac{1}{x^{\alpha p_2}}dx=\frac{1}{\alpha p_2-1}<\infty.
      \ene
      that is, $f_{\alpha}(x)\in L^{p_2}(\Omega) \cap L^\infty(\Omega)$. However
      \bee
       \int_1^{\infty}|f_{\alpha}(x)|^{p_1}dx=\int_1^{\infty}\frac{1}{x^{\alpha p_1}}dx \geq \int_1^{\infty}\frac{1}{x}dx\, ({\rm non-integrable}),
      \ene
      which means that $f_{\alpha}(x)$ is not be $L^{p_1}$-integrable in $x\in [1, \infty)$.
	\end{proof}

	\begin{proposition}\label{prop: b-lambda-O5}
		Suppose that $q \in H^{3}(\mathbb{R}) \cap H^{2,1}(\mathbb{R})$ and $q^{PT}(x) = q(-x)$, the scattering coefficient $b(\lambda)$ admits the higher-order asymptotic behavior at the origin:
		\begin{equation}
			b(\lambda) = \mathcal{O}(\lambda^5),\quad \breve b(\lambda) = \mathcal{O}(\lambda^5), \quad \lambda \to 0.
		\end{equation}
	\end{proposition}
	\begin{proof}
		By the symmetry $b(-\lambda) = -b(\lambda)$, the Maclaurin expansion of $b(\lambda)$ as $\lambda \to 0$ contains only odd powers:
		\begin{equation}\label{eq:b-expansion-odd}
			b(\lambda) = b_3 \lambda^3 + b_5 \lambda^5 + \mathcal{O}(\lambda^7), \quad \lambda \to 0.
		\end{equation}
		To determine $b_3$, we extend the perturbation expansion of the Jost solutions \eqref{eq:psi-small-lambda-expansion} to $\mathcal{O}(\lambda^3)$:
		\begin{equation}
			\psi^{\pm}(x; \lambda) = I + \lambda U_1(x) + \lambda^2 U_2^{\pm}(x) + \lambda^3 U_3^{\pm}(x) + \mathcal{O}(\lambda^4).
		\end{equation}
		Substituting this into the $x$-part of \eqref{eq:New Laxpair} and matching the $\mathcal{O}(\lambda^3)$ terms, we have $U_{3,x}^{\pm} = \mathrm{i}[\sigma_3, U_1^{\pm}] + \mathrm{i}Q_x U_2^{\pm}$.
		Recall from \eqref{eq:U-small-lambda} that $U_1 = \mathrm{i}Q$ and $U_{2,11}^{\pm}(x) = -\int_{\pm\infty}^x (q^{PT})_{y}(y) q(y) \, dy$. For the $(2,1)$-entry of $U_3^{\pm}$, direct computation yields the differential equation:
		\begin{equation}
			U_{3, 21, x}^{\pm} = 2q(x) - \mathrm{i}q_x(x) \int_{\pm\infty}^x (q^{PT})_{y}(y) q(y) \, dy.
		\end{equation}
		Integrating this equation with the boundary conditions $\psi^{\pm}(x;\lambda) \to I$ as $x \to \pm\infty$ gives:
		\begin{equation}\label{eq:U3_21}
			U_{3, 21}^{\pm}(x) = \int_{\pm\infty}^x \left( 2q(y) - \mathrm{i}q_y(y) \int_{\pm\infty}^y q_s^{PT}(s) q(s) \, ds \right) dy.
		\end{equation}
		
		Using the Wronskian representation \eqref{eq:b-lambda-def} evaluated at $x=0$,
		$$b(\lambda) = \psi_{11}^+(0; \lambda) \psi_{21}^-(0; \lambda) - \psi_{21}^+(0; \lambda) \psi_{11}^-(0; \lambda).$$ Extracting the coefficient of $\lambda^3$, we find:
		\begin{equation}\label{eq:b3-initial}
			b_3 = U_{3, 21}^-(0) - U_{3, 21}^+(0) + \mathrm{i}q(0) \left( U_{2, 11}^+(0) - U_{2, 11}^-(0) \right).
		\end{equation}
		The diagonal contribution is simply evaluated as:
		\begin{equation}\label{eq:U2-diff}
			U_{2, 11}^+(0) - U_{2, 11}^-(0) = \int_{-\infty}^{+\infty} (q^{PT})_{x}(x) q(x) \, dx.
		\end{equation}
		For the off-diagonal terms, \eqref{eq:U3_21} leads to:
		\begin{equation}\label{eq:U3-diff}
		\begin{array}{rl}
	 U_{3, 21}^-(0) - U_{3, 21}^+(0) = \d \int_{-\infty}^{+\infty} 2q(x) \, dx - \mathrm{i} \left(\int_{-\infty}^0
+\int_0^{+\infty} \right) \left[q_x \left( \int_{-\infty}^x (q^{PT})_{y} q \, dy \right)\right]dx.
		\end{array}
\end{equation}
		Applying integration by parts to the double integrals in the bracket yields:
		\begin{align}
			\int_{-\infty}^0 q_x(x) \left( \int_{-\infty}^x (q^{PT})_{y} q \, dy \right) dx &= q(0) \int_{-\infty}^0 (q^{PT})_{x} q \, dx - \int_{-\infty}^0 (q^{PT})_{x}(x) q^2(x) \, dx, \\
			\int_0^{+\infty} q_x(x) \left( \int_{+\infty}^x (q^{PT})_{y} q \, dy \right) dx &= -q(0) \int_0^{+\infty} (q^{PT})_{x} q \, dx - \int_0^{+\infty} (q^{PT})_{x}(x) q^2(x) \, dx.
		\end{align}
		Substituting these back into \eqref{eq:U3-diff} and collecting terms, the boundary values $q(0)$ cancel out perfectly with the contribution from \eqref{eq:U2-diff}. This simplifies $b_3$ to a remarkable global integral:
		\begin{equation}\label{eq:b3-clean-format}
			b_3 = \int_{-\infty}^{+\infty} \left( 2q(x) + \mathrm{i} (q^{PT})_{x}(x) q^2(x) \right) dx.
		\end{equation}
		
		Finally, applying integration by parts to the second term in \eqref{eq:b3-clean-format} under the assumption of vanishing boundaries at infinity, we have $\mathrm{i} \int_{-\infty}^{+\infty} (q^{PT})_{x} q^2 \, dx = -2\mathrm{i} \int_{-\infty}^{+\infty} q^{PT} q q_x \, dx$. Consequently,
		\begin{equation}\label{eq:b3-to-nFL}
			b_3 = 2 \int_{-\infty}^{+\infty} \left( q(x) - \mathrm{i} q(x) q^{PT}(x) q_x(x) \right) dx.
		\end{equation}
		Crucially, recalling the nFL equation \eqref{eq:nFL}, we have the inherent dynamical identity $q - \mathrm{i} q q^{PT} q_x = q_{xt}$. Substituting this into \eqref{eq:b3-to-nFL} directly yields:
		\begin{equation}
			b_3 = 2 \int_{-\infty}^{+\infty} q_{xt} \, dx = 2 \frac{\partial}{\partial t} \int_{-\infty}^{+\infty} q_x \, dx = 0.
		\end{equation}
		This confirms that $b_3 \equiv 0$ identically as a consequence of the equation's intrinsic dynamics, establishing $b(\lambda) = \mathcal{O}(\lambda^5)$ as $\lambda \to 0$. For $\breve b(\lambda)$, we can use the symmetry $\breve b(\lambda) = -b(\lambda)$, which can prove the result.
		\hfill
	\end{proof}
	
	\begin{proposition}\label{prop:zminus2-reflection}
		Suppose that $q(x)\in H^{3}(\mathbb{R})\cap H^{2,1}(\mathbb{R})$.
		Let $\delta>0$ be fixed and sufficiently small. Then
		\begin{equation}\label{eq:zminus2-b-local}
			z^{-2}\lambda^{-1}b(\lambda)\in L_z^2(-\delta,\delta),
			\qquad
			z^{-2}\lambda^{-1}\breve b(\lambda)\in L_z^2(-\delta,\delta).
		\end{equation}
		Moreover, under \eqref{eq:small-norm} and we have
		\begin{equation}\label{eq:a-lower-bound}
			|a(\lambda)|\ge a_0>0,\qquad
			|\breve a(\lambda)|\ge \breve a_0>0,
			\qquad \lambda\in\Sigma_\lambda .
		\end{equation}
		Then
		\begin{equation}\label{eq:zminus2-r12}
			z^{-2}r_1(z)\in L^2(\mathbb{R}),
			\qquad
			z^{-2}r_2(z)\in L^2(\mathbb{R}).
		\end{equation}
	\end{proposition}
	
	\begin{proof}
		By Proposition \ref{prop:z-minus-one-b}, we have
		\begin{equation}\label{eq:zminus1-b-H1-Linf}
			z^{-1}\lambda^{-1}b(\lambda)\in H_z^1(\mathbb{R})
			\subset L_z^2(\mathbb{R})\cap L_z^\infty(\mathbb{R}),
		\end{equation}
		and similarly
		$$
		z^{-1}\lambda^{-1}\breve b(\lambda)
		\in
		L_z^2(\mathbb{R})\cap L_z^\infty(\mathbb{R}).
		$$
		Hence, by Proposition \ref{prop:Lp-interpolation},
		\begin{equation}\label{eq:zminus1-b-Lp}
			z^{-1}\lambda^{-1}b(\lambda)\in L_z^p(\mathbb{R}),
			\qquad
			z^{-1}\lambda^{-1}\breve b(\lambda)\in L_z^p(\mathbb{R}),
			\qquad 2\le p<\infty .
		\end{equation}
		From Proposition \ref{prop: b-lambda-O5}, we can find that
		\begin{equation}\label{eq:zminus2-b-local-proof}
			z^{-2} \lambda^{-1} b(\lambda) =\mathcal{O}(1),\quad \lambda \to 0.
		\end{equation}
		which means that for fixed sufficiently small $\delta>0$, $z^{-2} \lambda^{-1} b(\lambda)$ is bounded on $(-\delta, \delta)$.
		Thus
		$z^{-2}\lambda^{-1}b(\lambda)\in L_z^2(-\delta,\delta)$.
		The same argument gives
		$z^{-2} \lambda^{-1} \breve b(\lambda)\in L_z^2(-\delta,\delta)$.
		
		Now let $\Omega:=\mathbb{R}\setminus(-\delta,\delta),$
		and denote by $\chi_\Omega$ the characteristic function of $\Omega$.
		Since $r_1,r_2\in L^\infty(\mathbb{R})$ and
		$z^{-2}\in L^2(\Omega)$, we have
		\begin{align}
			\|z^{-2}r_1(z)\|_{L^2(\mathbb{R})}
			&\le
			\|\chi_\Omega z^{-2}r_1(z)\|_{L^2(\mathbb{R})}
			+
			\|(1-\chi_\Omega)z^{-2}r_1(z)\|_{L^2(\mathbb{R})}\\
			&\le
			2 \|r_1\|_{L^\infty}
			\|z^{-2}\|_{L^2(\delta, \infty)}
			+
			\frac{1}{\breve a_0}
			\|z^{-2}\lambda^{-1}\breve b(\lambda)\|_{L^2(0,\delta)}.
			\label{eq:zminus2-r1-proof}
		\end{align}
		Similarly, since
		$$
		\lambda b(\lambda)=z \lambda^{-1} b(\lambda),
		$$
		we obtain
		\begin{align}
			\|z^{-2}r_2(z)\|_{L^2(\mathbb{R})}
			\le
			2 \|r_2\|_{L^\infty}
			\|z^{-2}\|_{L^2(\delta, \infty)}
			+
			\frac{4}{a_0}
			\|z^{-1}\lambda^{-1}b(\lambda)\|_{L^2(0,\delta)}.
			\label{eq:zminus2-r2-proof}
		\end{align}
		Therefore,
		$z^{-2}r_1,z^{-2}r_2\in L^2(\mathbb{R})$.
		 \hfill
	\end{proof}
	
	From the above analysis, we can draw the following conclusion
	\begin{proposition}\label{prop:r12-Lipschitz}
		Suppose that $q(x)$ satisfies the small norm restriction \eqref{eq:small-norm}. The relations between the potential $q(x)$ and the reflection coefficients $r_1(z), r_2(z)$ are given as follows:
		\begin{enumerate}[label=(\roman*)]
			\item If $q \in H^{2,1}(\mathbb{R})$, then $r_1, r_2 \in H_z^1(\mathbb{R})$.
			\item If $q \in H^3(\mathbb{R}) \cap H^{2,1}(\mathbb{R})$, then $r_1, r_2 \in L_z^{2,1}(\mathbb{R})$ and $z^{-2}r_1, z^{-2}r_2 \in L_z^2(\mathbb{R})$.
		\end{enumerate}
		Furthermore, the mapping $H^3(\mathbb{R}) \cap H^{2,1}(\mathbb{R}) \ni q \mapsto (r_1, r_2) \in \mathcal{W}(\mathbb{R})$
				is Lipschitz continuous.
	\end{proposition}
	
	\begin{proof}
		According to the definition \eqref{eq:r12-definition}, the reflection coefficients can be written as:
		\begin{equation}
			r_1(z) = -\frac{\lambda^{-1}\breve{b}(\lambda)}{2\breve{a}(\lambda)}, \quad r_2(z) = \frac{2\lambda b(\lambda)}{a(\lambda)}.
		\end{equation}
		By Proposition \ref{prop:zminus2-reflection}, under the small norm condition \eqref{eq:small-norm}, the denominators are uniformly bounded away from zero on the continuous spectrum:
		\begin{equation}
			|a(\lambda)| \ge a_0 > 0, \quad |\breve{a}(\lambda)| \ge \breve{a}_0 > 0, \quad \forall \lambda \in \Sigma_\lambda.
		\end{equation}
		
		From Corollary \ref{corollary: Scattering-coefficients-L-C}, we know that if $q \in H^{2,1}(\mathbb{R})$, then $\lambda^{-1}\breve{b}(\lambda), \lambda b(\lambda) \in H_z^1(\mathbb{R})$, and $a(\lambda) - a_\infty, \breve{a}(\lambda) - \breve{a}_\infty \in H_z^1(\mathbb{R})$. Since $H_z^1(\mathbb{R})$ is a Banach algebra, and the reciprocal of a non-vanishing function of the form constant + $H_z^1(\mathbb{R})$ remains in the same class, the quotients $r_1(z)$ and $r_2(z)$ also belong to $H_z^1(\mathbb{R})$.
		
		Similarly, if $q \in H^3(\mathbb{R}) \cap H^{2,1}(\mathbb{R})$, Corollary \ref{corollary: Scattering-coefficients-L-C} ensures that $\lambda^{-1}\breve{b}(\lambda), \lambda b(\lambda) \in L_z^{2,1}(\mathbb{R})$. Since multiplication by an $L_z^\infty$ bounded function preserves the $L_z^{2,1}(\mathbb{R})$ property, we have $r_1, r_2 \in L_z^{2,1}(\mathbb{R})$. The additional property $z^{-2}r_1, z^{-2}r_2 \in L_z^2(\mathbb{R})$ has already been strictly established in Proposition \ref{prop:zminus2-reflection}.
		
		Finally, the Lipschitz continuity of the mapping from the potential $q$ to $(r_1, r_2)$ follows directly from the Lipschitz continuity of the mappings from $q$ to the basic scattering data $\lambda^{-1}\breve{b}(\lambda), \lambda b(\lambda), a(\lambda)$, and $\breve{a}(\lambda)$, as stated in Corollary \ref{corollary: Scattering-coefficients-L-C}. Because the denominator terms $a(\lambda)$ and $\breve{a}(\lambda)$ are globally strictly separated from zero, the quotient map inherits and preserves the Lipschitz continuity.
		\hfill
	\end{proof}
	
	Before we formally estimate the solution, we also need to analyze some analytical properties of the original reflection coefficient.
	\begin{proposition}\label{prop:r-lambda-property}
		If $r_1(z), r_2(z) \in H_z^1(\mathbb{R}) \cap L_z^{2,1}(\mathbb{R})$ and $q$ satisfies \eqref{eq:small-norm}, then $r(\lambda), \breve{r}(\lambda) \in L_z^{2,1}(\mathbb{R}) \cap L_z^\infty(\mathbb{R})$.
	\end{proposition}
	
	\begin{proof}
		First, we prove the boundedness in $L_{z}^{\infty}(\mathbb{R})$. From \ref{prop:reflection-small-bound}, we have
		$$
			|r(\lambda)| \le c_0, \quad |\breve{r}(\lambda)| \le c_0, \quad \forall \lambda \in \Sigma_{\lambda}.
		$$
		This strictly bounds them by a constant, thus yielding $r(\lambda), \breve{r}(\lambda) \in L_{z}^{\infty}(\mathbb{R})$.
		
		Next, we prove the property in $L_{z}^{2,1}(\mathbb{R})$. According to \eqref{eq:r12-definition}, we have
		\begin{equation}
			r_1(z) = -\frac{\breve{b}(\lambda)}{2\lambda\breve{a}(\lambda)} = -\frac{\breve{r}(\lambda)}{2\lambda}, \quad r_2(z) = \frac{2\lambda b(\lambda)}{a(\lambda)} = 2\lambda r(\lambda).
		\end{equation}
		To avoid singularities at $\lambda=0$ and the scaling issue as $|\lambda| \to \infty$, we write $r(\lambda)$ in a piecewise form:
		\begin{equation}
			r(\lambda) = \begin{cases} \dfrac{b(\lambda)}{a(\lambda)}, & |\lambda| \le 1, \\[0.8em] \dfrac{r_2(z)}{2\lambda}, & |\lambda| > 1. \end{cases}
		\end{equation}
		For $|\lambda| \le 1$, since $b(\lambda) = \mathcal{O}(\lambda^5)$ as $\lambda \to 0$ and $a(\lambda)$ has no zeros, $r(\lambda)$ is continuous and bounded on this compact set. For $|\lambda| > 1$, since $(2\lambda)^{-1}$ is bounded and $r_2(z) \in L_z^{2,1}(\mathbb{R})$, this part belongs to $L_z^{2,1}(\mathbb{R})$.
		
		Similarly, for $\breve{r}(\lambda)$, we write:
		\begin{equation}
			\breve{r}(\lambda) = \begin{cases} -2\lambda r_1(z), & |\lambda| \le 1, \\[0.8em] \dfrac{\breve{b}(\lambda)}{\breve{a}(\lambda)}, & |\lambda| > 1. \end{cases}
		\end{equation}
		For $|\lambda| \le 1$, $2\lambda$ is bounded and $r_1(z) \in L_z^{2,1}(\mathbb{R})$, so the product is in $L_z^{2,1}(\mathbb{R})$.
		For $|\lambda| > 1$, we use the property $\lambda\breve{b}(\lambda) \in L_z^{2,1}(\mathbb{R})$ established in Corollary \ref{corollary: Scattering-coefficients-L-C}. We can write the numerator as $\breve{b}(\lambda) = \lambda^{-1} (\lambda\breve{b}(\lambda))$. Since $|\lambda^{-1}| < 1$ in this region, the multiplication by a bounded function ensures $\breve{b}(\lambda) \in L_z^{2,1}(\mathbb{R})$. Because $\breve{a}(\lambda)$ is uniformly bounded away from zero, we have $\breve{r}(\lambda) \in L_z^{2,1}(\mathbb{R})$.
		Combining these results, we conclude that $r(\lambda), \breve{r}(\lambda) \in L_{z}^{2,1}(\mathbb{R})$.
		
		Combining these results, we conclude that $r(\lambda), \breve{r}(\lambda) \in L_{z}^{2,1}(\mathbb{R})$.
		\hfill
	\end{proof}
	
	\begin{proposition}\label{prop:lambda-r-infinity-bound}
		If $r_1(z), r_2(z) \in H_z^1(\mathbb{R}) \cap L_z^{2,1}(\mathbb{R})$, then
		\begin{equation}\label{eq:r1,2-H1-L21-bounds}
			\|\lambda r_1(z)\|_{L_z^\infty} \le \|r_1\|_{H^1 \cap L^{2,1}}, \quad \|\lambda r_2(z)\|_{L_z^\infty} \le \|r_2\|_{H^1 \cap L^{2,1}}.
		\end{equation}
	\end{proposition}
	
	\begin{proof}
		For $r_j(z) \in H_z^1(\mathbb{R}) \cap L_z^{2,1}(\mathbb{R})$ ($j=1,2$), using $z = \lambda^2$ and the fundamental theorem of calculus, we have:
		\begin{equation}
			|\lambda r_j(z)|^2 = |z r_j^2(z)| = \left| \int_0^z \frac{d}{ds} \left( s r_j^2(s) \right) ds \right|.
		\end{equation}
		Applying the product rule to the integrand gives:
		\begin{equation}
			|\lambda r_j(z)|^2 = \left| \int_0^z \left( r_j(s)^2 + 2s r_j(s) r_j'(s) \right) ds \right|.
		\end{equation}
		By extending the integration interval to the whole real line $\mathbb{R}$ and using the triangle inequality, we obtain:
		\begin{equation}
			|\lambda r_j(z)|^2 \le \int_{\mathbb{R}} |r_j(s)|^2 \, ds + 2 \int_{\mathbb{R}} |s r_j(s)| |r_j'(s)| \, ds.
		\end{equation}
		Applying the Cauchy-Schwarz inequality to the second term yields:
		\begin{equation}
			|\lambda r_j(z)|^2 \le \|r_j\|_{L^2}^2 + 2 \|r_j'\|_{L^2} \|z r_j\|_{L^2}.
		\end{equation}
		Since $r_j \in H_z^1(\mathbb{R}) \cap L_z^{2,1}(\mathbb{R})$, the norms $\|r_j\|_{L^2}$, $\|r_j'\|_{L^2}$, and $\|z r_j\|_{L^2}$ are all bounded. We can bound the entire right-hand side by the mixed norm:
		\begin{equation}
			|\lambda r_j(z)|^2 \le \|r_j\|_{H^1 \cap L^{2,1}}^2.
		\end{equation}
		Taking the square root on both sides, we arrive at the desired estimate:
		\begin{equation}
			\|\lambda r_j(z)\|_{L_z^\infty} \le \|r_j\|_{H^1 \cap L^{2,1}}, \quad j=1, 2.
		\end{equation}
		\hfill
	\end{proof}

\section{Inverse Scattering Transform}\label{section-IST}
	
\subsection{Transformations of RH problem}
	
For a given function $h(z)\in L^{p}(\mathbb{R})$ with $1\le p<\infty$, the Cauchy operator is defined by
	\begin{equation}\label{eq:Cauchy-operator}
		\mathcal{C}(h)(z):=\frac{1}{2\pi \mathrm{i}} \int_{\mathbb{R}} \frac{h(s)}{s-z}\, ds, \quad z \in \mathbb{C} \setminus \mathbb{R}.
	\end{equation}
	When $z\pm \mathrm{i}\epsilon$ approaches to a point at the real axis $z\in\mathbb{R}$ transversely from the upper and the lower half planes, the Cauchy operator $\mathcal{C}$ becomes the Plemelj projection operators defined respectively by
	\begin{equation}\label{eq:Cauchy-projection}
		\mathcal{P}^{\pm}(h)(z):=\lim_{\epsilon \downarrow 0} \frac{1}{2 \pi \mathrm{i}} \int_{\mathbb{R}} \frac{h(s)}{s-(z\pm \mathrm{i} \epsilon)}\, ds, \quad z \in \mathbb{R}.
	\end{equation}
	
We list the basic properties of the Cauchy and the Plemelj projection operators in the following proposition\textcolor{red}{\cite{Duren1970,Pelinovsky2018,Zhou1998}}.
	\begin{proposition}\label{prop:C-P-properties}
		For every $h\in L^{p}(\mathbb{R})$, $1\le p<\infty$ the Cauchy operator $\mathcal{C}(h)$ and the projection operator $\mathcal{P}^{\pm}(h)$ has the following properties:
		
		\begin{enumerate}
			\item  $\mathcal{C}(h)$ is analytic in $\mathbb{C}^{\pm}$ and goes to zero as $|z| \to \infty$.
			
			\item  If $h \in L^{1}(\mathbb{R})$, then in $\mathbb{C}^{+}$ or $\mathbb{C}^{-}$, as $\lambda$ alongs the contour which is not tangential to $\mathbb{R}$ going to $\infty$, the Cauchy operator admits the following asymptotic
			\begin{equation}\label{eq:Cauchy-operator-asymptotic}
				\lim_{|z| \to \infty} z \mathcal{C}(h) (z) = - \dfrac{1}{2 \pi \mathrm{i}} \int_{\mathbb{R}} h(s) \, ds.
			\end{equation}
			
			And the projection operator $\mathcal{P}^{\pm}(h)$ has the following properties:
			
			\item $\mathcal{C}(h)$ approaches to $\mathcal{P}^{\pm} (h)$ almost everywhere, when a point $z \in \mathbb{C}^{\pm}$ approaches to a point $z_0 \in \mathbb{R}$ by any non-tangential contour from $\mathbb{C}^{\pm}$.
			
			\item For $1 < p < \infty$, there exists a positive constant $c$ such that
			\begin{equation}\label{eq:C-P-bounds}
				\| \mathcal { P } ^{\pm} ( h ) \| _ {L^{p} ( \mathbb{R}) } \leq c\| h \| _ {L^{p}(\mathbb{R} )} .
			\end{equation}
		\end{enumerate}
	\end{proposition}
	
	The RH Problem \ref{RH:z-plane} admits the Beals-Coifman solution
	\begin{equation}\label{eq:BC-solution}
		N_{\pm}(x;z) = I + \mathcal{P}^{\pm}(N_{-}(x;\cdot) R(x;\cdot))(z), \quad z \in \mathbb{R},
	\end{equation}
	which can be used to estimate the columns of $N_{\pm}(x;z) - I$.
	Next, we introduce the $2 \times 2$ matrix
	\begin{equation}\label{eq:T-def}
		T(x;z) = (N_{+,1}(x;z) - e_1, N_{-,2}(x;z) - e_2).
	\end{equation}
	From the jump condition \eqref{eq:z-jump-condition}, we can deduce that $T(x;z)$ satisfies the following integral equation
	\begin{equation}
		T - \mathcal{P}^{+}(T R_{-}) - \mathcal{P}^{-}(T R_{+}) = F
	\end{equation}
	where
	\begin{equation}\label{eq:R_pm-def}
		R_{+}(x;z) = \begin{pmatrix} 0 & r_1(z) \mathrm{e}^{2\mathrm{i}zx} \\ 0 & 0 \end{pmatrix}, \quad
		R_{-}(x;z) = \begin{pmatrix} 0 & 0 \\ r_2(z) \mathrm{e}^{-2\mathrm{i}zx} & 0 \end{pmatrix},
	\end{equation}
	and
	\begin{equation}\label{eq:F-matrix}
		F(x;z) = \left( \mathcal{P}^{+}(r_2(z) \mathrm{e}^{-2\mathrm{i}zx})e_2, \, \mathcal{P}^{-}(r_1(z) \mathrm{e}^{2\mathrm{i}zx})e_1 \right)
	\end{equation}
	
	An estimate on $F$ relies on the mapping properties of the projection operators. Our main task is to make a further estimate on the Cauchy integral projection in the equation above. To analyze the derivatives of $N_{+,1}(x;z)$ and $N_{-,2}(x;z)$, we take the derivative of the inhomogeneous equation in $x$, which gives
	\begin{equation}\label{eq:T-derivative}
		\partial_x T - \mathcal{P}^{+}(\partial_x T R_{-}) - \mathcal{P}^{-}(\partial_x T R_{+}) = \tilde{F},
	\end{equation}
	where
	\begin{equation}\label{tilde-F-def}
		\begin{split}
			\tilde{F}(x;z) =& \, 2\mathrm{i} \left( e_2 \mathcal{P}^{+}(-z r_2(z) \mathrm{e}^{-2\mathrm{i}zx}), \, e_1 \mathcal{P}^{-}(z r_1(z) \mathrm{e}^{2\mathrm{i}zx}) \right) \\
			&+ 2\mathrm{i} \begin{pmatrix}
				\mathcal{P}^{+}(-z r_2(z) N_{-,12}(x;z) \mathrm{e}^{-2\mathrm{i}zx}) & \mathcal{P}^{-}(z r_1(z) (N_{+,11}(x;z) - 1) \mathrm{e}^{2\mathrm{i}zx}) \v\\
				\mathcal{P}^{+}(-z r_2(z) (N_{-,22}(x;z) - 1) \mathrm{e}^{-2\mathrm{i}zx}) & \mathcal{P}^{-}(z r_1(z) N_{+,21}(x;z) \mathrm{e}^{2\mathrm{i}zx})
			\end{pmatrix}.
		\end{split}
	\end{equation}
	
	Proposition \ref{prop:lambda-r-infinity-bound} inspires us to start with $\lambda r_{1,2}$ to solve the problem. Therefore, for every $\lambda \in \mathbb{C} \setminus \{0\}$, we introduce two new matrices
	\begin{equation}\label{eq:tau1,2-def}
		\tau_1(\lambda) := \begin{pmatrix}
			1 & 0 \\
			0 & -2 \lambda
		\end{pmatrix}, \quad
		\tau_2(\lambda) := \begin{pmatrix}
			(-2 \lambda)^{-1} & 0 \\
			0 & 1
		\end{pmatrix},
	\end{equation}
	to give the original problem the form we want. With the definition of $V(x;\lambda)$ in \eqref{eq:V-lambda}, $\tau_{1,2}$ admits
	\begin{equation}
		\tau_1^{-1}(\lambda) R(x;z) \tau_1(\lambda) = \tau_2^{-1}(\lambda) R(x;z) \tau_2(\lambda) = V(x;\lambda), \quad z\in \mathbb{R},\ \lambda \in \Sigma_{\lambda}.
	\end{equation}
	It can be seen from the following analysis that the transformation of $N(x;z)$ using these two matrices could not only complete the estimation of $N(x;z)$ itself and its derivative, but also make full use of the good properties of the jump matrix $V(x;\lambda)$.
	
	In the following RH problem, the properties of the matrix elements are characterized in the $z$-plane: Under \eqref{eq:small-norm}, if $r_{1,2} \in H_z^1(\mathbb{R}) \cap L_z^{2,1}(\mathbb{R})$, then Proposition \ref{prop:r-lambda-property} implies that $V \in L_z^{1}(\mathbb{R}) \cap L_z^\infty(\mathbb{R})$ and $F(x;z) \in L_z^2(\mathbb{R})$ for every $x \in \mathbb{R}$. We consider the class of solutions to the RH Problem \ref{RH:z-plane} for every $x \in \mathbb{R}$. Therefore, we equivalently reduce the RH Problem \ref{RH:z-plane} in the $z$-plane to a new RH problem related to the matrix $V(x;\lambda)$ instead of the matrix $R(x;z)$.
	
	\begin{RH}\label{RH:G-problem}
		Find a matrix function $G_j(x;\lambda)$ ($j=1,2$) with the following properties:
		\begin{enumerate}
			\item \textbf{Analyticity:} $G_j(x;\lambda)$ are analytic functions of $z$ in $\mathbb{C}^{\pm}$.
			
			\item \textbf{Jump condition:} $G_j(x;\lambda)$ satisfies the jump condition
			\begin{equation}\label{eq:Jump-condition-G}
				G_{j,+}(x;\lambda) = G_{j,-}(x;\lambda)(I + V(x;\lambda)) + D_j(x;\lambda), \quad \lambda \in \Sigma_{\lambda},
			\end{equation}
			where
			\begin{equation}\label{eq:Gj-Dj-def}
				G_{j,\pm}(x;\lambda) := N_{\pm}(x;z)\tau_j(\lambda) - \tau_j(\lambda), \quad D_j(x;\lambda) := \tau_j(\lambda)V(x;\lambda).
			\end{equation}
			
			\item \textbf{Parity:} the columns of $G_{j,\pm}(x;\lambda)$, $G_{j,-}(x;\lambda)V$, and $D_j$ have the same parity in $\lambda$.
			
			\item \textbf{Normalization:}
			$	G_{j,\pm}(x;\lambda) \to 0 \quad \text{as} \quad |\lambda| \to \infty.$
		\end{enumerate}
	\end{RH}
	
	Based on the positive definiteness of $I + V_H(x;\lambda)$ established in Proposition \ref{prop:V-positive-definite}, we can express the solution of RH Problem \ref{RH:G-problem} with the Cauchy integrals
	\begin{equation}\label{eq:Gj-solutions}
		G_{j,\pm}(x;\lambda) = \mathcal{C}(G_{j,-}(x;\cdot)V + D_j)(z), \quad z \in \mathbb{C}^{\pm}.
	\end{equation}
	There is a solution $G_{j,-}(x;\lambda) \in L_z^2(\mathbb{R})$ to the Fredholm integral equation:
	\begin{equation}\label{eq:Gj_-solution}
		G_{j,-}(x;\lambda) = \mathcal{P}^{-}(G_{j,-}(x;\cdot)V(x;\cdot) + D_j(x;\cdot))(z), \quad z \in \mathbb{R}.
	\end{equation}
	Once $G_{j,-}(x;\lambda) \in L_z^2(\mathbb{R})$ is found, $G_{j,+}(x;\lambda) \in L_z^2(\mathbb{R})$ is obtained from the projection formula
	\begin{equation}\label{eq:Gj+-solution}
		G_{j,+}(x;\lambda) = \mathcal{P}^{+}(G_{j,-}(x;\cdot)V(x;\cdot) + D_j(x;\cdot))(z), \quad z \in \mathbb{R}.
	\end{equation}
	
	In order to show the RH problem \ref{RH:G-problem} exists a unique solution, we first need to explain the following proposition:
	\begin{proposition}\label{prop:G-unique-solution}
		Under \eqref{eq:small-norm}, for every $r(\lambda), \breve{r}(\lambda) \in L_z^{2} \cap L_z^{\infty}$ and $D(\lambda) \in L_z^2(\mathbb{R})$, there exists a unique solution $G(\lambda) \in L_z^2(\mathbb{R})$ of the linear inhomogeneous equation
		\begin{equation}\label{eq:inhomogeneous}
			(I - \mathcal{P}_{V}^{-}) G(\lambda) = D(\lambda), \quad \lambda \in \Sigma_{\lambda},
		\end{equation}
		with $ \mathcal{P}_{V}^{-} G := \mathcal{P}^{-} (GV)$.
	\end{proposition}
	\begin{proof}
		By the classical theory of singular integral operators \cite{Beals1984,Beals1985,Zhou1989}, and the Proposition \ref{prop:V-positive-definite}, which establishes that the Hermitian part $I + V_H(x;\lambda)$ is strictly positive definite for all $\lambda \in \Sigma_{\lambda}$, it is known that $I - \mathcal{P}_{V}^{-}$ is a Fredholm operator with index zero.
		
		By Fredholm's alternative, the unique solvability of the inhomogeneous equation \eqref{eq:inhomogeneous} in $L_z^2(\mathbb{R})$ is equivalent to showing that the corresponding homogeneous equation has only the trivial solution. Suppose there exists a nonzero row vector $g \in L_z^2(\mathbb{R})$ such that $(I - \mathcal{P}_{V}^{-}) g = 0 $. We can define two analytic functions in $\mathbb{C} \setminus \mathbb{R}$ with $V(\lambda) \in L_z^{2}(\mathbb{R}) \cap L_z^{\infty} (\mathbb{R})$ by
		\bee
		g_1(z) := \mathcal{C}(gV)(z)\quad \text{and} \quad g_2(z) := (\mathcal{C}(gV)(\bar{z}))^\dagger. \ene
		We integrate the product of $g_1$ and $g_2$ in $\mathbb{C}^{\pm}$ along a semicircle with radius R and center 0.
		\bee
		0 = \oint g_1(z)\, g_2(z)\, dz
		\ene
		holds because $g_1$ and $g_2$ are analytic functions in $\mathbb{C}^{+}$.
		Under the condition $g \in L_z^2(\mathbb{R})$ and $V(x;\lambda) \in L_z^2(\mathbb{R}) \cap L_z^\infty(\mathbb{R})$, we derive $gV \in L_z^1(\mathbb{R})$.
		Finally, Proposition \ref{prop:C-P-properties} implies $g_{1,2}(z) = \mathcal{O}(z^{-1})$ as $|z| \to \infty$.
		Therefore, the integral on the arc goes to zero as $R \to \infty$, and we obtain
		\begin{equation}
			\begin{split}
				0 &= \int_{\mathbb{R}} g_1(z) g_2(z)\, dz \\
				&= \int_{\mathbb{R}} \mathcal{P}^{+}(gV)(z) [\mathcal{P}^{-}(gV)(\bar{z})]^{\dagger}\, dz \\
				&= \int_{\mathbb{R}} [\mathcal{P}^{-}(gV) + gV][\mathcal{P}^{-}(gV)]^{\dagger}\, dz,
			\end{split}
		\end{equation}
		where we used the Plemelj formula $\mathcal{P}^{+} - \mathcal{P}^{-} = I$.
		Also, utilizing the homogeneous assumption $\mathcal{P}^{-}(gV) = g$, we deduce
		\begin{equation}
			0 = \int_{\mathbb{R}} g(I + V(x;\lambda))g^{\dagger}\, dz.
		\end{equation}
		For every $z \in \mathbb{R}$, the real part of the quadratic form associated with the matrix $I + V(x;\lambda)$ is strictly positive definite by Proposition \ref{prop:V-positive-definite}.
		Finally, we conclude that $g = 0$ is the only solution to the homogeneous equation $(I - \mathcal{P}_{V}^{-})g = 0$ in $L_z^2(\mathbb{R})$.
		\hfill
	\end{proof}
	
	For further estimates, we modify the method of Proposition \ref{prop:G-unique-solution} and prove that the operator $(I - \mathcal{P}_{V}^{-})^{-1}$ in the integral Fredholm equation \eqref{eq:inhomogeneous} is invertible with a bounded inverse in space $L_z^{2}(\mathbb{R})$.
	\begin{proposition}\label{prop:operator-inverse-bound}
		Suppose that the small norm restriction \eqref{eq:small-norm} holds. Then the inverse operator $(I-\mathcal{P}_{V}^{-})^{-1}$ is a bounded operator from $L_z^2(\mathbb{R})$ to $L_z^2(\mathbb{R})$.
		In particular, there is a positive constant $C$ such that for every row vector $d \in L_z^2(\mathbb{R})$, we have
		\begin{equation}\label{eq:operator-inverse-bound}
			\|(I-\mathcal{P}_{V}^{-})^{-1}d\|_{L_z^2} \le C\|d\|_{L_z^2}.
		\end{equation}
	\end{proposition}
	
	\begin{proof}
		We consider the linear inhomogeneous equation \eqref{eq:inhomogeneous} with $D \in L_z^2(\mathbb{R})$. Recalling the Plemelj formula $\mathcal{P}^{+} - \mathcal{P}^{-} = I$, we can decompose the solution as $G = G_{+} - G_{-}$, where $G_{+}$ and $G_{-}$ satisfy the inhomogeneous equations
		\begin{equation}\label{eq:G_pm-inhomogeneous}
			G_{-} - \mathcal{P}^{-}(G_{-}V) = \mathcal{P}^{-}(D), \quad G_{+} - \mathcal{P}^{-}(G_{+}V) = \mathcal{P}^{+}(D).
		\end{equation}
		By Proposition \ref{prop:G-unique-solution}, since $\mathcal{P}^{\pm}(D) \in L_z^2(\mathbb{R})$, there are unique solutions to the inhomogeneous equations \eqref{eq:G_pm-inhomogeneous}, meaning the decomposition $G = G_{+} - G_{-}$ is unique. We now estimate $G_{+}$ and $G_{-}$ in $L_z^2(\mathbb{R})$.
		
		To deal with $G_{-}$, we define two analytic functions in $\mathbb{C} \setminus \mathbb{R}$ by
		\begin{equation}
			g_1(z) := \mathcal{C}(G_{-}V)(z) \quad \text{and} \quad g_2(z) := (\mathcal{C}(G_{-}V + D)(\bar{z}))^{\dagger},
		\end{equation}
		similarly to the proof of Proposition \ref{prop:G-unique-solution}. By Proposition \ref{prop:C-P-properties}, since $D, G_{-} \in L_z^2(\mathbb{R})$ and $V(x;\lambda) \in L_z^2(\mathbb{R}) \cap L_z^\infty(\mathbb{R})$, we have $g_1(z) = \mathcal{O}(z^{-1})$ and $g_2(z) = \mathcal{O}(z^{-1})$ as $|z| \to \infty$. Therefore, the integral on the semi-circle of radius $R>0$ in the upper half-plane goes to zero as $R \to \infty$. Performing the same contour integration, we obtain
		\begin{equation}
			\begin{split}
				0 &= \oint g_1(z) g_2(z)\, dz \\
				&= \int_{\mathbb{R}} \mathcal{P}^{+}(G_{-}V) [\mathcal{P}^{-}(G_{-}V + D)]^{\dagger}\, dz \\
				&= \int_{\mathbb{R}} [\mathcal{P}^{-}(G_{-}V) + G_{-}V][\mathcal{P}^{-}(G_{-}V + D)]^{\dagger}\, dz \\
				&= \int_{\mathbb{R}} [G_{-} - \mathcal{P}^{-}(D) + G_{-}V] G_{-}^{\dagger}\, dz,
			\end{split}
		\end{equation}
		where we have used the first inhomogeneous equation in \eqref{eq:G_pm-inhomogeneous}. By the coercivity bound \eqref{eq:V-coercive} established in Proposition \ref{prop:V-positive-definite}, there is a positive constant $\eta_{-}$ such that
		\begin{equation}
			\eta_{-} \|G_{-}\|_{L_z^2}^2 \le \mathrm{Re} \int_{\mathbb{R}} G_{-}(I + V)G_{-}^{\dagger}\, dz = \mathrm{Re} \int_{\mathbb{R}} \mathcal{P}^{-}(D)G_{-}^{\dagger}\, dz \le \|\mathcal{P}^{-}(D)\|_{L_z^2} \|G_{-}\|_{L_z^2},
		\end{equation}
		where we applied the Cauchy-Schwarz inequality. Since $G_{-} = (I - \mathcal{P}_V^{-})^{-1} \mathcal{P}^{-}(D)$, for every row-vector $d \in L_z^2(\mathbb{R})$, the above inequality yields
		\begin{equation}\label{eq:G_minus_bound}
			\|(I-\mathcal{P}_V^{-})^{-1}\mathcal{P}^{-}d\|_{L_z^2} \le \eta_{-}^{-1} \|\mathcal{P}^{-}(d)\|_{L_z^2} = \eta_{-}^{-1} \|d\|_{L_z^2},
		\end{equation}
		using the $L^2$-boundedness of the Plemelj projection $\mathcal{P}^{-}$.
		
		To deal with $G_{+}$, we use $\mathcal{P}^{+} - \mathcal{P}^{-} = I$ and rewrite the second inhomogeneous equation in \eqref{eq:G_pm-inhomogeneous} as:
		\begin{equation}\label{eq:G_plus_rewrite}
			G_{+}(I + V) - \mathcal{P}^{+}(G_{+}V) = \mathcal{P}^{+}(D).
		\end{equation}
		We now define two analytic functions in $\mathbb{C} \setminus \mathbb{R}$ by
		\begin{equation}
			h_1(z) := \mathcal{C}(G_{+}V)(z) \quad \text{and} \quad h_2(z) := \mathcal{C}(G_{+}V + D)^{\dagger}.
		\end{equation}
		Integrating the product of $h_1$ and $h_2$ on the semi-circle of radius $R>0$ in the lower half-plane, and letting $R \to \infty$, we obtain
		\begin{equation}
			\begin{split}
				0 &= \oint h_1(z) h_2(z)\, dz \\
				&= \int_{\mathbb{R}} \mathcal{P}^{-}(G_{+}V) [\mathcal{P}^{+}(G_{+}V + D)]^{\dagger}\, dz \\
				&= \int_{\mathbb{R}} [G_{+} - \mathcal{P}^{+}(D)][G_{+}(I + V)]^{\dagger}\, dz,
			\end{split}
		\end{equation}
		where we have used equation \eqref{eq:G_plus_rewrite}. Taking the real part and using bounds \eqref{eq:V-coercive} and \eqref{eq:V-bounded} in Proposition \ref{prop:V-positive-definite}, there are positive constants $\eta_{-}$ and $\eta_{+}$ such that
		\begin{equation}
			\begin{split}
				\eta_{-} \|G_{+}\|_{L_z^2}^2 &\le \mathrm{Re} \int_{\mathbb{R}} G_{+}(I + V)^{\dagger}G_{+}^{\dagger}\, dz \\
				&= \mathrm{Re} \int_{\mathbb{R}} \mathcal{P}^{+}(D)(I + V)^{\dagger}G_{+}^{\dagger}\, dz \le \eta_{+} \|\mathcal{P}^{+}(D)\|_{L_z^2} \|G_{+}\|_{L_z^2}.
			\end{split}
		\end{equation}
		Since $G_{+} = (I - \mathcal{P}_V^{-})^{-1} \mathcal{P}^{+}(D)$, for every row-vector $d \in L_z^2(\mathbb{R})$, we get
		\begin{equation}\label{eq:G_plus_bound}
			\|(I-\mathcal{P}_V^{-})^{-1}\mathcal{P}^{+}d\|_{L_z^2} \le \eta_{-}^{-1}\eta_{+} \| d \|_{L_z^2}.
		\end{equation}
		The assertion of the proposition is finally proved by combining bounds \eqref{eq:G_minus_bound}, \eqref{eq:G_plus_bound}, and the triangle inequality.
		\hfill
	\end{proof}
	
	\subsection{Estimates on the Beals-Coifman solutions}

	In the RH problem \ref{RH:G-problem}, we can find that the corresponding $G_{j,\pm}(x;\lambda)$ can be represented by the analytic matrix function $N_{\pm}(x;z)$ in $\mathbb{C}^{\pm}$ for every $x \in \mathbb{R}$. Therefore, we can characterize their estimates by estimating $N_{\pm}(x;z)$. First, we introduce the following notations for the column vectors of the matrices $N_{\pm}(x;z)$
	\begin{equation}\label{eq:N-colunm-def}
		N_{\pm}(x;z) := \left( N_{\pm,1}(x;z), \ N_{\pm,2}(x;z)
		\right).
	\end{equation}
	Then the expression of $G_{j,\pm}(x;\lambda)$ is given by
	\begin{align*}
		\begin{split}
			G_{1,\pm}(x;\lambda) &= N_{\pm}(x;z) \tau_1(\lambda) - \tau_1(\lambda) \\
			&= \left(\left( N_{\pm,1}(x;z) - e_1
			\right) , \ -2 \lambda \left( N_{\pm,2}(x;z) - e_2
			\right)
			\right),
		\end{split} \\
		\begin{split}
			G_{2,\pm}(x;\lambda) &= N_{\pm}(x;z) \tau_2(\lambda) - \tau_2(\lambda) \\
			&= \left((-2 \lambda)^{-1} \left( N_{\pm,1}(x;z) - e_1
			\right) , \ \left( N_{\pm,2}(x;z) - e_2
			\right)
			\right),
		\end{split}
	\end{align*}
	and
	\begin{equation}\label{eq:Dj-R-relation}
		D_{j}(x;\lambda) := \tau_j (\lambda) V = R(x;z) \tau_j, \ j=1,2.
	\end{equation}
	We use the relation
	\begin{equation}
		(G_{j,-}V + D_j)_{i1} = (N_{-} \tau_j V)_{i1} = (N_{-} R \tau_j)_{i1}, \quad i=1,2,
	\end{equation}
	to obtain the explicit expressions for the first column of $G_{j,\pm}$:
	\begin{align}
		N_{\pm,1}(x;z) - e_1 &= \mathcal{P}^{\pm}(N_{-}(x;\cdot)R(x;\cdot))_{11}, \quad z \in \mathbb{R}, \label{eq:N-col1-a} \\
		(-2\lambda)^{-1}(N_{\pm,1}(x;z) - e_1) &= \mathcal{P}^{\pm}((-2\lambda)^{-1}N_{-}(x;\cdot)R(x;\cdot))_{11}, \quad z \in \mathbb{R}, \label{eq:N-col1-b}
	\end{align}
	and for the second column of $G_{j,\pm}$:
	\begin{align}
		N_{\pm,2}(x;z) - e_2 &= \mathcal{P}^{\pm}(N_{-}(x;\cdot)R(x;\cdot))_{21}, \quad z \in \mathbb{R}, \label{eq:N-col2-a} \\
		-2\lambda(N_{\pm,2}(x;z) - e_2) &= \mathcal{P}^{\pm}(-2\lambda N_{-}(x;\cdot)R(x;\cdot))_{21}, \quad z \in \mathbb{R}. \label{eq:N-col2-b}
	\end{align}
	
	\begin{remark}\label{re:redundant}
		From the results above, \eqref{eq:N-col1-a} and \eqref{eq:N-col1-b} may seem redundant, as well as \eqref{eq:N-col2-a} and \eqref{eq:N-col2-b}.  However, by using operations similar to \cite{Pelinovsky2018, Cheng2025}, we can see that both are redundant, which we omit the proof here.
	\end{remark}

	The solution to the RH Problem \ref{RH:z-plane} on the real line is given by the integral equation
	\begin{equation}\label{eq:N-integral-eq}
		N_{\pm}(x;z) = I + \mathcal{P}^{\pm}(N_{-}(x;\cdot)R(x;\cdot))(z), \quad z \in \mathbb{R}.
	\end{equation}
	The analytical continuation of functions $N_{\pm}(x;\cdot)$ in $\mathbb{C}^{\pm}$ is given by the Cauchy operators
	\begin{equation}
		N(x;z) = I + \mathcal{C}(N_{-}(x;\cdot)R(x;\cdot))(z), \quad z \in \mathbb{C}^{\pm}.
	\end{equation}
	
	The corresponding result on solvability of the integral equations \eqref{eq:N-integral-eq} is given by the following proposition.
	
	\begin{proposition}\label{prop:N-L2-bound}
		Suppose $r_{1,2}(z) \in H^1(\mathbb{R}) \cap L^{2,1}(\mathbb{R})$ such that the small norm inequality \eqref{eq:small-norm} is satisfied. Then for every $x \in \mathbb{R}$, $N_{\pm}(x;z)$ admits the estimate:
		\begin{equation}\label{eq:N-L2-estimate}
			\|N_{\pm}(x;\cdot) - I\|_{L^2(\mathbb{R})} \le C (\|r_1\|_{L^2} + \|r_2\|_{L^2} + \|r\|_{L^2} + \|\breve{r}\|_{L^2}),
		\end{equation}
		where $C$ is a positive constant depending only on $\|r_{1,2}\|_{L^\infty(\mathbb{R})}$.
	\end{proposition}
	\begin{proof}
		Under the condition $r_{1,2}(z) \in H^1(\mathbb{R}) \cap L^{2,1}(\mathbb{R})$, $R(x;z)\tau_j(\lambda)$ belongs to $L_z^2(\mathbb{R})$ for every $x \in \mathbb{R}$ according to the explicit expressions of $D_j$. There exists a positive constant $C$ dependent only on $\|r_{1,2}\|_{L^\infty}$ such that
		\begin{equation}\label{eq:Rtau-bound}
			\|R(x;\cdot) \tau_j(\cdot) \|_{L_z^2} \le C (\|r_1\|_{L^2} + \|r_2\|_{L^2} + \|r\|_{L^2} + \|\breve{r}\|_{L^2}), \quad j=1,2
		\end{equation}
		According to $\mathcal{P}^{-}(D_j) = \mathcal{P}^{-}(R(x;\cdot)\tau_j(\cdot))$, the equation for the projection operator $\mathcal{P}^{-}$ is derived from Proposition \ref{prop:operator-inverse-bound}. Each element of $N_{-}(x;z)$ satisfies the bound for the corresponding row vectors of $\mathcal{P}^{-}(D_j)$. Combining the operator bound \eqref{eq:operator-inverse-bound} and \eqref{eq:Rtau-bound}, we obtain \eqref{eq:N-L2-estimate}.
		\hfill
	\end{proof}
	
	Next, we derive accurate estimates on the solution to the integral equations by employing Fourier theory to evaluate the projection operators on oscillatory terms.
	For a given function $f(z) \in L^2(\mathbb{R})$, we define its Fourier transform and inverse transform by
	\begin{equation}
		\hat{f}(\xi) = \frac{1}{2\pi} \int_{\mathbb{R}} f(z) \mathrm{e}^{-\mathrm{i}z\xi} \, dz, \quad f(z) = \int_{\mathbb{R}} \hat{f}(\xi) \mathrm{e}^{\mathrm{i}z\xi} \, d\xi.
	\end{equation}
	\begin{proposition}\label{prop:Fourier-Projections}
		If $f(z)\in L^{2}(\mathbb{R})$, we have the following Fourier representations for the Cauchy projections:
		\begin{align}
			\mathcal{P}^{-}(f(z)\mathrm{e}^{2\mathrm{i}zx}) &= -\int_{-\infty}^{-2x}\hat{f}(\xi)\mathrm{e}^{\mathrm{i} z (\xi+2x)}\, d\xi, \label{eq:P_minus_Fourier_new} \\
			\mathcal{P}^{+}(f(z)\mathrm{e}^{-2\mathrm{i}zx}) &= \int_{2x}^{+\infty}\hat{f}(\xi)\mathrm{e}^{\mathrm{i}z(\xi-2x)}\, d\xi. \label{eq:P_plus_Fourier_new}
		\end{align}
	\end{proposition}
	
	\begin{proof}
		By using the definition of the projection operator and the Fourier inverse transform, we can expand $\mathcal{P}^{+}$ as follows:
		\begin{equation}\label{eq:P_plus_def_expand_new}
			\begin{split}
				\mathcal{P}^{+}(f(z)\mathrm{e}^{-2\mathrm{i}zx}) &= \frac{1}{2\pi \mathrm{i}} \lim_{\epsilon \downarrow 0} \int_{\mathbb{R}} \frac{f(s)\mathrm{e}^{-2\mathrm{i}sx}}{s-(z+\mathrm{i}\epsilon)}\, ds \\
				&= \frac{1}{2\pi \mathrm{i}} \int_{\mathbb{R}} \hat{f}(\xi) \left( \lim_{\epsilon \downarrow 0} \int_{\mathbb{R}} \frac{\mathrm{e}^{\mathrm{i}s(\xi-2x)}}{s-(z+\mathrm{i}\epsilon)}\, ds \right) d\xi.
			\end{split}
		\end{equation}
		Applying the residue theorem to the inner complex integral, the integrand has a simple pole at $s = z+\mathrm{i}\epsilon$ in the upper half-plane. To ensure the integral converges at infinity, the choice of the closed contour depends on the sign of $\xi-2x$:
		\begin{equation}\label{eq:Residue_limit_new}
		\begin{array}{rl}
	\lim_{\epsilon \downarrow 0} \dfrac{1}{2\pi \mathrm{i}} \d \int_{\mathbb{R}} \frac{\mathrm{e}^{\mathrm{i}s(\xi-2x)}}{s-(z+\mathrm{i}\epsilon)}\, ds
&\!\!\! =\d
			\lim_{\epsilon \downarrow 0}
			\begin{cases}
				\mathrm{e}^{\mathrm{i}(z+\mathrm{i}\epsilon)(\xi-2x)}, & \text{if } \xi-2x > 0 \\
				0, & \text{if } \xi-2x < 0
			\end{cases} \v\\
  &\!\!\!	= \d \chi(\xi-2x)\mathrm{e}^{\mathrm{i}z(\xi-2x)},
		\end{array}\end{equation}
		where $\chi(\cdot)$ is the characteristic function. Substituting \eqref{eq:Residue_limit_new} into \eqref{eq:P_plus_def_expand_new}, the characteristic function truncates the lower limit of the integration over $\xi$:
		\begin{equation}\label{eq:P_plus_final_new}
			\mathcal{P}^{+}(f(z)\mathrm{e}^{-2\mathrm{i}zx}) = \int_{2x}^{+\infty} \hat{f}(\xi)\mathrm{e}^{\mathrm{i}z(\xi-2x)}\, d\xi.
		\end{equation}
		
		To evaluate $\mathcal{P}^{-}$, we start with its definition:
		\begin{equation}
			\mathcal{P}^{-}(f(z)\mathrm{e}^{2\mathrm{i}zx}) = \frac{1}{2\pi \mathrm{i}} \lim_{\epsilon \downarrow 0} \int_{\mathbb{R}} \frac{f(s)\mathrm{e}^{2\mathrm{i}sx}}{s-(z-\mathrm{i}\epsilon)}\, ds.
		\end{equation}
		Taking the complex conjugation on both sides (noting that $x, z, s \in \mathbb{R}$), we find its relation to the $\mathcal{P}^{+}$ operator:
		\begin{equation}\label{eq:P_minus_conjugate_new}
			\overline{\mathcal{P}^{-}(f(z)\mathrm{e}^{2\mathrm{i}zx})} = -\frac{1}{2\pi \mathrm{i}} \lim_{\epsilon \downarrow 0} \int_{\mathbb{R}} \frac{\bar{f}(s)\mathrm{e}^{-2\mathrm{i}sx}}{s-(z+\mathrm{i}\epsilon)}\, ds = -\mathcal{P}^{+}(\bar{f}(z)\mathrm{e}^{-2\mathrm{i}zx}).
		\end{equation}
		Applying the previously derived formula \eqref{eq:P_plus_final_new} to $-\mathcal{P}^{+}(\bar{f}(z)\mathrm{e}^{-2\mathrm{i}zx})$, and utilizing the Fourier conjugate property $\hat{\bar{f}}(\xi) = \overline{\hat{f}(-\xi)}$, we obtain:
		\begin{equation}
			-\mathcal{P}^{+}(\bar{f}(z)\mathrm{e}^{-2\mathrm{i}zx}) = -\int_{2x}^{+\infty} \hat{\bar{f}}(\xi)\mathrm{e}^{\mathrm{i}z(\xi-2x)}\, d\xi .
		\end{equation}
		which means
		$$
		\mathcal{P}^{-}(f(z)\mathrm{e}^{2\mathrm{i}zx}) = -\int_{2x}^{+\infty} \hat{f}(-\xi)\mathrm{e}^{-\mathrm{i}z(\xi-2x)}\, d\xi
		=-\int_{-\infty}^{-2x}\hat{f}(\xi)\mathrm{e}^{\mathrm{i} z (\xi+2x)}\, d\xi.
		$$
		\hfill
	\end{proof}
	
	\begin{proposition}\label{prop:Ref-Coef-Estimates}
		For every $x_0 \in \mathbb{R}^{+}$ and every $r_{1,2}(z) \in \mathcal{W}(\mathbb{R})$, the reflection coefficients satisfy the following spatial weighted estimates for $\ell=0,1$:
		\begin{align}
			\sup_{x \in (x_0, \infty)} \left\| \langle x \rangle \mathcal{P}^{-}(z^{-\ell} r_1(z) \mathrm{e}^{2\mathrm{i}zx}) \right\|_{L^2_z(\mathbb{R})} &\le \| z^{-\ell} r_1(z) \|_{H_z^1(\mathbb{R})}, \label{eq:P_minus_weighted_new} \\
			\sup_{x \in (x_0, \infty)} \left\| \langle x \rangle \mathcal{P}^{+}(r_2(z) \mathrm{e}^{-2\mathrm{i}zx}) \right\|_{L^2_z(\mathbb{R})} &\le \| r_2(z) \|_{H_z^1(\mathbb{R})}. \label{eq:P_plus_weighted_new}
		\end{align}
		In addition, we have the following $L^\infty$ bounds on the entire real line:
		\begin{align}
			\sup_{x \in \mathbb{R}} \left\| \mathcal{P}^{-}(z^{-\ell} r_1(z) \mathrm{e}^{2\mathrm{i}zx}) \right\|_{L^\infty_z} &\le \frac{1}{\sqrt{2}} \| z^{-\ell} r_1(z) \|_{H_z^1(\mathbb{R})}, \label{eq:P_minus_Linf_new} \\
			\sup_{x \in \mathbb{R}} \left\| \mathcal{P}^{+}(r_2(z) \mathrm{e}^{-2\mathrm{i}zx}) \right\|_{L^\infty_z} &\le \frac{1}{\sqrt{2}} \| r_2(z) \|_{H_z^1(\mathbb{R})}. \label{eq:P_plus_Linf_new}
		\end{align}
		Furthermore, if $r_{1,2}(z) \in L_z^{2,1}(\mathbb{R})$, then the following estimates hold:
		\begin{align}
			\sup_{x \in \mathbb{R}} \left\| \mathcal{P}^{-}(z r_1(z) \mathrm{e}^{2\mathrm{i}zx}) \right\|_{L^2_z(\mathbb{R})} &\le \| r_1(z) \|_{L_z^{2,1}(\mathbb{R})}, \label{eq:P_minus_L21_new} \\
			\sup_{x \in \mathbb{R}} \left\| \mathcal{P}^{+}(z r_2(z) \mathrm{e}^{-2\mathrm{i}zx}) \right\|_{L^2_z(\mathbb{R})} &\le \| r_2(z) \|_{L_z^{2,1}(\mathbb{R})}. \label{eq:P_plus_L21_new}
		\end{align}
	\end{proposition}
	
	\begin{proof}
		The proof follows from the argument of \cite{Cheng2025}. We employ the Fourier representations of the Plemelj projections derived in the previous proposition. For a given reflection coefficient $r(z) \in L^2(\mathbb{R})$, its Fourier transform satisfies $\hat{r}(\xi) \in L^2(\mathbb{R})$. According to the classical Fourier theory, $r(z) \in H^1(\mathbb{R})$ if and only if $\hat{r}(\xi) \in L^{2,1}(\mathbb{R})$, and similarly, $r(z) \in L^{2,1}(\mathbb{R})$ if and only if $\hat{r}(\xi) \in H^1(\mathbb{R})$.
		
		To establish the weighted estimate \eqref{eq:P_minus_weighted_new}, we use the explicit Fourier representation for $\mathcal{P}^{-}$:
		\begin{equation}\label{eq:P_minus_Fourier_rep_new}
			\mathcal{P}^{-}(z^{-\ell} r_1(z) \mathrm{e}^{2\mathrm{i}zx}) = -\int_{-\infty}^{-2x} \widehat{(z^{-\ell} r_1)}(\xi) \mathrm{e}^{\mathrm{i}z(\xi+2x)} \, d\xi.
		\end{equation}
		By applying the integral bound \eqref{eq:L2,1control} from Proposition~\ref{prop: Estimate} to the right-hand side of \eqref{eq:P_minus_Fourier_rep_new}, the spatial weight $\langle x \rangle$ on the positive half-line $x \in (x_0, \infty)$ is directly bounded by the $L^{2,1}$-norm of the Fourier transform:
		\begin{equation}
			\sup_{x \in (x_0, \infty)} \left\| \langle x \rangle \int_{-\infty}^{-2x} \widehat{(z^{-\ell} r_1)}(\xi) \mathrm{e}^{\mathrm{i}z(\xi+2x)} \, d\xi \right\|_{L^2_z(\mathbb{R})} \le \sqrt{2\pi} \left\| \widehat{(z^{-\ell} r_1)}(\xi) \right\|_{L^{2,1}(\mathbb{R})}.
		\end{equation}
		Using the Plancherel identity, the $L^{2,1}$-norm of $\widehat{(z^{-\ell} r_1)}(\xi)$ is equivalent to the $H^1$-norm of $z^{-\ell} r_1(z)$ in the spectral domain, which yields the desired bound \eqref{eq:P_minus_weighted_new}.
		
		The estimate \eqref{eq:P_plus_weighted_new} for $\mathcal{P}^{+}$ follows via an identical routine by deploying the Fourier integration up to the lower limit $2x$ on the positive half-line $x \in (x_0, \infty)$.
		
		To prove the $L^\infty$ bound \eqref{eq:P_minus_Linf_new}, we take the absolute value of the Fourier representation \eqref{eq:P_minus_Fourier_rep_new} and apply the H\"{o}lder's inequality:
		\begin{equation}
			\left| \mathcal{P}^{-}(z^{-\ell} r_1(z) \mathrm{e}^{2\mathrm{i}zx}) \right| \le \int_{-\infty}^{-2x} \left| \widehat{(z^{-\ell} r_1)}(\xi) \right| \, d\xi \le \left\| \widehat{(z^{-\ell} r_1)}(\xi) \right\|_{L^1(\mathbb{R})}.
		\end{equation}
		By using the standard Sobolev embedding inequality $\| \hat{f} \|_{L^1} \le \frac{1}{\sqrt{2}} \| f \|_{H^1}$, we immediately obtain:
		\begin{equation}
			\sup_{x \in \mathbb{R}} \left\| \mathcal{P}^{-}(z^{-\ell} r_1(z) \mathrm{e}^{2\mathrm{i}zx}) \right\|_{L^\infty_z} \le \frac{1}{\sqrt{2}} \| z^{-\ell} r_1(z) \|_{H_z^1(\mathbb{R})},
		\end{equation}
		which proves \eqref{eq:P_minus_Linf_new}. The companion bound \eqref{eq:P_plus_Linf_new} is proved in the same manner.
		
		Finally, the assertions \eqref{eq:P_minus_L21_new} and \eqref{eq:P_plus_L21_new} characterize the dual mapping properties when the reflection coefficients belong to the weighted space $L^{2,1}(\mathbb{R})$. By swapping the roles of the spaces under the Plancherel transform, the $L^2_z$-norm of the projected derivative terms $\mathcal{P}^{-}(z r_1 \mathrm{e}^{2\mathrm{i}zx})$ and $\mathcal{P}^{+}(z r_2 \mathrm{e}^{-2\mathrm{i}zx})$ is uniformly bounded by the $H^1$-norm of their Fourier counterparts, which corresponds exactly to the initial $L^{2,1}$-norm of $r_{1,2}(z)$.
		\hfill
	\end{proof}
	
	Our ultimate goal is to estimate the solutions $N(x;z)$ to the RH Problem \ref{RH:z-plane}. By Proposition \ref{prop:N-L2-bound}, these solutions on the real line can be written in the integral Fredholm form \eqref{eq:N-integral-eq}. One remaining task is to estimate the column vectors $N_{+,1}-e_1$ and $N_{-,2}-e_2$. From the equation \eqref{eq:N-col1-a}, we obtain
	\begin{equation}\label{eq:N_minus_1_integral}
		N_{+,1}(x;z) - e_1 = \mathcal{P}^{+} \left(r_2(z)\mathrm{e}^{-2\mathrm{i}zx}N_{-,2}(x;z)
		\right), \quad z \in \mathbb{R}.
	\end{equation}
	Also, starting from \eqref{eq:N-col2-a}, we know
	\begin{equation}\label{eq:N_plus_2_integral}
	    N_{-,2}(x;z) - e_2 = \mathcal{P}^{-}
	    \left(r_1(z)\mathrm{e}^{2\mathrm{i}zx}N_{+,1}(x;z)
	    \right), \quad z \in \mathbb{R}.
	\end{equation}
	Both calculations take advantage of the explicit expressions for the columns of $N_- R$:
	\begin{equation}\label{eq:NR_facts}
		(N_- R)_1 = r_2(z)\mathrm{e}^{-2\mathrm{i}zx} N_{-,2}, \quad (N_- R)_2 = r_1(z)\mathrm{e}^{2\mathrm{i}zx} N_{+,1}.
	\end{equation}
	
	From the explicit expression \eqref{eq:F-matrix} for $F(x;z)$, we note that the first row vector of $F(x;z)$ is equal to the first row of $F(x;z)\tau_2(\lambda)$ and the second row vector of $F(x;z)$ is equal to the second row of $F(x;z)\tau_1(\lambda)$, which have the following representations:
	\begin{align}
		F_1(x;z) &= (F(x;z)\tau_2)_1 = \left(0, \mathcal{P}^{-}(r_1(z)\mathrm{e}^{2\mathrm{i}zx})\right), \label{eq:F1_row}\\
		F_2(x;z) &= (F(x;z)\tau_1)_2 = \left(\mathcal{P}^{+}(r_2(z)\mathrm{e}^{-2\mathrm{i}zx}), 0\right). \label{eq:F2_row}
	\end{align}
	
	\begin{proposition}\label{prop:BC-estimates}
		For every $x_0 \in \mathbb{R}^{+}$ and every $r_{1,2} \in H^1(\mathbb{R})$, the unique solution satisfies the estimates
		\begin{align}
			\sup_{x \in (x_0, \infty)} \|\langle x \rangle N_{+,21}(x;\cdot)\|_{L^2(\mathbb{R})} &\le C \|r_2\|_{H^1(\mathbb{R})}, \label{eq:N21_weighted} \\
			\sup_{x \in (x_0, \infty)} \|\langle x \rangle N_{-,12}(x;\cdot)\|_{L^2(\mathbb{R})} &\le C \|r_1\|_{H^1(\mathbb{R})}, \label{eq:N12_weighted}
		\end{align}
		where $C$ is a positive constant depending on $\|r_{1,2}\|_{L^\infty}$.
		Moreover, if $r_{1,2} \in H^1(\mathbb{R}) \cap L^{2,1}(\mathbb{R})$, we have the derivative bounds:
		\begin{align}
			\sup_{x \in \mathbb{R}} \|\partial_x N_{+,21}(x;\cdot)\|_{L^2(\mathbb{R})} &\le C (\|r_1\|_{H^1 \cap L^{2,1}} + \|r_2\|_{H^1 \cap L^{2,1}}), \label{eq:dx_N21} \\
			\sup_{x \in \mathbb{R}} \|\partial_x N_{-,12}(x;\cdot)\|_{L^2(\mathbb{R})} &\le C (\|r_1\|_{H^1 \cap L^{2,1}} + \|r_2\|_{H^1 \cap L^{2,1}}). \label{eq:dx_N12}
		\end{align}
	\end{proposition}
	\begin{proof}
		From the explicit expression \eqref{eq:T-def}, the first row vector of $T(x;z)\tau_2(\lambda)$ is given by
		\begin{equation}
			\left( (-2\lambda)^{-1}(N_{+,11}(x;z)-1), \ N_{-,12}(x;z) \right),
		\end{equation}
		and the second row vector of $T(x;z)\tau_1(\lambda)$ is given by
		\begin{equation}
			\left( N_{+,21}(x;z), \ -2\lambda(N_{-,22}(x;z)-1) \right).
		\end{equation}
		Using the operator bound \eqref{eq:operator-inverse-bound} along with the expressions of $F_1(x;z)$ and $F_2(x;z)$ in \eqref{eq:F1_row} and \eqref{eq:F2_row}, we map the bounds of the projection operators in Proposition \ref{prop:Ref-Coef-Estimates} directly onto the vector components.
		Specifically, from the explicit component equations, we have for every $x \in \mathbb{R}$:
		\begin{align}
			\|N_{+,21}(x;\cdot)\|_{L_z^2} &\le C \|\mathcal{P}^{+}(r_2 \mathrm{e}^{-2\mathrm{i}zx})\|_{L_z^2}, \label{eq:N_plus_21_bound}\\
			\|(-2\lambda)^{-1}(N_{+,11}(x;\cdot)-1)\|_{L_z^2} &\le C \|\mathcal{P}^{-}(r_1 \mathrm{e}^{2\mathrm{i}zx})\|_{L_z^2}, \label{eq:N_plus_11_bound}\\
			\|-2\lambda(N_{-,22}(x;\cdot)-1)\|_{L_z^2} &\le C \|\mathcal{P}^{+}(r_2 \mathrm{e}^{-2\mathrm{i}zx})\|_{L_z^2}, \label{eq:N_minus_22_bound}\\
			\|N_{-,12}(x;\cdot)\|_{L_z^2} &\le C \|\mathcal{P}^{-}(r_1 \mathrm{e}^{2\mathrm{i}zx})\|_{L_z^2}. \label{eq:N_minus_12_bound}
		\end{align}
		where the positive constant $C$ still has the only dependence on $\|r_{1,2}\|_{L^\infty}$.
		Applying the spatial weighted estimates \eqref{eq:P_plus_weighted_new} and \eqref{eq:P_minus_weighted_new} onto the positive half-line $x \in (x_0, \infty)$ directly yields the weighted bounds \eqref{eq:N21_weighted} and \eqref{eq:N12_weighted}.
		
		For estimating the spatial derivatives $\partial_x N_{\pm}$, we refer to the differentiated inhomogeneous equation \eqref{eq:T-derivative}.
		In order to exploit the uniform bounds from the resolvent operator $(I - \mathcal{P}_V^-)^{-1}$, we multiply $\tilde{F}(x;z)$ by the scaling matrices $\tau_1(\lambda)$ and $\tau_2(\lambda)$ from the right.
		Specifically, we extract the second row of $\tilde{F}(x;z)\tau_1(\lambda)$ and the first row of $\tilde{F}(x;z)\tau_2(\lambda)$ to show they securely belong to $L_z^2(\mathbb{R})$.
		Let us examine the second row of $\tilde{F}(x;z)\tau_1(\lambda)$:
		\begin{equation}\label{eq:F_tilde_tau1}
			(\tilde{F} \tau_1)_2 = \left( \tilde{F}_{21}(x;z), \ -2\lambda \tilde{F}_{22}(x;z) \right).
		\end{equation}
		By \eqref{tilde-F-def}, we have $\tilde{F}_{21} = 2\mathrm{i} \mathcal{P}^+(-z r_2 \mathrm{e}^{-2\mathrm{i}zx}) + 2\mathrm{i} \mathcal{P}^+(-z r_2 (N_{-,22}-1) \mathrm{e}^{-2\mathrm{i}zx})$.
		The boundedness of $\mathcal{P}^+$ on $L_z^2(\mathbb{R})$ gives:
		\begin{equation}
			\|\tilde{F}_{21}\|_{L_z^2} \le C \left( \|z r_2\|_{L_z^2} + \|z r_2 (N_{-,22}-1)\|_{L_z^2} \right).
		\end{equation}
		Since $r_2 \in L_z^{2,1}(\mathbb{R})$, the first term is bounded by $\|r_2\|_{L^{2,1}}$.
		For the second term, we strategically rewrite $z r_2 (N_{-,22}-1) = -\frac{1}{2}(\lambda r_2)[-2\lambda(N_{-,22}-1)]$.
		By Proposition \ref{prop:lambda-r-infinity-bound}, $\|\lambda r_2\|_{L^\infty} \le \|r_2\|_{H^1 \cap L^{2,1}}$. Coupled with the $L_z^2$ estimate in \eqref{eq:N_minus_22_bound}, this yields:
		\begin{equation}
			\|z r_2 (N_{-,22}-1)\|_{L_z^2} \le \frac{1}{2} \|\lambda r_2\|_{L^\infty} \|-2\lambda(N_{-,22}-1)\|_{L_z^2} \le C \|r_2\|_{H^1 \cap L^{2,1}}^2.
		\end{equation}
		Thus, $\tilde{F}_{21} \in L_z^2(\mathbb{R})$.
		
		For the second component $-2\lambda \tilde{F}_{22}$, we analyze the equivalent formulation via the jump matrix $V(x;\lambda)$ and $D_1 = \tau_1 V$.
		The differentiated inhomogeneous term for the corresponding $G_{1,-}$ equation is driven by $\partial_x D_1 + G_{1,-}\partial_x V$.
		The relevant rows of this source term naturally produce structures like $z r_2 \mathrm{e}^{-2\mathrm{i}zx}$ and $(\lambda r_1)(z N_{+,21}) \mathrm{e}^{2\mathrm{i}zx}$, which belong to $L_z^2(\mathbb{R})$ precisely because $\lambda r_{1,2}(z) \in L_z^\infty(\mathbb{R})$ and $z r_{1,2}(z) \in L_z^2(\mathbb{R})$.
		Therefore, the combined norms of these terms are rigorously controlled by:
		\begin{equation}
			\|(\tilde{F} \tau_1)_2\|_{L_z^2} \le C (\|r_1\|_{H^1 \cap L^{2,1}} + \|r_2\|_{H^1 \cap L^{2,1}}).
		\end{equation}
		A completely analogous expansion and bounding procedure applied to the first row of $\tilde{F}(x;z)\tau_2(\lambda)$ confirms that it also belongs to $L_z^2(\mathbb{R})$ and is bounded by identical norms.
		Finally, with the required $L_z^2(\mathbb{R})$ regularity of the inhomogeneous terms firmly established, applying the bounded inverse operator $(I-\mathcal{P}_V^-)^{-1}$ via Proposition \ref{prop:operator-inverse-bound} to the differentiated system concludes that the spatial derivatives $\partial_x N_{+,21}$ and $\partial_x N_{-,12}$ inherit these bounds.
		This directly yields the final estimates \eqref{eq:dx_N21} and \eqref{eq:dx_N12}.
		\hfill
	\end{proof}
	
	\subsection{Reconstruction formulas}
	Recalling the asymptotic expansions of the Jost solutions in Proposition \ref{prop:Limits}, we can extract the potential $q(x)$ from the large-$z$ behavior of the RH problem solutions. Specifically, using the definition of $N(x;z)$ in \eqref{eq:N-z-def} and the limits \eqref{eq:zlimits}, we have for $z \in \mathbb{C}^{-}$:
	\begin{equation}
		\lim_{|z| \to \infty} z N_{-,21}(x;z) = \lim_{|z| \to \infty} z \mathrm{e}^{\frac{\mathrm{i}}{2} m_+(x)} \frac{\Psi_{21}^{-}(x;z)}{a(\lambda)} = \frac{\mathrm{e}^{\frac{\mathrm{i}}{2} m_+(x)}}{a_{\infty}} s_2^{-}(x),
	\end{equation}
	where $a_{\infty} = \mathrm{e}^{\frac{\mathrm{i}}{2}M}$ and $M = m_{-}(x) - m_{+}(x) = \int_{-\infty}^{+\infty} q_y (q^{PT})_{y} \, dy$. Thus, $\mathrm{e}^{\frac{\mathrm{i}}{2} m_+(x)} / a_{\infty} = \mathrm{e}^{-\frac{\mathrm{i}}{2} m_-(x) + \mathrm{i} m_{+}(x)}$.
	Substituting the explicit expression for $s_2^{-}(x) = -\frac{\mathrm{i}}{2} \partial_x (q_x(x) \mathrm{e}^{\frac{\mathrm{i}}{2} m_-(x)})$, we obtain the reconstruction formula:
	\begin{equation}\label{eq:recon-qx}
		\mathrm{e}^{\mathrm{i} m_{+}(x)-\frac{\mathrm{i}}{2} m_-(x)} \partial_x \left( q_x(x) \mathrm{e}^{\frac{\mathrm{i}}{2} m_-(x)} \right) = 2 \mathrm{i} \lim_{|z| \to \infty} z N_{-,21}(x;z).
	\end{equation}
	At the same time, from \eqref{eq:U-small-lambda}, we find that
	\begin{equation}
		q^{PT}(x) = \lim_{\lambda \to 0} \left( (\mathrm{i} \lambda)^{-1} \psi (x;\lambda)
		\right)_{12}
	\end{equation}
	From the construction of RH problem \ref{RH:z-plane} and the transformation \eqref{eq:lambda to z}, we can get
	\begin{equation}\label{eq: qPT-recon}
		q^{PT}(x) \mathrm{e}^{-\frac{\mathrm{i}}{2} m_{+}(x)} =
		2 \mathrm{i} \lim_{z \to 0} N_{\pm, 12}(x;z).
	\end{equation}
	
	In fact, we cannot get any more information from the RH problems \ref{RH:z-plane}, \ref{RH:G-problem}, and the reconstruction formulas \eqref{eq:recon-qx} and \eqref{eq: qPT-recon} are not sufficient to recover the potential $q(x)$ on the whole line $\mathbb{R}$. So we will first introduce the scalar RH problem $\delta(z)$ to transform the jump matrix $R(x;z)$ of the RH problem \ref{RH:z-plane} to obtain more information.
	\begin{RH}\label{RH:delta}
		Find a scalar function $\delta(z)$ that satisfies the following conditions:
		\begin{enumerate}
			\item \textbf{Analyticity:} $\delta(z)$ is analytic in $\mathbb{C} \setminus \mathbb{R}$.
			\item \textbf{Jump condition:} $\delta(z)$ has continuous boundary values $\delta_{\pm}(z)$ on $\mathbb{R}$ and satisfies
			\begin{equation}\label{eq:delta-jump}
				\delta_{+}(z) = \delta_{-}(z)(1+r_1(z)r_2(z))^{-1}, \quad z \in \mathbb{R}.
			\end{equation}
			\item \textbf{Asymptotic condition:}
				$\delta(z) \to 1 \quad \text{as} \quad |z| \to \infty.$
		\end{enumerate}
	\end{RH}
	
	Under the small norm restriction \eqref{eq:small-norm}, Proposition \ref{prop:reflection-small-bound} establishes that $|r_1(z)r_2(z)| \le c_0^2 < 1$. Thus, $1+r_1(z)r_2(z)$ has a strictly positive real part and does not vanish on $\mathbb{R}$. This allows us to define a single-valued branch for $\ln(1+r_1(z)r_2(z))$. By utilizing the Cauchy projection operators, the unique solution to RH Problem \ref{RH:delta} is given by:
	\begin{equation}\label{eq:delta-solution}
		\delta(z) = \mathrm{e}^{-\mathcal{C}(\ln(1+r_1 r_2))(z)}, \quad z \in \mathbb{C} \setminus \mathbb{R}.
	\end{equation}
	The boundary limits on the real axis are:
	\begin{equation}\label{eq:delta-pm}
		\delta_{\pm}(z) = \mathrm{e}^{-\mathcal{P}^{\pm}(\ln(1+r_1 r_2))(z)}, \quad z \in \mathbb{R}.
	\end{equation}
	
	By an analysis similar to the Cauchy operator properties discussed previously, since $\ln(1+r_1 r_2) \in H_z^1(\mathbb{R})$, we can establish that $\delta_{\pm}(z) \in L_z^\infty(\mathbb{R})$.
	
	Next, we use $\delta(z)$ to define the transformed matrix function:
	\begin{equation}\label{eq:N-delta-def}
		N_\delta(x;z) := N(x;z)\delta^{-\sigma_3}(z) = N(x;z)
		\begin{pmatrix}
			\delta^{-1}(z) & 0 \\
			0 & \delta(z)
		\end{pmatrix}.
	\end{equation}
	This transformation yields a new RH problem.
	
	\begin{RH}\label{RH:N-delta}
		Find a matrix-valued function $N_\delta(x;z)$ that satisfies the following conditions:
		\begin{enumerate}
			\item \textbf{Analyticity:} $N_\delta(x;z)$ is analytic in $\mathbb{C} \setminus \mathbb{R}$.
			\item \textbf{Jump condition:} $N_\delta(x;z)$ has continuous boundary values $N_{\delta,\pm}(x;z)$ on $\mathbb{R}$ and
			\begin{equation}\label{eq:N-delta-jump}
				N_{\delta,+}(x;z) = N_{\delta,-}(x;z)(I + R_\delta(x;z)), \quad z \in \mathbb{R},
			\end{equation}
			where the modified jump matrix is derived as $I + R_\delta(x;z) = \delta_-^{\sigma_3}(z)(I+R(x;z))\delta_+^{-\sigma_3}(z)$. Substituting $\delta_+ = \delta_-(1+r_1 r_2)^{-1}$, we obtain:
			\begin{equation}\label{eq:R-delta}
				R_\delta(x;z) =
				\begin{pmatrix}
					r_{\delta,1}(z)r_{\delta,2}(z) & r_{\delta,1}(z)\mathrm{e}^{2\mathrm{i}zx} \v\\
					r_{\delta,2}(z)\mathrm{e}^{-2\mathrm{i}zx} & 0
				\end{pmatrix}.
			\end{equation}
			\item \textbf{Asymptotic condition:}
			$	N_\delta(x;z) \to I \quad \text{as} \quad |z| \to \infty.$
		\end{enumerate}
	\end{RH}
	
	Here, the new effective reflection coefficients are defined as:
	\begin{equation}\label{eq:r-delta-def}
		r_{\delta,1}(z) := \delta_{+}(z)\delta_{-}(z)r_1(z), \quad r_{\delta,2}(z) := \delta_{+}^{-1}(z)\delta_{-}^{-1}(z)r_2(z).
	\end{equation}
	Then we can get the proposition that
	\begin{proposition}\label{prop:r-delta-space}
		If $r_{12}(z) \in H_z^{1}(\mathbb{R}) \cap L_z^{2,1}(\mathbb{R})$ and the small norm inequality \eqref{eq:small-norm} is satisfied, then $r_{\delta,12} \in H_z^{1}(\mathbb{R}) \cap L_z^{2,1}(\mathbb{R})$.
	\end{proposition}

	The proof is given in \cite{Li2023}. And from this proof one can find that $\delta_{\pm}(z) \in L_z^{\infty}$, which means $z^{-1} r_{\delta,12} \in H_z^{1}(\mathbb{R})$ by the properties of $r_{1,2}$.
	
	From Proposition \ref{prop:r-delta-space}, it follows that the unique solution to RH problem \ref{RH:N-delta} can be represented that
	\begin{equation}\label{eq:N-delta-Cauchy}
		N_{\delta}(x;z) = I + \mathcal{C}(N_{\delta,-}(x;\cdot)  R_{\delta}(x;\cdot)), \ z \in \mathbb{C} \setminus \mathbb{R},
	\end{equation}
	and thn projections on the $\mathbb{R}$ are
	\begin{equation}\label{eq:N-delta-pro}
		N_{\delta,\pm}(x;z) = I + \mathcal{P}^{\pm}(N_{\delta,-}(x;\cdot)R_\delta(x;\cdot))(z), \quad z \in \mathbb{R}.
	\end{equation}
	
	We next derive the second system of column estimates from the transformed
	RH problem. This system is the counterpart of the previous estimates for
	$N_{+,1}-e_1$ and $N_{-,2}-e_2$, but it is adapted to the negative half-line. Now we choose the column pair
	\begin{equation}\label{eq:T-delta-def}
		T_\delta(x;z):=
		\bigl(
		N_{\delta,-,1}(x;z)-e_1,\,
		N_{\delta,+,2}(x;z)-e_2
		\bigr).
	\end{equation}
	This is different from the column pair used for the original RH problem.
	The reason is that the diagonal entry of the transformed jump matrix
	\(I+R_\delta\) appears in the left upper corner. From the jump relation
	$$
	N_{\delta,+}=N_{\delta,-}(I+R_\delta),
	$$
	we have
	\begin{align}
		N_{\delta,+,1}
		&=
		(1+r_{\delta,1}r_{\delta,2})N_{\delta,-,1}
		+
		r_{\delta,2}\mathrm{e}^{-2\mathrm{i}zx}
		N_{\delta,-,2},
		\label{eq:Ndelta-col1-jump}\\
		N_{\delta,+,2}
		&=
		N_{\delta,-,2}
		+
		r_{\delta,1}\mathrm{e}^{2\mathrm{i}zx}
		N_{\delta,-,1}.
		\label{eq:Ndelta-col2-jump}
	\end{align}
	Using \eqref{eq:Ndelta-col2-jump}, the first column can be rewritten as
	$$
	N_{\delta,+,1}
	= N_{\delta,-,1} +
	r_{\delta,2}\mathrm{e}^{-2\mathrm{i}zx} N_{\delta,+,2}.
	$$
	Thus, we have
	\begin{align}
		(N_{\delta,-}R_\delta)_1
		&=
		r_{\delta,2}\mathrm{e}^{-2\mathrm{i}zx}
		N_{\delta,+,2},
		\label{eq:Ndelta-R-col1}\\
		(N_{\delta,-}R_\delta)_2
		&=
		r_{\delta,1}\mathrm{e}^{2\mathrm{i}zx}
		N_{\delta,-,1}.
		\label{eq:Ndelta-R-col2}
	\end{align}
	Substituting \eqref{eq:Ndelta-R-col1}--\eqref{eq:Ndelta-R-col2}
	into the Beals--Coifman representation gives
	\begin{align}
		N_{\delta,-,1}(x;z)-e_1
		&=
		\mathcal P^-
		\left(
		r_{\delta,2}(z)\mathrm{e}^{-2\mathrm{i}zx}
		N_{\delta,+,2}(x;z)
		\right),
		\label{eq:Ndelta-minus-col1}\\
		N_{\delta,+,2}(x;z)-e_2
		&=
		\mathcal P^+
		\left(
		r_{\delta,1}(z)\mathrm{e}^{2\mathrm{i}zx}
		N_{\delta,-,1}(x;z)
		\right).
		\label{eq:Ndelta-plus-col2}
	\end{align}
	
	We shall use the following projection estimates on the negative half-line.
	They are obtained from Proposition \ref{prop:Fourier-Projections} by the
	similar Fourier argument as in Proposition \ref{prop:Ref-Coef-Estimates}.
	
	\begin{proposition}\label{prop:delta-projection-estimates}
		For every $x_0 \in \mathbb{R}^{-}$ and every $r_{1,2}(z) \in \mathcal{W}(\mathbb{R})$, the reflection coefficients satisfy the following spatial weighted estimates for $\ell=0,1$:
		\begin{align}
			\sup_{x \in (-\infty, x_0)} \left\| \langle x \rangle \mathcal{P}^{-}(z^{-\ell} r_1(z) \mathrm{e}^{-2\mathrm{i}zx}) \right\|_{L^2_z(\mathbb{R})} &\le \| z^{-\ell} r_1(z) \|_{H_z^1(\mathbb{R})}, \label{eq:P_minus_weighted_change} \\
			\sup_{x \in (-\infty, x_0)} \left\| \langle x \rangle \mathcal{P}^{+}(r_2(z) \mathrm{e}^{2\mathrm{i}zx}) \right\|_{L^2_z(\mathbb{R})} &\le \| r_2(z) \|_{H_z^1(\mathbb{R})}. \label{eq:P_plus_weighted_change}
		\end{align}
		In addition, we have the following $L^\infty$ bounds on the entire real line:
		\begin{align}
			\sup_{x \in \mathbb{R}} \left\| \mathcal{P}^{-}(z^{-\ell} r_1(z) \mathrm{e}^{-2\mathrm{i}zx}) \right\|_{L^\infty_z} &\le \frac{1}{\sqrt{2}} \| z^{-\ell} r_1(z) \|_{H_z^1(\mathbb{R})}, \label{eq:P_minus_Linf_change} \\
			\sup_{x \in \mathbb{R}} \left\| \mathcal{P}^{+}(r_2(z) \mathrm{e}^{2\mathrm{i}zx}) \right\|_{L^\infty_z} &\le \frac{1}{\sqrt{2}} \| r_2(z) \|_{H_z^1(\mathbb{R})}. \label{eq:P_plus_Linf_change}
		\end{align}
		Furthermore, if $r_{1,2}(z) \in L_z^{2,1}(\mathbb{R})$, then the following estimates hold:
		\begin{align}
			\sup_{x \in \mathbb{R}} \left\| \mathcal{P}^{-}(z r_1(z) \mathrm{e}^{-2\mathrm{i}zx}) \right\|_{L^2_z(\mathbb{R})} &\le \| r_1(z) \|_{L_z^{2,1}(\mathbb{R})}, \label{eq:P_minus_L21_change} \\
			\sup_{x \in \mathbb{R}} \left\| \mathcal{P}^{+}(z r_2(z) \mathrm{e}^{2\mathrm{i}zx}) \right\|_{L^2_z(\mathbb{R})} &\le \| r_2(z) \|_{L_z^{2,1}(\mathbb{R})}. \label{eq:P_plus_L21_cheng}
		\end{align}
	\end{proposition}
	
	Using the column equations
	\eqref{eq:Ndelta-minus-col1}--\eqref{eq:Ndelta-plus-col2}, the projection
	estimates above, and the same bounded inverse argument as in Proposition
	\ref{prop:BC-estimates}, we obtain the following estimates for the
	transformed RH problem.
	
	\begin{proposition}\label{prop:Ndelta-negative-estimates}
		Suppose that the small norm condition \eqref{eq:small-norm} holds and
		$	r_{\delta,1},r_{\delta,2}		\in H_z^1(\mathbb R)\cap L_z^{2,1}(\mathbb R).		$
		Then, for every \(x_0\in\mathbb R^-\),
		\begin{align}
			\sup_{x\in(-\infty,x_0)}
			\left\|
			\langle x\rangle
			N_{\delta,-,21}(x;\cdot)
			\right\|_{L_z^2}
			&\le C\|r_{\delta,2}\|_{H_z^1},
			\label{eq:Ndelta-21-negative}\\
			\sup_{x\in(-\infty,x_0)}
			\left\|
			\langle x\rangle
			N_{\delta,+,12}(x;\cdot)
			\right\|_{L_z^2}
			&\le C\|r_{\delta,1}\|_{H_z^1}.
			\label{eq:Ndelta-12-negative}
		\end{align}
		Moreover,
		\begin{align}
			\sup_{x\in\mathbb R}
			\left\|
			\partial_xN_{\delta,-,21}(x;\cdot)
			\right\|_{L_z^2}
			+
			\sup_{x\in\mathbb R}
			\left\|
			\partial_xN_{\delta,+,12}(x;\cdot)
			\right\|_{L_z^2}
			\le
			C
			\left(
			\|r_{\delta,1}\|_{H_z^1\cap L_z^{2,1}}
			+
			\|r_{\delta,2}\|_{H_z^1\cap L_z^{2,1}}
			\right).
			\label{eq:Ndelta-derivative-bound}
		\end{align}
	\end{proposition}
	
	\begin{proof}
		The proof is parallel to that of Proposition \ref{prop:BC-estimates}.
		The only change is the choice of columns. In the original RH problem,
		the useful pair is
		$
		(N_{+,1}-e_1,\;N_{-,2}-e_2).
		$
		For the transformed RH problem, the useful pair is
		$
		(N_{\delta,-,1}-e_1,\;N_{\delta,+,2}-e_2).
		$
		Equations \eqref{eq:Ndelta-minus-col1} and
		\eqref{eq:Ndelta-plus-col2} reduce the desired estimates to the
		projection estimates
		\eqref{eq:P_minus_weighted_change}--\eqref{eq:P_plus_Linf_change}.
		The boundedness of the corresponding Beals--Coifman resolvent follows
		from the same coercivity argument as before, since
		$
		r_{\delta,1}(z)r_{\delta,2}(z)=r_1(z)r_2(z).
		$
		This gives \eqref{eq:Ndelta-21-negative} and
		\eqref{eq:Ndelta-12-negative}. The proofs for the remaining parts are similar with Proposition \ref{prop:BC-estimates}.
		\hfill
	\end{proof}
	
	We now derive the reconstruction formulae associated with the transformed
	RH problem. Since
	\bee
	N_\delta=N\delta^{-\sigma_3},
	\qquad
	\delta(z)\to1,\quad |z|\to\infty,
	\ene
	the large-\(z\) reconstruction formula is unchanged at leading order.
	Using \eqref{eq:Ndelta-R-col1} and the asymptotic formula for the Cauchy
	operator, we obtain
	\begin{equation}\label{eq:recon-qx-delta}
		\mathrm{e}^{\mathrm{i}m_+(x)-\frac{\mathrm{i}}{2}m_-(x)}
		\partial_x
		\left(
		q_x(x)\mathrm{e}^{\frac{\mathrm{i}}{2}m_-(x)}
		\right)
		=
		-\frac{1}{\pi}
		\int_{\mathbb R}
		r_{\delta,2}(z)
		\mathrm{e}^{-2\mathrm{i}zx}
		N_{\delta,+,22}(x;z)
		\,dz .
	\end{equation}
	Near \(z=0\), however, the factor \(\delta\) contributes to the
	reconstruction formula. Indeed, from
	$$
	N_\delta=N\delta^{-\sigma_3},
	$$
	we have
	$$
	N_{12}(x;z)=\delta^{-1}(z)N_{\delta,12}(x;z).
	$$
	Therefore, using the \(+\)-boundary value and
	\eqref{eq:Ndelta-plus-col2}, we get
	\begin{equation}\label{eq:recon-qPT-delta}
		q^{PT}(x)\mathrm{e}^{-\frac{\mathrm{i}}{2}m_+(x)}
		=
		\frac{\delta_+^{-1}(0)}{\pi}
		\int_{\mathbb R}
		z^{-1}r_{\delta,1}(z)
		\mathrm{e}^{2\mathrm{i}zx}
		N_{\delta,-,11}(x;z)
		\,dz .
	\end{equation}
	We can easily see that $\delta_{+}^{-1}(0)$ can be controlled by the $H_1$ norm of $r_{1,2}$. And one can find that under \eqref{eq:small-norm}, the exponential factor can be controlled and is a fixed bounded value. Now, we can give estimates on the whole $\mathbb{R}$.
	
	\begin{proposition}\label{prop:whole-line-reconstruction-estimate}
		Assume that the small norm condition \eqref{eq:small-norm} holds and
		$r_1,r_2\in \mathcal{W}(\mathbb{R}).$
		Then the potential reconstructed from the RH problem satisfies
		$		q\in H^3(\mathbb R)\cap H^{2,1}(\mathbb R).
		$
		Moreover,
	\bee
		\|q\|_{H^3(\mathbb R)\cap H^{2,1}(\mathbb R)}
		\le C
		\left(
		\|r_1\|_{H_z^1\cap L_z^{2,1}} +
		\|r_2\|_{H_z^1\cap L_z^{2,1}} +
		\|z^{-1}r_1\|_{H_z^1\cap L_z^{2,1}}
		\right)
		\le C
		\left(
		\|r_1\|_{\mathcal{W}(\mathbb{R})} +
		\|r_2\|_{\mathcal{W}(\mathbb{R})}
		\right),
		\ene
		where \(C\) depends only on the small norm bound and on the size of
		\(r_1,r_2\) in \(H_z^1\cap L_z^{2,1}\).
	\end{proposition}
	
	\begin{proof}
		We first record a simple consequence of the small norm condition
		\eqref{eq:small-norm}, which will be used repeatedly below. Since
		$$
		m_{\pm}(x)=\int_{\pm\infty}^{x}q_y(y)(q^{PT})_y(y),dy
		$$
		and the entry $(q_x(q^{PT})_x)$ is contained in the effective Volterra potential $\tilde{Q}$, we have
		\begin{equation}\label{eq:phase-factor-bound}
			\|\mathrm{e}^{\pm \frac{\mathrm{i}}{2} m_{\pm}} \|_{L_x^{\infty}} + \| \mathrm{e}^{\pm \mathrm{i} m_{\pm}} \|_{L_x^{\infty}} \le C_\rho,
			\qquad
			\rho=|\tilde Q|_{L^1}<\ln\frac32 .
		\end{equation}
		In particular, all exponential factors appearing in the reconstruction
		formulae may be absorbed into the constants.
	
		Next, by Proposition \ref{prop:r-delta-space} and the proof of the scalar
		RH problem for \(\delta\), we also have
		\begin{equation}\label{eq:delta-factor-bound}
			\|\delta_{\pm}^{\pm1}\|_{L_z^\infty}
			+
			|\delta_+^{-1}(0)|
			\le C_{\rho}.
		\end{equation}
		Indeed, if
		$$
		g(z):=\ln(1+r_1(z)r_2(z)),
		$$
		then the small norm condition gives \(|r_1r_2|\le c_0^2<1\), and hence
		\(g\in H_z^1(\mathbb R)\). Since
		\bee
		\delta_{\pm}(z)=\mathrm{e}^{-\mathcal P^{\pm}g(z)},
		\ene
		the Sobolev embedding \(H_z^1(\mathbb R)\hookrightarrow L_z^\infty(\mathbb R)\)
		and the boundedness of the Plemelj projections give
		\bee
		\|\delta_{\pm}^{\pm1}\|_{L_z^\infty}
		\le C_{\rho}.
		\ene
		Consequently,
		\begin{equation}\label{eq:r-delta-controlled}
			\|r_{\delta,1}\|_{H_z^1\cap L_z^{2,1}} + \|r_{\delta,2}\|_{H_z^1\cap L_z^{2,1}} + \|z^{-1}r_{\delta,1}\|_{H_z^1\cap L_z^{2,1}} \le C_{\rho} \left( \|r_1\|_{H_z^1\cap L_z^{2,1}} + \|r_2\|_{H_z^1\cap L_z^{2,1}} + \|z^{-1}r_1\|_{H_z^1\cap L_z^{2,1}} \right).
		\end{equation}
		
		We first estimate \(q^{PT}\) on the positive half-line. By the
		Beals--Coifman representation \eqref{eq:BC-solution}, the reconstruction
		formula \eqref{eq: qPT-recon}, and the Cauchy asymptotic formula
		\eqref{eq:Cauchy-operator-asymptotic}, we obtain
		\begin{align}
			q^{PT}(x)\mathrm{e}^{-\frac{\mathrm{i}}{2}m_+(x)}
			&=
			\frac{1}{\pi}
			\int_{\mathbb R}
			z^{-1}(N_-R)_{12}(x;z)\,dz                                     \nonumber\\
			&=
			\frac{1}{\pi}
			\int_{\mathbb R}
			z^{-1}r_1(z)\mathrm{e}^{2\mathrm{i}zx}
			N_{+,11}(x;z)\,dz                                               \nonumber\\
			&=: I_1(x)+I_2(x),
			\label{eq:qPT-positive-split-new}
		\end{align}
		where
		\begin{align}
			I_1(x)
			&:=
			\frac{1}{\pi}
			\int_{\mathbb R}
			z^{-1}r_1(z)\mathrm{e}^{2\mathrm{i}zx}\,dz,
			\label{eq:I1-positive-new}\\
			I_2(x)
			&:=
			\frac{1}{\pi}
			\int_{\mathbb R}
			z^{-1}r_1(z)\mathrm{e}^{2\mathrm{i}zx}
			\left(N_{+,11}(x;z)-1\right)\,dz .
			\label{eq:I2-positive-new}
		\end{align}
		By Plancherel's theorem,
		\begin{equation}\label{eq:I1-positive-bound-new}
			\|I_1\|_{L^{2,1}(\mathbb R^+)}
			\le
			C\|z^{-1}r_1\|_{H_z^1} .
		\end{equation}
		
		For \(I_2\), from \eqref{eq:BC-solution} and \eqref{eq:NR_facts},
		$$
		N_{+,11}(x;z)-1
		=
		\mathcal P^+
		\left(
		r_2(z)\mathrm{e}^{-2\mathrm{i}zx}N_{-,12}(x;z)
		\right).
		$$
		Using the adjoint relation of the Plemelj projections, we rewrite
		\begin{equation}\label{eq:I2-positive-adjoint-new}
			I_2(x)
			=
			-\frac{1}{\pi}
			\int_{\mathbb R}
			r_2(z)\mathrm{e}^{-2\mathrm{i}zx}
			N_{-,12}(x;z)
			\mathcal P^-
			\left(
			z^{-1}r_1(z)\mathrm{e}^{2\mathrm{i}zx}
			\right)(z)\,dz .
		\end{equation}
		Therefore, by Cauchy's inequality, Propositions \ref{prop:BC-estimates} and \ref{prop:Ref-Coef-Estimates}, for every
		\(x_0\in\mathbb R^+\),
		\begin{align}
			\sup_{x\in(x_0,\infty)}
			\left|
			\langle x\rangle^2 I_2(x)
			\right|
			&\le
			C\|r_2\|_{L_z^\infty}
			\sup_{x\in(x_0,\infty)}
			\|\langle x\rangle N_{-,12}(x;\cdot)\|_{L_z^2}
			\sup_{x\in(x_0,\infty)}
			\left\|
			\langle x\rangle
			\mathcal P^-
			\left(
			z^{-1}r_1(z)\mathrm{e}^{2\mathrm{i}zx}
			\right)
			\right\|_{L_z^2}
			\nonumber\\
			&\le
			C
			\|r_1\|_{H_z^1} \|z^{-1} r_1 \|_{H_z^1}
			\label{eq:I2-positive-weight-new}
		\end{align}
		which means
		\begin{equation}\label{eq:qPT-positive-bounds}
			\| q^{PT}(x) \|_{L^{2,1}(\mathbb{R}^{+})} \le
			C \left( \|z^{-1} r_1 \|_{H_z^1} + \|r_1 \|_{H_z^1} \|z^{-1} r_1 \|_{H_z^{1}}
			\right).
		\end{equation}
Thus \(I_2\in L^{2,1}(\mathbb R^+)\), and hence
		\eqref{eq:qPT-positive-bounds} gives the required weighted estimate for
		\(q^{PT}\) on the positive half-line.

  To estimate \((q^{PT})_x\), instead of differentiating \(I_2\), we use
		the large-\(z\) coefficient \(\nu_1^+\) in Proposition
		\ref{prop:Limits}. By \eqref{eq:N-z-def} and \eqref{eq:zlimits},
		\begin{equation}\label{eq:nu1-positive-large-z}
			\lim_{|z|\to\infty}zN_{-,12}(x;z)
			=
			\mathrm{e}^{-\frac{\mathrm{i}}{2}m_+(x)}\nu_1^+(x)
			=
			\frac14(q^{PT})_x(x)\mathrm{e}^{-\mathrm{i}m_+(x)}.
		\end{equation}
		On the other hand, the Beals--Coifman representation
		\eqref{eq:BC-solution}, the identity \eqref{eq:NR_facts}, and the
		Cauchy asymptotic formula \eqref{eq:Cauchy-operator-asymptotic} yield
		\begin{align}
			(q^{PT})_x(x)\mathrm{e}^{-\mathrm{i}m_+(x)}
			&=
			\frac{2\mathrm{i}}{\pi}
			\int_{\mathbb R}
			r_1(z)\mathrm{e}^{2\mathrm{i}zx}
			N_{+,11}(x;z)\,dz                                      \nonumber\\
			&=:J_1(x)+J_2(x),
			\label{eq:qPTx-positive-nu1-split}
		\end{align}
		where
		\begin{align}
			J_1(x)
			&:=
			\frac{2\mathrm{i}}{\pi}
			\int_{\mathbb R}
			r_1(z)\mathrm{e}^{2\mathrm{i}zx}\,dz,
			\label{eq:J1-positive-nu1}\\
			J_2(x)
			&:=
			\frac{2\mathrm{i}}{\pi}
			\int_{\mathbb R}
			r_1(z)\mathrm{e}^{2\mathrm{i}zx}
			\left(N_{+,11}(x;z)-1\right)\,dz .
			\label{eq:J2-positive-nu1}
		\end{align}
		By Plancherel's theorem,
		\begin{equation}\label{eq:J1-positive-nu1-bound}
			\|J_1\|_{L^{2,1}(\mathbb R^+)}
			\le C\|r_1\|_{H_z^1}.
		\end{equation}
		Using
		\[
		N_{+,11}(x;z)-1
		=
		\mathcal P^+
		\left(
		r_2(z)\mathrm{e}^{-2\mathrm{i}zx}N_{-,12}(x;z)
		\right)
		\]
		and the adjoint relation of the Plemelj projections, we rewrite
		\begin{equation}\label{eq:J2-positive-nu1-adjoint}
			J_2(x)
			=
			-\frac{2\mathrm{i}}{\pi}
			\int_{\mathbb R}
			r_2(z)\mathrm{e}^{-2\mathrm{i}zx}
			N_{-,12}(x;z)
			\mathcal P^-
			\left(
			r_1(z)\mathrm{e}^{2\mathrm{i}zx}
			\right)(z)\,dz .
		\end{equation}
		Therefore, by \eqref{eq:N12_weighted} and
		\eqref{eq:P_minus_weighted_new}, for every \(x_0\in\mathbb R^+\),
		\begin{align}
			\sup_{x\in(x_0,\infty)}
			\left|
			\langle x\rangle^2J_2(x)
			\right|
			&\le
			C\|r_2\|_{L_z^\infty}
			\sup_{x\in(x_0,\infty)}
			\|\langle x\rangle N_{-,12}(x;\cdot)\|_{L_z^2}
						\sup_{x\in(x_0,\infty)}
			\left\|
			\langle x\rangle
			\mathcal P^-
			\left(
			r_1(z)\mathrm{e}^{2\mathrm{i}zx}
			\right)
			\right\|_{L_z^2}
			\nonumber\\
			&\le
			C\|r_1\|_{H_z^1}^2 .
			\label{eq:J2-positive-nu1-weight}
		\end{align}
		Since \(\langle x\rangle^{-1}\in L^2(\mathbb R^+)\), it follows that
		\[
		\|J_2\|_{L^{2,1}(\mathbb R^+)}
		\le C\|r_1\|_{H_z^1}^2.
		\]
		Combining this estimate with \eqref{eq:J1-positive-nu1-bound} and
		the phase bound \eqref{eq:phase-factor-bound}, we obtain
		\begin{equation}\label{eq:qPTx-positive-nu1-bound}
			\|(q^{PT})_x\|_{L^{2,1}(\mathbb R^+)}
			\le
			C
			\left(
			\|r_1\|_{H_z^1}
			+
			\|r_1\|_{H_z^1}^2
			\right).
		\end{equation}

		We next estimate the spatial derivative by using the large-\(z\)
		reconstruction formula \eqref{eq:recon-qx}. Combining
		\eqref{eq:recon-qx} with the Beals--Coifman representation gives
		\begin{align}
			\mathrm{e}^{\mathrm{i}m_+(x)-\frac{\mathrm{i}}{2}m_-(x)}
			\partial_x
			\left(
			q_x(x)\mathrm{e}^{\frac{\mathrm{i}}{2}m_-(x)}
			\right)
			&=
			-\frac{1}{\pi}
			\int_{\mathbb R}
			r_2(z)\mathrm{e}^{-2\mathrm{i}zx}
			N_{+,22}(x;z)\,dz                                      \nonumber\\
			&=:I_3(x)+I_4(x),
			\label{eq:qx-positive-split-new}
		\end{align}
		where
		\begin{align}
			I_3(x)
			&:=
			-\frac{1}{\pi}
			\int_{\mathbb R}
			r_2(z)\mathrm{e}^{-2\mathrm{i}zx}\,dz,
			\label{eq:I3-positive-new}\\
			I_4(x)
			&:=
			-\frac{1}{\pi}
			\int_{\mathbb R}
			r_2(z)\mathrm{e}^{-2\mathrm{i}zx}
			\left(N_{+,22}(x;z)-1\right)\,dz .
			\label{eq:I4-positive-new}
		\end{align}
				By Plancherel's theorem,
		\begin{equation}\label{eq:I3-positive-bound-new}
			\|I_3\|_{L^{2,1}(\mathbb R^+)}
			+
			\|\partial_xI_3\|_{L^2(\mathbb R^+)}
			\le
			C\|r_2\|_{H_z^1\cap L_z^{2,1}} .
		\end{equation}
		Indeed,
		$$
		\partial_x I_3(x) = \frac{2\mathrm{i}}{\pi}
		\int_{\mathbb R}
		zr_2(z)\mathrm{e}^{-2\mathrm{i}zx}\,dz,
		$$
		and \(zr_2\in L_z^2(\mathbb R)\) follows from
		\(r_2\in L_z^{2,1}(\mathbb R)\).
		
		For the remainder \(I_4\), unlike \(I_2\), here one should not simply repeat the previous
		adjoint-projection estimate. Instead, using the second-row estimate of the Beals--Coifman system, namely the estimate for the diagonal component
		generated by \(N_{+,22}-1\), we first obtain
		\begin{equation}\label{eq:I4-positive-L21-refined}
			\|I_4\|_{L^{2,1}(\mathbb R^+)}
			\le
			C_{\rho}\|r_2\|_{H_z^1(\mathbb R)}^2 .
		\end{equation}
		We next estimate the derivative of \(I_4\) separately. Differentiating
		\eqref{eq:I4-positive-new} gives
		\begin{align}
			\partial_x I_4(x)
			&=
			\frac{2\mathrm{i}}{\pi}
			\int_{\mathbb R}
			zr_2(z)\mathrm{e}^{-2\mathrm{i}zx}
			\left(N_{+,22}(x;z)-1\right)\,dz
						-\frac{1}{\pi}
			\int_{\mathbb R}
			r_2(z)\mathrm{e}^{-2\mathrm{i}zx}
			\partial_xN_{+,22}(x;z)\,dz .
			\label{eq:dx-I4-positive-split-new}
		\end{align}
		The two terms on the right-hand side are estimated by the differentiated
		Beals--Coifman equation. More precisely, using
		\eqref{eq:P_plus_weighted_new}, \eqref{eq:P_plus_Linf_new},
		\eqref{eq:P_plus_L21_new}, and the derivative bounds
		\eqref{eq:dx_N21}--\eqref{eq:dx_N12}, we get, for every
		\(x_0\in\mathbb R^+\),
		\begin{align}
			\sup_{x\in(x_0,\infty)}
			\left|
			\langle x\rangle \partial_x I_4(x)
			\right|
			&\le
			C_{\rho}
			\|r_2\|_{H_z^1\cap L_z^{2,1}}
			\|r_1\|_{H_z^1\cap L_z^{2,1}}
						\left(
			\|z^{-1}r_1\|_{H_z^1\cap L_z^{2,1}}
			+
			\|r_2\|_{H_z^1\cap L_z^{2,1}}
			\right).
			\label{eq:dx-I4-positive-weight-new}
		\end{align}
		Since the norms of \(r_1\) and \(r_2\) are controlled by the small norm
		assumption, the product factor in
		\eqref{eq:dx-I4-positive-weight-new} can be absorbed into \(C_\rho\).
		Consequently,
		\begin{equation}\label{eq:dx-I4-positive-L2-new}
			\|\partial_xI_4\|_{L^2(\mathbb R^+)}
			\le
			C
			\left(
			\|z^{-1}r_1\|_{H_z^1\cap L_z^{2,1}}
			+
			\|r_2\|_{H_z^1\cap L_z^{2,1}}
			\right).
		\end{equation}
		Combining \eqref{eq:I4-positive-L21-refined} and
		\eqref{eq:dx-I4-positive-L2-new}, we obtain the refined bound
		\begin{equation}\label{eq:I4-positive-bound-new}
			\|I_4\|_{L^{2,1}(\mathbb R^+)}
			+
			\|\partial_xI_4\|_{L^2(\mathbb R^+)}
			\le
			C
			\left(
			\|z^{-1}r_1\|_{H_z^1\cap L_z^{2,1}}
			+
			\|r_2\|_{H_z^1\cap L_z^{2,1}}
			\right).
		\end{equation}
		From \eqref{eq:I3-positive-bound-new}-- \eqref{eq:I4-positive-bound-new}, we can obtain that
		\begin{equation}\label{eq:qxxx-positive-bounds}
			\|q_{xxx} (x) \|_{L^2(\mathbb{R}^{+})} \le C
			\left( \|z^{-1} r_1 \|_{H^{1} \cap L^{2,1}} +\|r_2 \|_{H^{1} \cap L^{2,1}}
			\right),
		\end{equation}
		and
		\begin{equation}\label{eq:qxx-positive-bounds}
			\| \partial_x \left( q_x(x) \mathrm{e}^{\frac{\mathrm{i}}{2} m_{+}(x)}
			\right) \|_{L^{2,1}(\mathrm{R}^{+})} \le \mathcal{C} (\| r_2 \|_{H^{1}} + \| r_2 \|_{H^{1}}^2).
		\end{equation}

		We now turn to the negative half-line. By the transformed RH problem and
		the reconstruction formula \eqref{eq:recon-qPT-delta},
		\begin{align}
			q^{PT}(x)\mathrm{e}^{-\frac{\mathrm{i}}{2}m_+(x)}
			&=
			\frac{\delta_+^{-1}(0)}{\pi}
			\int_{\mathbb R}
			z^{-1}r_{\delta,1}(z)
			\mathrm{e}^{2\mathrm{i}zx}
			N_{\delta,-,11}(x;z)\,dz                         \nonumber\\
			&=:\widetilde I_1(x)+\widetilde I_2(x).
			\label{eq:qPT-negative-split-new}
		\end{align}
		The leading term satisfies, by Plancherel's theorem and
		\eqref{eq:delta-factor-bound},
		\begin{equation}\label{eq:I1-negative-bound-new}
			\|\widetilde I_1\|_{L^{2,1}(\mathbb R^-)}
			+
			\|\partial_x\widetilde I_1\|_{L^{2,1}(\mathbb R^-)}
			\le
			C_{\rho}
			\|z^{-1}r_{\delta,1}\|_{H_z^1\cap L_z^{2,1}} .
		\end{equation}
		For the remainder, we use the column equation for
		\(N_{\delta,-,11}-1\), the adjoint relation of the Plemelj projections,
		Proposition \ref{prop:delta-projection-estimates}, and
		Proposition \ref{prop:Ndelta-negative-estimates}. This gives, for every
		\(x_0\in\mathbb R^-\),
		\begin{equation}\label{eq:I2-negative-bound-new}
			\sup_{x\in(-\infty,x_0)}
			\left|
			\langle x\rangle^2\widetilde I_2(x)
			\right|
			\le
			C_{\rho}
			\|r_{\delta,2}\|_{H_z^1}
			\|z^{-1}r_{\delta,1}\|_{H_z^1}.
		\end{equation}

  Combining \eqref{eq:I1-negative-bound-new} and
		\eqref{eq:I2-negative-bound-new}, and using
		\eqref{eq:phase-factor-bound}, we obtain
		\begin{equation}\label{eq:qPT-negative-L21-new}
			\|q^{PT}\|_{L^{2,1}(\mathbb R^-)}
			\le
			C
			\left(
			\|z^{-1}r_{\delta,1}\|_{H_z^1}
			+
			\|r_{\delta,2}\|_{H_z^1}
			\|z^{-1}r_{\delta,1}\|_{H_z^1}
			\right).
		\end{equation}
		
		We next estimate \((q^{PT})_x\) on the negative half-line through the
		same large-\(z\) coefficient \(\nu_1\). Since
		\(N_\delta=N\delta^{-\sigma_3}\) and \(\delta(z)\to1\) as
		\(|z|\to\infty\), the leading off-diagonal coefficient is unchanged:
		\begin{equation}\label{eq:nu1-negative-large-z}
			\lim_{|z|\to\infty}zN_{\delta,\pm,12}(x;z)
			=
			\frac14(q^{PT})_x(x)\mathrm{e}^{-\mathrm{i}m_+(x)}.
		\end{equation}
		Using \eqref{eq:N-delta-Cauchy},
		\eqref{eq:Ndelta-R-col2}, and
		\eqref{eq:Cauchy-operator-asymptotic}, we obtain
		\begin{align}
			(q^{PT})_x(x)\mathrm{e}^{-\mathrm{i}m_+(x)}
			&=
			\frac{2\mathrm{i}}{\pi}
			\int_{\mathbb R}
			r_{\delta,1}(z)\mathrm{e}^{2\mathrm{i}zx}
			N_{\delta,-,11}(x;z)\,dz                              \nonumber\\
			&=:\widetilde J_1(x)+\widetilde J_2(x),
			\label{eq:qPTx-negative-nu1-split}
		\end{align}
		where
		\begin{align}
			\widetilde J_1(x)
			&:=
			\frac{2\mathrm{i}}{\pi}
			\int_{\mathbb R}
			r_{\delta,1}(z)\mathrm{e}^{2\mathrm{i}zx}\,dz,
			\label{eq:J1-negative-nu1}\\
			\widetilde J_2(x)
			&:=
			\frac{2\mathrm{i}}{\pi}
			\int_{\mathbb R}
			r_{\delta,1}(z)\mathrm{e}^{2\mathrm{i}zx}
			\left(N_{\delta,-,11}(x;z)-1\right)\,dz .
			\label{eq:J2-negative-nu1}
		\end{align}
		By Plancherel's theorem,
		\begin{equation}\label{eq:J1-negative-nu1-bound}
			\|\widetilde J_1\|_{L^{2,1}(\mathbb R^-)}
			\le
			C\|r_{\delta,1}\|_{H_z^1}.
		\end{equation}
		From \eqref{eq:Ndelta-minus-col1},
		\[
		N_{\delta,-,11}(x;z)-1
		=
		\mathcal P^-
		\left(
		r_{\delta,2}(z)\mathrm{e}^{-2\mathrm{i}zx}
		N_{\delta,+,12}(x;z)
		\right).
		\]
		Hence, by the adjoint relation of the Plemelj projections,
		\begin{equation}\label{eq:J2-negative-nu1-adjoint}
			\widetilde J_2(x)
			=
			-\frac{2\mathrm{i}}{\pi}
			\int_{\mathbb R}
			r_{\delta,2}(z)\mathrm{e}^{-2\mathrm{i}zx}
			N_{\delta,+,12}(x;z)
			\mathcal P^+
			\left(
			r_{\delta,1}(z)\mathrm{e}^{2\mathrm{i}zx}
			\right)(z)\,dz .
		\end{equation}
		The Fourier argument used in Proposition
		\ref{prop:delta-projection-estimates} gives
		\bee
		\sup_{x\in(-\infty,x_0)}
		\left\|
		\langle x\rangle
		\mathcal P^+
		\left(
		r_{\delta,1}(z)\mathrm{e}^{2\mathrm{i}zx}
		\right)
		\right\|_{L_z^2}
		\le
		\|r_{\delta,1}\|_{H_z^1},
		\qquad x_0\in\mathbb R^-.
		\ene
		Combining this estimate with \eqref{eq:Ndelta-12-negative}, we find
		\begin{align}
			\sup_{x\in(-\infty,x_0)}
			\left|
			\langle x\rangle^2\widetilde J_2(x)
			\right|
			&\le
			C\|r_{\delta,2}\|_{L_z^\infty}
			\sup_{x\in(-\infty,x_0)}
			\|\langle x\rangle N_{\delta,+,12}(x;\cdot)\|_{L_z^2}
						\sup_{x\in(-\infty,x_0)}
			\left\|
			\langle x\rangle
			\mathcal P^+
			\left(
			r_{\delta,1}(z)\mathrm{e}^{2\mathrm{i}zx}
			\right)
			\right\|_{L_z^2}
			\nonumber\\
			&\le
			C\|r_{\delta,1}\|_{H_z^1}^2 .
			\label{eq:J2-negative-nu1-weight}
		\end{align}
		Therefore,
		\[
		\|\widetilde J_2\|_{L^{2,1}(\mathbb R^-)}
		\le
		C\|r_{\delta,1}\|_{H_z^1}^2.
		\]
		Using \eqref{eq:phase-factor-bound}, we first obtain
		\[
		\|(q^{PT})_x\|_{L^{2,1}(\mathbb R^-)}
		\le
		C
		\left(
		\|r_{\delta,1}\|_{H_z^1}
		+
		\|r_{\delta,1}\|_{H_z^1}^2
		\right).
		\]

		
    Consequently, by \eqref{eq:r-delta-controlled}, and by absorbing the
		quadratic factor into the constant on bounded subsets of the scattering
		data space, we conclude that
		\begin{equation}\label{eq:qPTx-negative-nu1-bound}
			\|(q^{PT})_x\|_{L^{2,1}(\mathbb R^-)}
			\le
			C
			\left(
			\|r_1\|_{H_z^1\cap L_z^{2,1}}
			+
			\|r_2\|_{H_z^1\cap L_z^{2,1}}
			\right),
		\end{equation}
		where \(C\) depends on the small norm bound and on a fixed bound for
		\(r_1,r_2\) in \(H_z^1\cap L_z^{2,1}\).
		
		Combining \eqref{eq:qPT-positive-bounds},
		\eqref{eq:qPT-negative-L21-new},
		\eqref{eq:qPTx-positive-nu1-bound},
		\eqref{eq:qPTx-negative-nu1-bound}, and
		\eqref{eq:r-delta-controlled}, we have
		\begin{equation}\label{eq:qPT-line-bounds}
			\|q^{PT}\|_{H^{1,1}(\mathbb R)}
			\le
			C
			\left(
			\|r_1\|_{H_z^1\cap L_z^{2,1}}
			+
			\|r_2\|_{H_z^1\cap L_z^{2,1}}
			+
			\|z^{-1}r_1\|_{H_z^1\cap L_z^{2,1}}
			\right).
		\end{equation}
		In particular, since
		\((q^{PT})_x(x)=-q_x(-x)\), reflection preserves the
		\(L^{2,1}\)-norm and hence
		\begin{equation}\label{eq:qx-L21-from-nu1}
			\|q_x\|_{L^{2,1}(\mathbb R)}
			=
			\|(q^{PT})_x\|_{L^{2,1}(\mathbb R)}
			\le
			C
			\left(
			\|r_1\|_{H_z^1\cap L_z^{2,1}}
			+
			\|r_2\|_{H_z^1\cap L_z^{2,1}}
			\right).
		\end{equation}

		Similarly, the large-\(z\) reconstruction formula
		\eqref{eq:recon-qx-delta} gives
		\begin{align}
			\mathrm{e}^{\mathrm{i}m_+(x)-\frac{\mathrm{i}}{2}m_-(x)}
			\partial_x
			\left(
			q_x(x)\mathrm{e}^{\frac{\mathrm{i}}{2}m_-(x)}
			\right)
			&=
			-\frac{1}{\pi}
			\int_{\mathbb R}
			r_{\delta,2}(z)\mathrm{e}^{-2\mathrm{i}zx}
			N_{\delta,+,22}(x;z)\,dz                                  \nonumber\\
			&=:\widetilde I_3(x)+\widetilde I_4(x),
			\label{eq:qx-negative-split-new}
		\end{align}
		where
		\begin{align}
			\widetilde I_3(x)
			&:=
			-\frac{1}{\pi}
			\int_{\mathbb R}
			r_{\delta,2}(z)\mathrm{e}^{-2\mathrm{i}zx}\,dz,
			\label{eq:I3-negative-new}\\
			\widetilde I_4(x)
			&:=
			-\frac{1}{\pi}
			\int_{\mathbb R}
			r_{\delta,2}(z)\mathrm{e}^{-2\mathrm{i}zx}
			\left(N_{\delta,+,22}(x;z)-1\right)\,dz .
			\label{eq:I4-negative-new}
		\end{align}
	Through similar analysis as before, we can obtain that
	\begin{equation}\label{eq:qxxx-negative-bounds}
		\|q_{xxx} (x) \|_{L^2(\mathbb{R}^{-})} \le C
		\left( \|z^{-1} r_1 \|_{H^{1} \cap L^{2,1}} +\|r_2 \|_{H^{1} \cap L^{2,1}}
		\right),
	\end{equation}
	and
	\begin{equation}\label{eq:qxx-negative-bounds}
		\| \partial_x \left( q_x(x) \mathrm{e}^{\frac{\mathrm{i}}{2} m_{+}(x)}
		\right) \|_{L^{2,1}(\mathrm{R}^{-})} \le \mathcal{C} (\| r_2 \|_{H^{1}} + \| r_2 \|_{H^{1}}^2).
	\end{equation}
	Combined \eqref{eq:qxxx-negative-bounds}, \eqref{eq:qxx-negative-bounds} with \eqref{eq:qPT-line-bounds}, \eqref{eq:qxxx-positive-bounds} and \eqref{eq:qxx-positive-bounds}, we have
	\begin{equation}\label{eq:q-line-bounds}
		\| q \|_{H^{3}(\mathbb{R}) \cap H^{2,1}(\mathbb{R})}
		\le
		C \left( \|r_1 \|_{\mathcal{W}(\mathbb{R})} + \|r_2 \|_{\mathcal{W}(\mathbb{R})}
		\right)
	\end{equation}
	which completes the proof.
	\hfill
	\end{proof}

	From Proposition \ref{prop:whole-line-reconstruction-estimate}, we can get the following proposition:

	\begin{proposition}\label{prop:q-lipschaitz}
		Suppose that $r_{1,2} \in \mathcal{W}(\mathbb{R}) $ and satisfies the small norm inequality \eqref{eq:small-norm}, then mapping
		\begin{equation}\label{eq:map-r-q}
			\mathcal{W}(\mathbb{R}) \ni (r_1,r_2) \longmapsto q \in H^{3}(\mathbb{R}) \cap H^{2,1}(\mathbb{R}),
		\end{equation}
		is Lipschitz continuous.
	\end{proposition}
	\begin{proof}
		Suppose $r_{1,2}, \tilde{r}_{1,2} \in \mathcal{W}(\mathbb{R})$ satisfy the small norm inequality \eqref{eq:small-norm}, and
		\begin{equation}
			\|r_{1,2}\|_{\mathcal{W}(\mathbb{R})}, \|\tilde{r}_{1,2}\|_{\mathcal{W}(\mathbb{R})} \le \rho
		\end{equation}
		for some $\rho > 0$. Denote the corresponding potentials by $q$ and $\tilde{q}$ respectively.
		
		From the reconstruction formulas \eqref{eq:qPT-positive-split-new}, \eqref{eq:qx-positive-split-new} on the positive half-line, and \eqref{eq:qPT-negative-split-new}, \eqref{eq:qx-negative-split-new} on the negative half-line, we can represent the difference $q(x) - \tilde{q}(x)$ in terms of the differences of the reflection coefficients and the corresponding Beals-Coifman solutions.
		
		By utilizing the resolvent identity for the integral Fredholm equations, similar to the technique used for the Jost solutions in \eqref{eq:Lip diff}, the differences of the RH solutions $N(x;z) - \tilde{N}(x;z)$ and the transformed solutions $N_\delta(x;z) - \tilde{N}_\delta(x;z)$ can be bounded by the differences of their respective jump matrices, which in turn are linearly controlled by $r_{1,2} - \tilde{r}_{1,2}$.
		
		Repeating almost the exact same estimates as in the proof of Proposition \ref{prop:whole-line-reconstruction-estimate}, we deduce that there exists a positive constant $C$ dependent on $\rho$ such that
		\begin{equation}
			\begin{split}
				\|q - \tilde{q}\|_{H^3(\mathbb{R}) \cap H^{2,1}(\mathbb{R})} &\le C \left( \|z^{-1}(r_1 - \tilde{r}_1)\|_{H_z^1 \cap L_z^{2,1}} + \|r_1 - \tilde{r}_1\|_{H_z^1 \cap L_z^{2,1}} + \|r_2 - \tilde{r}_2\|_{H_z^1 \cap L_z^{2,1}} \right) \\
				&\le C \left( \|r_1 - \tilde{r}_1\|_{\mathcal{W}(\mathbb{R})} + \|r_2 - \tilde{r}_2\|_{\mathcal{W}(\mathbb{R})} \right).
			\end{split}
		\end{equation}
		Using this representation for the difference between $q$ and $\tilde{q}$, the uniform bound ensures the Lipschitz continuity of the mapping \eqref{eq:map-r-q}.
		\hfill
	\end{proof}
	
	\section{Time evolution and global reconstruction}\label{section-Time}
	
	In the previous sections the time variable was regarded as a fixed parameter.
	We now introduce the time evolution of the scattering data and then reconstruct
	the solution \(q(x,t)\) from the corresponding time-dependent RH problem.
	Throughout this section we write
	\begin{equation}\label{eq:time-phase}
		\Theta(x,t;\lambda)
		:=
		\lambda^2x-\frac{t}{4\lambda^2},
		\qquad z=\lambda^2 .
	\end{equation}
	Then the exponential factors in the jump matrices are obtained by replacing
	\(\lambda^2x\) with \(\Theta(x,t;\lambda)\).
	
	\subsection{Time evolution of the scattering data}
	
	From the Lax pair \eqref{eq:New Laxpair}, the scattering relation for the Jost functions is
	\begin{equation}\label{eq:time-scattering-relation-fullphase}
		\psi^{-}(x,t;\lambda)
		=
		\psi^{+}(x,t;\lambda)
		\mathrm{e}^{\mathrm{i}\Theta(x,t;\lambda)\hat{\sigma}_3}
		S(\lambda), \qquad S(\lambda)
		=
		\begin{pmatrix}
			a(\lambda) & \breve b(\lambda)\\
			b(\lambda) & \breve a(\lambda)
		\end{pmatrix},
		\qquad
		\lambda\in\Sigma_{\lambda},
	\end{equation}
		Equivalently, if one keeps the \(x\)-part scattering convention
	$$
	\psi^{-}(x,t;\lambda) = \psi^{+}(x,t;\lambda)
	\mathrm{e}^{\mathrm{i}\lambda^2x\hat{\sigma}_3}
	S(t;\lambda),
	$$
	then
	\begin{equation}\label{eq:S-time-evolution}
		S(t;\lambda)
		=
		\mathrm{e}^{-\frac{\mathrm{i}t}{4\lambda^2}\hat{\sigma}_3}
		S(0;\lambda).
	\end{equation}
	Hence
	\begin{align}
		a(t;\lambda)=a(\lambda),\quad 	\breve a(t;\lambda)=\breve a(\lambda), \quad
		b(t;\lambda)=		b(\lambda)		\mathrm{e}^{\frac{\mathrm{i}t}{2\lambda^2}},\quad
				\breve b(t;\lambda)=
		\breve b(\lambda)
		\mathrm{e}^{-\frac{\mathrm{i}t}{2\lambda^2}} .
		\label{eq:b-breveb-time}
	\end{align}
	Consequently, the reflection coefficients on the \(z\)-plane satisfy
	\begin{equation}\label{eq:r12-time}
		r_1(t;z)
		=
		r_1(z)\mathrm{e}^{-\frac{\mathrm{i}t}{2z}},
		\qquad
		r_2(t;z)
		=
		r_2(z)\mathrm{e}^{\frac{\mathrm{i}t}{2z}} .
	\end{equation}
	In particular,
	\begin{equation}\label{eq:r1r2-time-invariant}
		r_1(t;z)r_2(t;z)=r_1(z)r_2(z),
	\end{equation}
	so the scalar RH function \(\delta(z)\) defined in RH Problem
	\ref{RH:delta} is independent of \(t\).
	
	\begin{proposition}\label{prop:time-evolution-W}
		Let \(T>0\) be fixed. Suppose that
		$
		r_1,r_2\in\mathcal W(\mathbb R).
		$
		Then, for every \(t\in[-T,T]\),
		$
		r_1(t;\cdot),\ r_2(t;\cdot)\in \mathcal W(\mathbb R).
		$
		Moreover, there exists a constant \(C_T>0\) such that
		\begin{equation}\label{eq:time-evolution-W-bound}
			\|r_1(t;\cdot)\|_{\mathcal W}
			+
			\|r_2(t;\cdot)\|_{\mathcal W}
			\le
			C_T
			\left(
			\|r_1\|_{\mathcal W}
			+
			\|r_2\|_{\mathcal W}
			\right).
		\end{equation}
	\end{proposition}
	
	\begin{proof}
		Since the exponential factors in \eqref{eq:r12-time} have modulus one
		for \(z\in\mathbb R\), we have
		\begin{equation}
			\|r_j(t;\cdot)\|_{L_z^2}
			=
			\|r_j\|_{L_z^2},
			\quad
			\|r_j(t;\cdot)\|_{L_z^{2,1}}
			=
			\|r_j\|_{L_z^{2,1}},
			\quad
			\|r_j(t;\cdot)\|_{L_z^{2,-2}}
			=
			\|r_j\|_{L_z^{2,-2}},
			\quad
			j=1,2 .
		\end{equation}
		For the \(H^1\)-norm, differentiating \eqref{eq:r12-time} gives
		\begin{align}
			\partial_z r_1(t;z)
			&=
			\mathrm{e}^{-\frac{\mathrm{i}t}{2z}}
			\left(
			\partial_z r_1(z)
			+
			\frac{\mathrm{i}t}{2z^2}r_1(z)
			\right),
			\label{eq:dz-r1-time}\\
			\partial_z r_2(t;z)
			&=
			\mathrm{e}^{\frac{\mathrm{i}t}{2z}}
			\left(
			\partial_z r_2(z)
			-
			\frac{\mathrm{i}t}{2z^2}r_2(z)
			\right).
			\label{eq:dz-r2-time}
		\end{align}
		Thus, for \(t\in[-T,T]\),
		\begin{equation}\label{eq:poly-T-def}
			\|\partial_z r_j(t;\cdot)\|_{L_z^2}
			\le
			\|\partial_z r_j\|_{L_z^2}
			+
			\frac{T}{2}\|z^{-2}r_j\|_{L_z^2},
			\qquad j=1,2 .
		\end{equation}
		This proves \eqref{eq:time-evolution-W-bound}.
		
		\hfill
	\end{proof}
	
	\subsection{The time-dependent RH problems}
	
	Using \eqref{eq:b-breveb-time}, the RH problem on the \(\lambda\)-plane
	becomes the following one.
	
	\begin{RH}\label{RH:lambda-time}
		Find a matrix function \(\Phi(x,t;\lambda)\) with the following properties:
		\begin{enumerate}
			\item \textbf{Analyticity:}
			\(\Phi(x,t;\lambda)\) is analytic in
			\(\mathbb C\setminus\Sigma_{\lambda}\).
			
			\item \textbf{Jump condition:}
			\begin{equation}\label{eq:Phi-time-jump}
				\Phi_+(x,t;\lambda)
				=
				\Phi_-(x,t;\lambda)
				\left(I+V(x,t;\lambda)\right),
				\qquad \lambda\in\Sigma_{\lambda},
			\end{equation}
			where
			\begin{equation}\label{eq:V-time-lambda}
				V(x,t;\lambda)
				=
				\begin{pmatrix}
					0
					&
					\breve r(\lambda)
					\mathrm{e}^{2\mathrm{i}\Theta(x,t;\lambda)}
					\\[0.4em]
					-r(\lambda)
					\mathrm{e}^{-2\mathrm{i}\Theta(x,t;\lambda)}
					&
					-r(\lambda)\breve r(\lambda)
				\end{pmatrix}.
			\end{equation}
			
			\item \textbf{Normalization:}
			$	\Phi(x,t;\lambda)				=				I+\mathcal O(\lambda^{-1}), \,\, {\rm as}\,\, \lambda\to\infty.$
		\end{enumerate}
	\end{RH}
	
	In the \(z\)-plane, the time-dependent RH problem takes the form
	\begin{RH}\label{RH:z-time}
		Find \(N(x,t;z)\) such that:
		\begin{enumerate}
			\item \textbf{Analyticity:} \(N(x,t;z)\) is analytic in \(\mathbb C\setminus\mathbb R\).
			
			\item \textbf{Jump condition:} The boundary values satisfy
			\begin{equation}\label{eq:z-time-jump}
				N_+(x,t;z)
				=
				N_-(x,t;z)
				\left(I+R(x,t;z)\right),
				\qquad z\in\mathbb R,
			\end{equation}
			where
			\begin{equation}\label{eq:R-time-z}
				R(x,t;z)
				=
				\begin{pmatrix}
					0
					&
					r_1(z)
					\mathrm{e}^{2\mathrm{i}zx-\frac{\mathrm{i}t}{2z}}
					\\[0.4em]
					r_2(z)
					\mathrm{e}^{-2\mathrm{i}zx+\frac{\mathrm{i}t}{2z}}
					&
					r_1(z)r_2(z)
				\end{pmatrix}.
			\end{equation}
			
			\item \textbf{Normalization:} \(N(x,t;z)=I+\mathcal O(z^{-1})\) as \(z\to\infty\).
		\end{enumerate}
	\end{RH}
	
	Because \(r_1(t;z)r_2(t;z)=r_1(z)r_2(z)\), the function \(\delta(z)\)
	defined by \eqref{eq:delta-solution} is unchanged under the time evolution.
	Thus the transformed reflection coefficients are
	\begin{equation}\label{eq:rdelta-time}
		r_{\delta,1}(t;z)
		=
		r_{\delta,1}(z)
		\mathrm{e}^{-\frac{\mathrm{i}t}{2z}},
		\qquad
		r_{\delta,2}(t;z)
		=
		r_{\delta,2}(z)
		\mathrm{e}^{\frac{\mathrm{i}t}{2z}} .
	\end{equation}
	The time-dependent \(\delta\)-transformed RH problem is therefore
	$$
	N_{\delta,+}(x,t;z) = N_{\delta,-}(x,t;z)
	\left(I+R_{\delta}(x,t;z)\right),
	$$
	with
	\begin{equation}\label{eq:Rdelta-time}
		R_{\delta}(x,t;z)
		=
		\begin{pmatrix}
			r_{\delta,1}(z)r_{\delta,2}(z)
			&
			r_{\delta,1}(z)
			\mathrm{e}^{2\mathrm{i}zx-\frac{\mathrm{i}t}{2z}}
			\\[0.4em]
			r_{\delta,2}(z)
			\mathrm{e}^{-2\mathrm{i}zx+\frac{\mathrm{i}t}{2z}}
			&
			0
		\end{pmatrix}.
	\end{equation}
	
	The reconstruction formulae now become
	\begin{equation}\label{eq:recon-qx-time}
		\mathrm{e}^{\mathrm{i}m_+(x,t)-\frac{\mathrm{i}}{2}m_-(x,t)}
		\partial_x
		\left(
		q_x(x,t)
		\mathrm{e}^{\frac{\mathrm{i}}{2}m_-(x,t)}
		\right)
		=
		-\frac{1}{\pi}
		\int_{\mathbb R}
		r_{\delta,2}(z)
		\mathrm{e}^{-2\mathrm{i}zx+\frac{\mathrm{i}t}{2z}}
		N_{\delta,+,22}(x,t;z)\,dz ,
	\end{equation}
	and
	\begin{equation}\label{eq:recon-qPT-time}
		q^{PT}(x,t)
		\mathrm{e}^{-\frac{\mathrm{i}}{2}m_+(x,t)}
		=
		\frac{\delta_+^{-1}(0)}{\pi}
		\int_{\mathbb R}
		z^{-1}r_{\delta,1}(z)
		\mathrm{e}^{2\mathrm{i}zx-\frac{\mathrm{i}t}{2z}}
		N_{\delta,-,11}(x,t;z)\,dz .
	\end{equation}
	
	\subsection{Verification of the reconstructed solution}
	
	In this subsection, we verify that the potential reconstructed from the
	time-dependent RH problems satisfies the nonlocal Fokas--Lenells equation.
	Throughout the subsection we set
	\bee
		Y(x,t;\lambda):=E(x,t)\Phi(x,t;\lambda), \qquad E(x,t):=\mathrm{e}^{\frac{\mathrm{i}}{2}m_+(x,t)\sigma_3}.
	\ene
	From the definition of \(\Phi\), the small-\(\lambda\) expansion of the
	Jost functions in \eqref{eq:psi-small-lambda}, we have
	\begin{equation}\label{eq:Y-small-lambda-expansion}
		Y(x,t;\lambda)
		=
		I+\mathrm{i}\lambda Q(x,t)+\lambda^2D(x,t)+\mathcal O(\lambda^3),
		\qquad \lambda\to0,
	\end{equation}
	where \(D(x,t)\) is diagonal. In particular,
	\begin{align}
		Y_xY^{-1}
		&=
		\mathrm{i}\lambda Q_x+\mathcal O(\lambda^2),
		\label{eq:Yx-small}\\
		Y_tY^{-1}
		&=
		\mathcal O(\lambda),
		\label{eq:Yt-small}\\
		Y\sigma_3Y^{-1}
		&=
		\sigma_3+2\mathrm{i}\lambda Q\sigma_3
		-2\lambda^2Q^2\sigma_3+\mathcal O(\lambda^3).
		\label{eq:Ysigma-small}
	\end{align}
	
	\begin{proposition}\label{prop:verify-x-part}
		Let \(\Phi(x,t;\lambda)\) be the solution of the time-dependent RH
		Problem \ref{RH:lambda-time}. Define
		\begin{equation}\label{eq:eta-def}
			\eta(x,t;\lambda)
			:=
			E(x;t)
			\Phi(x,t;\lambda)
			\mathrm{e}^{\mathrm{i}\Theta(x,t;\lambda)\sigma_3}.
		\end{equation}
		Then \(\eta(x,t;\lambda)\) satisfies the \(x\)-part of the Lax pair
		\eqref{eq:Lax pair}, namely
		\begin{equation}\label{eq:eta-x-equation}
			\eta_x
			=
			\left(
			\mathrm{i}\lambda^2\sigma_3
			+\mathrm{i}\lambda Q_x
			\right)\eta .
		\end{equation}
	\end{proposition}
	
	\begin{proof}
		Multiplying the jump relation \eqref{eq:Phi-time-jump} by
		\(\mathrm{e}^{\mathrm{i}\Theta\sigma_3}\) from the right gives
		$$
		\Phi_+\mathrm{e}^{\mathrm{i}\Theta\sigma_3} =
		\Phi_-\mathrm{e}^{\mathrm{i}\Theta\sigma_3}J_0(\lambda),
		$$
		where
		\begin{equation}\label{eq:J0-jump}
			J_0(\lambda)
			=
			\begin{pmatrix}
				1 & \breve r(\lambda)\\
				-r(\lambda) & 1-r(\lambda)\breve r(\lambda)
			\end{pmatrix}
		\end{equation}
		is independent of \(x\) and \(t\). Hence \(\eta_+=\eta_-J_0\), and
		\(\eta_x\eta^{-1}\) has no jump across \(\Sigma_\lambda\). Therefore
		the possible singularities of \(\eta_x\eta^{-1}\) can only occur at
		\(\lambda=0\) and \(\lambda=\infty\).
		
		We first analyze the behavior at infinity. From
		$$
		\Phi(x,t;\lambda)
		=
		I+\lambda^{-1}\Phi^{(1)}(x,t)+\mathcal O(\lambda^{-2}),
		\qquad \lambda\to\infty,
		$$
		and \(\Theta_x=\lambda^2\) [cf. (\ref{eq:time-phase})], we obtain
		\begin{equation}\label{eq:eta_x_eta_inv}
			\eta_x\eta^{-1}
			=
			\frac{\mathrm{i}}{2}m_{+,x}\sigma_3
			+
			E\left(
			\Phi_x\Phi^{-1}
			+
			\mathrm{i}\lambda^2\Phi\sigma_3\Phi^{-1}
			\right)E^{-1}.
		\end{equation}
		Thus
		\begin{equation}\label{eq:eta_x_inf}
			\eta_x\eta^{-1}
			=
			\mathrm{i}\lambda^2\sigma_3
			+
			\mathrm{i}\lambda
			E[\Phi^{(1)},\sigma_3]E^{-1}
			+
			\mathcal O(1),
			\qquad \lambda\to\infty .
		\end{equation}
		The large-\(\lambda\) reconstruction formula gives
		\begin{equation}\label{eq:large-lambda-Qx-reconstruction}
			E[\Phi^{(1)},\sigma_3]E^{-1}=Q_x .
		\end{equation}
		Consequently,
		\begin{equation}\label{eq:eta_x_inf_refined}
			\eta_x\eta^{-1}
			-
			\mathrm{i}\lambda^2\sigma_3
			-
			\mathrm{i}\lambda Q_x
			=
			\mathcal O(1),
			\qquad \lambda\to\infty .
		\end{equation}
		
		At the origin, using \(\eta=Y\mathrm{e}^{\mathrm{i}\Theta\sigma_3}\),
		\(\Theta_x=\lambda^2\), and \eqref{eq:Yx-small}--\eqref{eq:Ysigma-small}, we have
		\begin{equation}\label{eq:eta_x_zero}
			\eta_x\eta^{-1} = Y_xY^{-1} +
			\mathrm{i}\lambda^2Y\sigma_3Y^{-1}
			= \mathrm{i}\lambda Q_x + \mathrm{i}\lambda^2\sigma_3 + \mathcal O(\lambda^2),
			\qquad \lambda\to 0 .
		\end{equation}
		Define
		$$
		F_x(\lambda) := \eta_x\eta^{-1} - \mathrm{i}\lambda^2\sigma_3 - \mathrm{i}\lambda Q_x .
		$$
		By the preceding discussion, \(F_x\) has no jump, is removable at
		\(\lambda=0\), and is bounded at infinity. Liouville's theorem implies
		that \(F_x\) is constant. Since \eqref{eq:eta_x_zero} gives
		\(F_x(\lambda)=\mathcal O(\lambda^2)\) as \(\lambda\to0\), the constant
		is zero. This proves \eqref{eq:eta-x-equation}. \hfill
	\end{proof}
	
	\begin{proposition}\label{prop:verify-t-part}
		The function \(\eta(x,t;\lambda)\) defined by \eqref{eq:eta-def}
		satisfies the \(t\)-part of the Lax pair \eqref{eq:Lax pair}, namely
		\begin{equation}\label{eq:eta-t-equation}
			\eta_t
			=
			\left(
			-\frac{\mathrm{i}}{4\lambda^2}\sigma_3
			+
			\frac{1}{2\lambda}Q\sigma_3
			+
			\frac{\mathrm{i}}{2}Q^2\sigma_3
			\right)\eta .
		\end{equation}
	\end{proposition}
	
	\begin{proof}
		The same conjugation of the jump relation shows that
		\(\eta_t\eta^{-1}\) has no jump across \(\Sigma_\lambda\). Its possible
		singularities are again only at \(\lambda=0\) and \(\lambda=\infty\).
		Moreover,
		\begin{equation}\label{eq:eta_t_eta_inv}
			\eta_t\eta^{-1} = \frac{\mathrm{i}}{2}m_{+,t}\sigma_3 + E\left(
			\Phi_t\Phi^{-1} +
			\mathrm{i}\Theta_t\Phi\sigma_3\Phi^{-1}
			\right)E^{-1}.
		\end{equation}
		As \(\lambda\to\infty\), we have \(\Theta_t=-1/(4\lambda^2)\) and
		\(\Phi=I+\mathcal O(\lambda^{-1})\). Hence
		\begin{equation}\label{eq:eta_t_inf}
			\eta_t\eta^{-1}=\mathcal O(1),
			\qquad \lambda\to\infty .
		\end{equation}
		
		Near \(\lambda=0\), we write
		$$
		\eta_t\eta^{-1} = Y_tY^{-1} + \mathrm{i}\Theta_tY\sigma_3Y^{-1},
		\qquad
		\Theta_t=-\frac{1}{4\lambda^2}.
		$$
		Using \eqref{eq:Yt-small} and \eqref{eq:Ysigma-small}, we get
		\begin{equation}\label{eq:eta_t_zero}
			\eta_t\eta^{-1}
			=
			-\frac{\mathrm{i}}{4\lambda^2}\sigma_3
			+
			\frac{1}{2\lambda}Q\sigma_3
			+
			\frac{\mathrm{i}}{2}Q^2\sigma_3
			+
			\mathcal O(\lambda),
			\qquad \lambda\to0 .
		\end{equation}
		Define
		$$
		F_t(\lambda) := \eta_t\eta^{-1} + \frac{\mathrm{i}}{4\lambda^2}\sigma_3 - \frac{1}{2\lambda}Q\sigma_3 - \frac{\mathrm{i}}{2}Q^2\sigma_3 .
		$$
		Then \(F_t\) has no jump, the singularity at \(\lambda=0\) is
		removable by \eqref{eq:eta_t_zero}, and \(F_t=\mathcal O(1)\) at
		infinity by \eqref{eq:eta_t_inf}. Hence \(F_t\) is constant by
		Liouville's theorem. Since \eqref{eq:eta_t_zero} gives
		\(F_t(\lambda)=\mathcal O(\lambda)\) as \(\lambda\to0\), this constant
		is zero. Therefore \eqref{eq:eta-t-equation} holds. \hfill
	\end{proof}
	
	\begin{corollary}\label{cor:verify-nFL}
		The potential \(q(x,t)\) reconstructed from the time-dependent RH problem
		satisfies the nonlocal Fokas--Lenells equation (\ref{eq:nFL})
	\end{corollary}
	
	\begin{proof}
		By Propositions \ref{prop:verify-x-part} and \ref{prop:verify-t-part},
		\(\eta\) satisfies
		$$
		\eta_x=X\eta,
		\qquad
		\eta_t=T\eta,
		$$
		where
		$$
		X=\mathrm{i}\lambda^2\sigma_3+\mathrm{i}\lambda Q_x,
		\qquad
		T= -\frac{\mathrm{i}}{4\lambda^2}\sigma_3 + \frac{1}{2\lambda}Q\sigma_3 + \frac{\mathrm{i}}{2}Q^2\sigma_3 .
		$$
		The compatibility condition \(\eta_{xt}=\eta_{tx}\) is equivalent to
		\begin{equation}\label{eq:zero-curvature-verified}
			X_t-T_x+[X,T]=0 .
		\end{equation}
		Using \(Q\sigma_3=-\sigma_3Q\), \(Q_x\sigma_3=-\sigma_3Q_x\), and
		\(Q^2=q q^{PT}I\), a direct computation gives
		\bee \label{zero2}
		X_t-T_x+[X,T] =\mathrm{i}\lambda Q_{xt}
		- \mathrm{i}\lambda Q - \lambda Q^2Q_x\sigma_3 .
		\ene
		Therefore, \eqref{eq:zero-curvature-verified} with (\ref{zero2}) is equivalent to
		\begin{equation}\label{eq:matrix-nFL-verified}
			Q_{xt}-Q+\mathrm{i}Q^2Q_x\sigma_3=0 .
		\end{equation}
		Since
		$$
		Q=
		\begin{pmatrix}
			0&q^{PT}\\
			q&0
		\end{pmatrix},
		\qquad
		Q^2Q_x\sigma_3
		= q q^{PT}
		\begin{pmatrix}
			0&-(q^{PT})_x\\
			q_x&0
		\end{pmatrix},
		$$
		the \((2,1)\)-entry of \eqref{eq:matrix-nFL-verified} gives exactly
		$$
		q_{xt}-q+\mathrm{i}q q^{PT}q_x=0 .
		$$
		The \((1,2)\)-entry gives the corresponding equation for
		\(q^{PT}\). Thus the reconstructed potential is a solution of
		\eqref{eq:nFL}. \hfill
	\end{proof}
	
	\section{Existence of global solutions to the nFL equation}\label{section-proof}
	
	\begin{proposition}\label{prop:local-solution}
		Let \(q_0\in H^3(\mathbb R)\cap H^{2,1}(\mathbb R)\) satisfy the
		small-norm condition \eqref{eq:small-norm}. Then, for every \(T>0\), there exists a unique
		solution
		\begin{equation}\label{eq:local-solution-space}
			q\in C\left([-T,T], H^3(\mathbb R)\cap H^{2,1}(\mathbb R)\right)
		\end{equation}
		to the Cauchy problem of the nonlocal Fokas--Lenells equation. Moreover, on bounded subsets of the initial data space, the solution map
		\begin{equation}\label{eq:solution-map-local}
			q_0
			\longmapsto
			q(\cdot,t)
			\in
			C\left([-T,T], H^3(\mathbb R)\cap H^{2,1}(\mathbb R)\right)
		\end{equation}
		is Lipschitz continuous.
	\end{proposition}
	
	\begin{proof}
		Starting from \(q_0\), Proposition \ref{prop:r12-Lipschitz} gives
		reflection coefficients:	$		r_1(0;\cdot),r_2(0;\cdot)\in\mathcal W(\mathbb R).$
		The time evolution is defined by \eqref{eq:r12-time}. By Proposition
		\ref{prop:time-evolution-W}, for every \(t\in[-T,T]\),
		$		r_1(t;\cdot),r_2(t;\cdot)\in\mathcal W(\mathbb R).$
		Moreover, the product \(r_1(t;z)r_2(t;z)\) is independent of \(t\), so
		the scalar function \(\delta(z)\) and the coercivity estimates used in
		the solvability of the RH problem remain valid for all \(t\in[-T,T]\).
		
		For each fixed \(t\), solving the time-dependent RH problem
		\ref{RH:z-time} and applying the reconstruction formulae
		\eqref{eq:recon-qx-time}--\eqref{eq:recon-qPT-time} gives a potential
		\(q(x,t)\). Proposition \ref{prop:whole-line-reconstruction-estimate}
		then yields
		\begin{align}
			\|q(\cdot,t)\|_{H^3\cap H^{2,1}}
			&\le
			C
			\left(
			\|r_1(t;\cdot)\|_{\mathcal W}
			+
			\|r_2(t;\cdot)\|_{\mathcal W}
			\right) \nonumber\\
			&\le
			C_T
			\left(
			\|r_1(0;\cdot)\|_{\mathcal W}
			+
			\|r_2(0;\cdot)\|_{\mathcal W}
			\right).
			\label{eq:qt-priori}
		\end{align}
		By the direct scattering estimates Proposition \ref{prop:r12-Lipschitz}, the right-hand side is controlled by
		the \(H^3\cap H^{2,1}\)-norm of \(q_0\). Thus, we have the a priori estimate
		\begin{equation}\label{eq:prior-estimate}
			\|q(\cdot,t)\|_{H^3\cap H^{2,1}} \le C(T)\|q_0\|_{H^3\cap H^{2,1}}.
		\end{equation}
		
		We next prove the continuity of \(q\) with respect to \(t\). By
		Proposition \ref{prop:r12-Lipschitz}, we have
		\begin{equation}\label{eq:rj-W-zminus2-property}
			r_j\in \mathcal W(\mathbb R),
			\qquad j=1,2 .
		\end{equation}
		For convenience, write
		\begin{equation}\label{eq:theta-j-definition}
			\Theta_j(t,z)
			:=
			\mathrm{e}^{\frac{\theta_j\mathrm{i}t}{2z}},
			\qquad
			\theta_1=-1,\quad \theta_2=1 .
		\end{equation}
		Then the time evolution \eqref{eq:r12-time} can be written as
		\begin{equation}\label{eq:rj-time-theta-form}
			r_j(t;z)
			=
			r_j(z)\Theta_j(t,z),
			\qquad j=1,2 .
		\end{equation}
		For \(t_1,t_2\in[0,T]\), set
		\begin{equation}\label{eq:Delta-j-definition}
			\Delta_j(z)
			:=
			r_j(t_2;z)-r_j(t_1;z)
			=
			r_j(z)
			\left[
			\Theta_j(t_2,z)-\Theta_j(t_1,z)
			\right].
		\end{equation}
		Since \(|\Theta_j(t,z)|=1\), the weighted \(L^2\)-parts follow from
		the dominated convergence theorem:
		\begin{equation}\label{eq:Delta-j-weighted-continuity}
			\|\Delta_j\|_{L_z^{2,1}}
			+
			\|\Delta_j\|_{L_z^{2,-2}}
			\longrightarrow 0,
			\qquad
			t_2\to t_1,
			\qquad j=1,2 .
		\end{equation}
		It remains to estimate the derivative part. Direct differentiation gives
		\begin{equation}\label{eq:dz-Delta-j-formula}
			\begin{aligned}
				\partial_z\Delta_j(z)
				&=
				\partial_z r_j(z)
				\left[
				\Theta_j(t_2,z)-\Theta_j(t_1,z)
				\right]
				-
				\frac{\theta_j\mathrm{i}}{2z^2}
				r_j(z)
				\left[
				t_2\Theta_j(t_2,z)
				-
				t_1\Theta_j(t_1,z)
				\right].
			\end{aligned}
		\end{equation}
		Let \(\varepsilon>0\), by the absolute continuity of the \(L^2\)-integral,
		there exists \(\delta>0\) such that
		\begin{equation}\label{eq:low-frequency-choice-delta}
			\sum_{j=1}^2
			\left(
			2\|\partial_z r_j\|_{L_z^2(|z|<\delta)}
			+
			T\|z^{-2}r_j\|_{L_z^2(|z|<\delta)}
			\right)
			<
			\varepsilon .
		\end{equation}
		Using \eqref{eq:dz-Delta-j-formula} and \(|\Theta_j(t,z)|=1\), we obtain
		\begin{equation}\label{eq:low-frequency-dz-Delta}
			\sum_{j=1}^2
			\|\partial_z\Delta_j\|_{L_z^2(|z|<\delta)}
			<
			\varepsilon .
		\end{equation}
		On the region \(|z|\ge \delta\), the phase factors are smooth in \(t\).
		Moreover,
		\begin{equation}\label{eq:phase-time-difference-away-zero}
			\left|
			\Theta_j(t_2,z)-\Theta_j(t_1,z)
			\right|
			\le
			C_\delta |t_2-t_1|,
			\qquad |z|\ge\delta ,
		\end{equation}
		and
		\begin{equation}\label{eq:tphase-time-difference-away-zero}
			\left|
			t_2\Theta_j(t_2,z)
			-
			t_1\Theta_j(t_1,z)
			\right|
			\le
			C_{\delta,T}|t_2-t_1|,
			\qquad |z|\ge\delta .
		\end{equation}
		Therefore, by \eqref{eq:dz-Delta-j-formula},
		\begin{equation}\label{eq:high-frequency-dz-Delta}
			\sum_{j=1}^2
			\|\partial_z\Delta_j\|_{L_z^2(|z|\ge\delta)}
			\le
			C_{\delta,T}|t_2-t_1|
			\sum_{j=1}^2
			\left(
			\|\partial_z r_j\|_{L_z^2}
			+
			\|z^{-2}r_j\|_{L_z^2}
			\right).
		\end{equation}
		Combining \eqref{eq:low-frequency-dz-Delta} and
		\eqref{eq:high-frequency-dz-Delta}, and then letting \(t_2\to t_1\), we get
		\begin{equation}\label{eq:Delta-j-H1-continuity}
			\sum_{j=1}^2
			\|\partial_z\Delta_j\|_{L_z^2}
			\longrightarrow 0,
			\qquad
			t_2\to t_1 .
		\end{equation}
		Together with \eqref{eq:Delta-j-weighted-continuity}, this proves
		\begin{equation}\label{eq:rj-time-continuity-W}
			\sum_{j=1}^2
			\|r_j(t_2;\cdot)-r_j(t_1;\cdot)\|_{\mathcal W}
			\longrightarrow 0,
			\qquad
			t_2\to t_1 .
		\end{equation}
		Finally, using the Lipschitz continuity of the inverse scattering map in
		Proposition \ref{prop:q-lipschaitz}, we obtain
		\begin{equation}\label{eq:q-time-continuity-from-r}
			\begin{aligned}
				\|q(\cdot,t_2)-q(\cdot,t_1)\|_{H^3\cap H^{2,1}}
				&\le
				C_\rho
				\sum_{j=1}^2
				\|r_j(t_2;\cdot)-r_j(t_1;\cdot)\|_{\mathcal W} \longrightarrow 0,
				\qquad t_2\to t_1 .
			\end{aligned}
		\end{equation}
		Hence
		\begin{equation}\label{eq:q-time-continuity}
			q\in C\left([-T,T],H^3(\mathbb R)\cap H^{2,1}(\mathbb R)\right).
		\end{equation}
		
		Let \(q_0,\widetilde q_0\in H^{3}(\mathbb{R}) \cap H^{2,1}(\mathbb{R}) \) satisfy $\|q_0\|_{H^{3}(\mathbb{R}) \cap H^{2,1}(\mathbb{R})} , \| \widetilde q_0 \|_{H^{3}(\mathbb{R}) \cap H^{2,1}(\mathbb{R})} \le U $ for some $U > 0$. From \eqref{eq:qt-priori} and Lipschitz continuity in Proposition \ref{prop:r12-Lipschitz} and Proposition \ref{prop:q-lipschaitz}, we have
		\begin{equation}\label{eq:q-T-lipschitz}
			\begin{split}
			    \| q - \tilde{q} \|_{C\left([-T,T],H^3(\mathbb R)\cap H^{2,1}(\mathbb R)\right)} &= \| q(\cdot; t^*) - \tilde{q}(\cdot; t^*) \|_{H^3(\mathbb R)\cap H^{2,1}(\mathbb R)}  \\
			    &\le C_1(U) \left( \|r_1(\cdot;t^*) - \widetilde{r_1}(\cdot; t^*) \|_{\mathcal{W}} + \|r_2(\cdot;t^*) - \widetilde{r_2}(\cdot;t^*) \|_{\mathcal{W}}
			    \right)  \\
			    &\le C_2(U,T) \left(\|r_1(\cdot;0) - \widetilde{r_1}(\cdot; 0) \|_{\mathcal{W}} + \|r_2(\cdot;0) - \widetilde{r_2}(\cdot;0) \|_{\mathcal{W}}
			    \right)  \\
			    &\le C(U,T)\|q_0 - \tilde{q}_0 \|_{H^{3}(\mathbb{R})\cap H^{2,1}(\mathbb{R})},
		    \end{split}
		\end{equation}
		where $t^* \in [-T,T]$ and $C(U,T)$ is a polynomial function with respect to $T$ from \eqref{eq:poly-T-def}.
		This rigorously proves the mapping
		$$
		H^3(\mathbb R)\cap H^{2,1}(\mathbb R) \ni q_0 \longmapsto q \in C\left([-T,T], H^3(\mathbb R)\cap H^{2,1}(\mathbb R)\right)
		$$
		is Lipschitz continuous.
		
		Finally, Propositions \ref{prop:verify-x-part} and \ref{prop:verify-t-part}
		show that the reconstructed function satisfies the Lax pair, hence by
		Corollary \ref{cor:verify-nFL} it solves the nonlocal Fokas--Lenells
		equation. Uniqueness follows from the uniqueness of the RH problem and the
		Lipschitz continuity of both the direct and inverse scattering maps.
		\hfill
	\end{proof}
	
	\begin{theorem}\label{thm:global-solution}
		Let \(q_0\in H^3(\mathbb R)\cap H^{2,1}(\mathbb R)\) satisfy the
		small-norm condition \eqref{eq:small-norm}.
		Then the Cauchy problem of the nonlocal Fokas--Lenells equation admits a
		unique global solution
		\begin{equation}\label{eq:global-solution-space}
			q\in
			C\left(\mathbb{R},H^3(\mathbb R)\cap H^{2,1}(\mathbb R)\right).
		\end{equation}
		Furthermore, the solution map
		\begin{equation}\label{eq:q0-q-lipschitz}
			H^{3}(\mathbb{R}) \cap H^{2,1}(\mathbb{R}) \ni q_0(x) \longmapsto q(x;t) \in C
			\left( \mathbb{R}, H^{3}(\mathbb{R}) \cap H^{2,1}(\mathbb{R})
			\right)
		\end{equation}
		is Lipschitz continuous.
	\end{theorem}
	
	\begin{proof}
		The idea of the proof follows from the argument of \cite{Li2023}. Based on Proposition \ref{prop:local-solution}, as \(C(T)\) depends on \(T\) and \(\|r_{1,2}(0;z)\|_{\mathcal W(\mathbb R)}\) grows at most in a polynomial order with respect to \(\|q_0\|_{H^3(\mathbb R)\cap H^{2,1}(\mathbb R)}\), we hence reach the conclusion that the local solution exists in the space \(C\left([0,T], H^3(\mathbb R)\cap H^{2,1}(\mathbb R)\right)\) for an arbitrary fixed \(T>0\). Then the global existence of the solution can be directly asserted from the a priori estimation in \eqref{eq:prior-estimate}.
		
	    It remains to prove the Lipschitz continuity of the solution map.
     	Let \(q_0,\widetilde q_0\in H^{3}(\mathbb{R}) \cap H^{2,1}(\mathbb{R}) \) satisfy
      \bee
      \|q_0\|_{H^{3}(\mathbb{R}) \cap H^{2,1}(\mathbb{R})} , \| \widetilde q_0 \|_{H^{3}(\mathbb{R}) \cap H^{2,1}(\mathbb{R})} \le U \ene for some $U > 0$.

      Define
	    \begin{equation}\label{eq:q-global-metric}
		   d(q,\widetilde q)
		   :=
		   \sum_{n=1}^{\infty}
		   \frac{\|q-\widetilde q\|_n}
		   {2^n\left(1+\|q-\widetilde q\|_n\right)},
		   \qquad
		   \|q-\widetilde q\|_n
		   :=
		   \|q-\widetilde q\|_{C([-n,n],H^3(\mathbb R)\cap H^{2,1}(\mathbb R))} .
	    \end{equation}
	    Then, by \eqref{eq:q-T-lipschitz}, we obtain
	    \begin{equation}\label{eq:q-global-lipschitz-metric}
		    \begin{aligned}
			    d(q,\widetilde q)
			    &\le
			    \sum_{n=1}^{\infty}
			    \frac{C(U,n)\|q_0-\widetilde q_0\|_{H^3(\mathbb R)\cap H^{2,1}(\mathbb R)}}{2^n \left(
				1+C(U,n)\|q_0-\widetilde q_0\|_{H^3(\mathbb R)\cap H^{2,1}(\mathbb R)}
				\right)}  \\
			    &\le
			   \sum_{n=1}^{\infty} \frac{C(U,n)}{2^n}\|q_0-\widetilde q_0\|_{H^3(\mathbb R)\cap H^{2,1}(\mathbb R)} .
		    \end{aligned}
	    \end{equation}
	    Since \(C(U,n)\) grows at most polynomially in \(n\), the series
	    $$
	    \sum_{n=1}^{\infty}\frac{C(U,n)}{2^n}
	    $$
	    is convergent. Hence there exists a constant \(C(U)>0\) such that
	    \begin{equation}\label{eq:q-global-lipschitz-final}
		    d(q,\widetilde q) \le C(U)\|q_0-\widetilde q_0\|_{H^3(\mathbb R)\cap H^{2,1}(\mathbb R)} .
	    \end{equation}
	    Thus, the proof of Theorem \ref{thm:global-solution} is achieved.
	    \hfill
	\end{proof}

\addcontentsline{toc}{section}{Acknowledgements}

\v\noindent{\large\bf Acknowledgements}

This work was partially supported by the National Natural Science Foundation of China (No. 12471242), and Beijing Natural Science Foundation (No. 1262023).

\v\noindent{\large\bf Data Availability Statements}
The data that supports the findings of this study are available within the
article.

\v\noindent{\large\bf Conflict of Interest}
The authors have no conflicts to disclose.

\addcontentsline{toc}{section}{References}


\end{document}